\documentclass[11pt]{article}

\usepackage[a4paper,total={7.0in,10.0in}]{geometry}

\usepackage[utf8]{inputenc}
\usepackage[T1]{fontenc}
\usepackage[american]{babel}

\usepackage{setspace}
\usepackage[bottom]{footmisc}

\usepackage{mathpazo}
\usepackage{microtype}

\usepackage{amsmath,amsthm,amssymb}
\usepackage{bm}
\usepackage{bbm}

\allowdisplaybreaks

\newcommand*\diff{\mathop{}\!\mathrm{d}}

\DeclareMathOperator*{\argmin}{arg\,min}

\usepackage{graphicx}
\usepackage{float}
\usepackage{booktabs}
\usepackage{tabularx}
\usepackage{array}
\usepackage{multirow}
\usepackage{subcaption}
\usepackage{caption}

\usepackage{enumitem}
\usepackage[toc,page]{appendix}
\usepackage[gen]{eurosym}
\usepackage{titling}
\usepackage{authblk}

\setlist[itemize]{topsep=0.25em,itemsep=0.1em,parsep=0em}
\setlist[enumerate]{topsep=0.25em,itemsep=0.1em,parsep=0em}

\usepackage[natbibapa]{apacite}
\usepackage{xcolor}

\definecolor{linkblue}{RGB}{0,45,110}

\usepackage[
colorlinks=true,
linkcolor=linkblue,
citecolor=linkblue,
urlcolor=linkblue
]{hyperref}

\newcommand*{\secref}[1]{\hyperref[{#1}]{Section~\ref*{#1}}}
\newcommand*{\tabref}[1]{\hyperref[{#1}]{Table~\ref*{#1}}}
\newcommand*{\figref}[1]{\hyperref[{#1}]{Figure~\ref*{#1}}}
\newcommand*{\appref}[1]{\hyperref[{#1}]{Appendix~\ref*{#1}}}
\newcommand*{\eqrefnew}[1]{\hyperref[{#1}]{Equation~(\ref*{#1})}}

\usepackage{soul}
\usepackage[normalem]{ulem}
\setstcolor{red}
\usepackage{comment}

\newcommand{\calA}{\mathcal{A}}
\newcommand{\calB}{\mathcal{B}}
\newcommand{\calC}{\mathcal{C}}
\newcommand{\calD}{\mathcal{D}}
\newcommand{\calF}{\mathcal{F}}

\newcommand{\calN}{\mathcal{N}}

\newcommand{\calV}{\mathcal{V}}

\newcommand{\vdelta}{\bm{\delta}}

\newtheoremstyle{myplain}
{5pt}
{5pt}
{\itshape}
{}
{\bfseries}
{.}
{.5em}
{}

\newtheoremstyle{mydefinition}
{5pt}
{5pt}
{\normalfont}
{}
{\bfseries}
{.}
{.5em}
{}

\theoremstyle{myplain}

\newtheorem{theorem}{Theorem}
\newtheorem{lemma}[theorem]{Lemma}

\newtheorem{corollary}[theorem]{Corollary}

\theoremstyle{mydefinition}

\newtheorem{assumption}{Assumption}

\newtheorem{example}{Example}

\newtheorem*{remark}{Remark}

	\newcommand{\appendixnumbering}{%
		\setcounter{equation}{0}%
		\setcounter{theorem}{0}%
		\setcounter{assumption}{0}%
		\setcounter{condition}{0}%
		\setcounter{definition}{0}%
		\setcounter{example}{0}%
		\setcounter{table}{0}%
		\setcounter{figure}{0}%
		\numberwithin{equation}{section}%
		\numberwithin{theorem}{section}%
		\numberwithin{assumption}{section}%
		\numberwithin{condition}{section}%
		\numberwithin{definition}{section}%
		\numberwithin{example}{section}%
		\numberwithin{table}{section}%
		\numberwithin{figure}{section}%
	}
	
	\makeatletter
	\renewcommand\paragraph{\@startsection{paragraph}{4}{\z@}%
		{3.25ex \@plus1ex \@minus.2ex}%
		{-1em}%
		{\normalfont\itshape}}
	\makeatother 
	
\begin{document}

			
			\title{From dense grids to valid inference: Accounting for regularization bias in nonparametric random coefficient models}
			
			\renewcommand\Affilfont{\normalfont\footnotesize}
			\author{
				Lingwei Kong\textsuperscript{a},
				Maximilian Osterhaus\textsuperscript{a},
				and Michael Pen\textsuperscript{a}
			}

			\date{July 2026}
			
			\maketitle

			{\setstretch{.8}

				
				\vspace{1cm}
				
\begin{abstract}
\noindent
This paper develops an inference procedure for average functionals of random-coefficient distributions, such as mean willingness-to-pay and average elasticities, when the distribution is estimated nonparametrically using the penalized fixed-grid estimator of \citet{heiss2022}. We establish asymptotic normality of the corresponding penalized plug-in estimator centered at the functional evaluated at the penalized pseudo-true value and propose a confidence interval that accounts for the regularization bias. Our method applies to a broad class of linear and nonlinear functionals and allows researchers to use dense grids to reduce approximation bias while maintaining valid inference. Monte Carlo simulations show that the proposed intervals achieve coverage close to the nominal level while remaining informative in finite samples. An empirical application to travel mode demand illustrates that flexible nonparametric specifications can yield economically meaningful differences relative to standard parametric models.

\vspace{1cm}

\end{abstract}

}

\vfill
\hrule
\vspace{5mm}
\noindent {\footnotesize $^a$ University of Groningen, Nettelbosje 2, 9747 AE Groningen, The Netherlands.}\\
\noindent {\footnotesize Correspondence to: \href{mailto:m.d.pen@rug.nl}{m.d.pen@rug.nl}. (M.D. Pen).}	

\onehalfspacing


\newpage

\section{Introduction}

Random-coefficients (RC) models are widely used in economics to accommodate heterogeneous preferences across economic agents. In these models, individual-level parameters vary according to an unknown distribution that researchers aim to recover from the data. The objects of central interest in applications are typically functionals of the RC distribution, such as willingness-to-pay measures, demand elasticities, or predicted outcomes under counterfactual policies. This paper develops bias-aware inference for such functionals when the RC distribution is estimated nonparametrically using the penalized fixed-grid estimator of \citet{heiss2022}. We characterize the sampling distribution of the penalized plug-in estimator and construct confidence intervals that incorporate a feasible upper bound for the bias induced by regularization. These inferential needs arise in applications across environmental economics \citep[e.g.,][]{faure2021}, public economics \citep[e.g.,][]{simonov2022}, industrial organization \citep[e.g.,][]{dubois2020}, and health economics \citep[e.g.,][]{heiss2021}.

Empirical implementations of RC models commonly impose a parametric distribution on unobserved preference heterogeneity, most often a normal or log-normal distribution. Such assumptions are convenient and supported by standard software, but they restrict the shape and tails of the heterogeneity distribution. When these restrictions are misspecified, they can substantially affect the economic quantities computed from the model. For example, \citet{miravete2022} show how parametric restrictions on heterogeneity constrain the curvature of demand functions and the range of estimated elasticities, while \citet{compiani2022} emphasizes that counterfactual predictions can be sensitive to the assumed distribution of heterogeneous preferences. These concerns motivate the use of nonparametric estimators, which estimate the RC distribution under weaker ex ante restrictions and accommodate a wider range of feasible distributional shapes.

A computationally tractable nonparametric approach is the fixed-grid estimator introduced by \citet{bajari2007} and further developed by \citet{fox2011}. It approximates the unknown RC distribution by a discrete distribution over prespecified support points and estimates the probability mass assigned to each point. \citet{fox2016} provide convergence results for the estimated distribution, and the approach has been used in empirical applications such as \citet{nevo2016} and \citet{blundell2020}. Other flexible approaches include the likelihood-based estimators of \citet{train2008} and \citet{train2016}, based on discrete mixtures and approximations using polynomials, splines, or step functions, and the Bayesian hierarchical models of \citet{rossi2005} and \citet{burda2008}, based on countably infinite mixtures of normals. Bayesian methods provide posterior inference for features of the RC distribution, although their implementation typically relies on Markov chain Monte Carlo and can be computationally demanding in large discrete-choice applications. Our analysis focuses on the distinct inferential issues created by regularization in the fixed-grid approach.

An important trade-off researchers face when applying the fixed-grid estimator is the density of the grid. Within an appropriately chosen support domain, a finer grid can reduce the approximation error generated by representing a flexible distribution with finitely many support points. On a dense grid, however, choice probabilities evaluated at nearby support points become difficult to distinguish, producing severe multicollinearity and unstable estimates of the probability weights. \citet{heiss2022} analyze support recovery in this setting and propose adding an $\ell_2$-penalty to stabilize estimation. The penalty can make substantially denser grids feasible, but it also introduces regularization bias into plug-in estimates of the functionals of interest. Relatedly, \citet{escanciano2023} shows that some features of unobserved heterogeneity, including quantiles in relevant settings, can be irregularly identified. Thus, a coarse grid can leave non-negligible approximation bias, whereas a dense grid may call for regularization whose inferential consequences cannot be ignored.

This trade-off has a direct implication for inference. Because the penalty changes the criterion, the penalized plug-in estimator is centered around the functional evaluated at a penalized pseudo-true parameter rather than the corresponding unpenalized fixed-grid target. An interval that accounts only for sampling uncertainty, therefore, need not cover the latter target. Some prominent fixed-grid applications use bootstrap-based inference for functionals obtained from counterfactual simulations \citep[e.g.,][]{nevo2016, blundell2020}. In the penalized setting studied here, however, resampling the estimator without an explicit bias allowance need not account for the shift between the penalized and unpenalized targets. More generally, \citet{fang2018} show that standard bootstrap procedures can fail for directionally differentiable transformations. A related literature develops inference for estimators subject to inequality or boundary constraints \citep[e.g.,][]{andrews1999estimation, fang2023inference, fan2023wald}. Those methods address important features of the unpenalized fixed-grid problem and can apply to linear functionals in some settings, but they do not directly address the additional displacement induced by the $\ell_2$-penalty.

We address this problem by adapting the sieve-inference framework of \citet{chen2014sieve} and \citet{chen2015sieve} to the penalized fixed-grid estimator of \citet{heiss2022}. The support-recovery result in \citet[Theorem~1]{heiss2022} motivates conducting the local analysis on the active set of nonzero probability weights. We complement that result with active-set stability conditions under which the penalized and unpenalized fixed-grid targets share the same active face of the simplex and local perturbations remain on that face with probability approaching one. Under these conditions, binding nonnegativity constraints do not affect the first-order local expansion through varying zero-weight directions. Within this framework, we derive the asymptotic distribution of the plug-in estimator around the functional evaluated at the penalized pseudo-true parameter, use a sieve Riesz representer to estimate sampling uncertainty, and construct a feasible upper bound on the difference between the penalized and unpenalized fixed-grid functionals.

Combining these components yields confidence intervals that are asymptotically conservative for the fixed-grid best approximation target, under the stated conditions, even when regularization bias is not negligible relative to sampling uncertainty. The general construction applies to a class of smooth scalar functionals that includes mean random coefficients, willingness-to-pay measures, elasticities, welfare changes, and counterfactual choice probabilities. It does not require the unpenalized design matrix to be nonsingular and can accommodate irregular functionals whose corresponding Riesz representer's norm may diverge. We also provide a simpler direct bootstrap procedure for regular functionals. In contrast to approaches that remove first-order regularization effects through debiasing or orthogonalization \citep[e.g.,][]{chernozhukov2018, chernozhukov2022}, we retain the penalized estimator and incorporate a bound on its bias into the interval. The contribution is therefore a bias-aware confidence-interval construction for an existing estimator: it makes explicit both the stabilizing role of regularization and its cost for inference.

Monte Carlo simulations illustrate the distinct roles of grid approximation and regularization. In the designs examined, a coarse grid leaves consequential approximation bias over the sample sizes considered and produces substantial undercoverage, while our proposed method performs better along a denser grid. 

In reported designs, the bias-aware intervals are only moderately wider or even narrower on average than the corresponding intervals based on the unpenalized fixed-grid estimator. This finite-design comparison illustrates that a reduction in sampling variance can partly offset the additional allowance for regularization bias. In an application to the travel-mode data of \citet{meijer2006}, the nonparametric specifications place more probability mass at large absolute values in the Value-of-Time (VoT) distribution and imply larger mean VoT magnitudes than the parametric log-normal benchmark. The bias-aware intervals remain informative and, in this application, are only moderately wider than the corresponding unpenalized fixed-grid intervals. These empirical differences illustrate the sensitivity of the reported functionals to distributional assumptions.

The remainder of the paper is organized as follows. Section~\ref{sec:model_est} introduces the RC logit model, the functionals considered in the analysis, and the penalized and unpenalized fixed-grid estimators. Section~\ref{sec:theory} derives the asymptotic theory and constructs the confidence-conservative intervals. Section~\ref{sec:monte_carlo} reports the Monte Carlo study, and Section~\ref{sec:application} presents the empirical application. Section~\ref{sec:conclusion} concludes.


\section{Model and estimator}\label{sec:model_est}
	
This section reviews the RC model and common functionals frequently considered in empirical applications, which are the primary objects of inference in this paper. It then revisits the nonparametric fixed-grid estimator of \citet{fox2011} and the regularized extension of \citet{heiss2022}, which are used for plug-in estimators for these functionals. Finally, Subsection~\ref{sec:bias} discusses the trade-off between approximation and regularization bias, thereby motivating the confidence-conservative inference procedure developed in the following sections.

\subsection{Model}\label{sec:model}

Suppose that each decision maker $i$ faces a finite set of mutually exclusive alternatives $j=0, \ldots, J$, where $j=0$ denotes a possible outside option. For each decision maker $i$ and inside alternative $j=1, \ldots, J$, the researcher observes $y_{i,j}$, where $y_{i,j}=1$ if decision maker $i$ chooses alternative $j$ and $0$ otherwise, together with a $K$-dimensional vector of covariates $x_{i,j}$.\footnote{We focus on cross-sectional data, while extensions of the nonparametric fixed-grid estimator to panel data are discussed in \citet{fox2011}.} 

In line with the theoretical results derived in Section \ref{sec:theory}, we consider a model that allows for a combination of fixed and random coefficients.\footnote{Although assuming some coefficients to be homogeneous across individuals reduces the model's flexibility, we consider specifications with both fixed and random coefficients because the fixed-grid estimator presented below suffers from the curse of dimensionality, making it computationally feasible only for models with relatively few random coefficients in practice.} Let $\delta_0$ denote the $d_\delta$-dimensional vector of true fixed coefficients and $\beta$ the $d_\beta$-dimensional vector of random coefficients with true distribution $F_0$. Furthermore, let 
$x_i=(x_{i,1}',\ldots,x_{i,J}')'$ denote the matrix of stacked covariates across inside alternatives.  Integrating over the distribution of the RCs yields the unconditional choice probability

\begin{equation}\label{eq:rc_logit}
P(y_{i,j}=1\mid x_i)=\int g_j(x_i,\delta_0,\beta)\,\diff F_0(\beta),
\end{equation} 
where $g_j(\cdot)$ is a known function determined by the assumed discrete choice model. An example widely applied in the literature is the RC logit model, in which $g_j(\cdot)$ is given by the logit kernel,
\[
g_j(x_i, \delta_0, \beta)
=
\frac{\exp(x_{i,j,1}'\delta_0 + x_{i,j,2}'\beta)}
{1+\sum_{l=1}^{J}\exp(x_{i,l,1}'\delta_0 + x_{i,l,2}'\beta)},
\]
where $x_{i,j, 1}$ and $x_{i,j,2}$ denote the subvectors of product characteristics associated with the fixed and random coefficients, respectively.\footnote{For a detailed treatment of RC logit models, see \citet[pp.~134--150]{train2009}.}

The researcher's objective is to estimate the unknown parameters $\delta_0$ and $F_0$ from the observed choice data. In empirical applications, however, researchers are typically not directly interested in these estimates, but rather in functionals of the unknown model parameters, such as willingness-to-pay measures, own- and cross-price elasticities, or counterfactual predictions. A broad class of economically meaningful quantities of interest can be written as average functionals of the form
\begin{equation}\label{eq:avg_functionals} 
\tau(\delta_0, F_0) = \int \left\{ \int s(x, \delta_0, \beta)\,\text dF_0(\beta) \right\}\text dG(x),
\end{equation}
where $s(\cdot)$ is a known measurable function and $G$ denotes the population distribution of $x$. In several examples, the outer integral with respect to $G$ drops out. Otherwise, we treat $G$ as known and ignore the additional estimation error arising from estimating $G$.

\begin{example}[Willingness-to-Pay]\label{ex:wtp}
A central quantity of interest in empirical applications of RC logit models is consumers' willingness-to-pay (WTP) for alternatives' attributes. For example, in environmental economics, WTP measures are commonly used to quantify consumers' valuations of environmentally friendly product attributes \citep[e.g.,][]{faure2021}. Let $m\in\{1,\ldots,d_{\beta}\}$ index the monetary attribute (e.g., price/cost) and let $k\neq m$ index a non-monetary attribute, with corresponding coefficients $\beta_m$ and $\beta_k$. Under linear utility, marginal WTP for attribute $k$ is $-\beta_k/\beta_m$ \citep[e.g., ][]{meijer2006}. The corresponding population mean marginal WTP is therefore
\begin{equation}\notag
\tau_{\text{wtp},0}=\int(-\frac{\beta_k}{\beta_m})\,\text dF_0(\beta).
\end{equation} 
\end{example}

\begin{example}[Mean Elasticities]\label{ex:elasticity}
Mean own- and cross-price elasticities are frequently reported in industrial organization and marketing to measure consumers' responses to price changes in RC logit models \citep[e.g.,][]{dubois2018}. 
For alternatives $j,l\in\{1,\ldots,J\}$ and attribute $k\in\{1,\ldots,K\}$ (typically price), the corresponding population mean elasticity is given by
\begin{equation}\notag
\tau_{\text{elasticity},0} = \int \left\{ \frac{x_{l,k}}{P_j(x)} \int \frac{\partial g_j(x,\delta_0,\beta)} {\partial x_{l,k}} \,\mathrm dF_0(\beta) \right\} \,\mathrm dG(x).
\end{equation}
For $j=l$, this corresponds to an own-price elasticity, whereas for $j\neq l$ it corresponds to a cross-price elasticity \citep[see][pp.~160--161]{train2009}. In practice, these elasticities are often evaluated at a representative value $x^*$, where $(x^*)'=((x_1^*)',\ldots,(x_J^*)')$ is a point in the support of the stacked covariates with distribution $G$. In that case, the outer integral is replaced by point evaluation.
\end{example}

\begin{example}[Welfare Changes]\label{ex:welfare}
Welfare changes are frequently used in empirical applications of RC logit models to quantify the effects of policy interventions, such as taxes, subsidies, or product innovations, on consumer welfare \citep[e.g.,][]{grigolon2018}. Let $J_{\text{pre}}$ and $J_{\text{post}}$ denote the numbers of inside alternatives available before and after the policy intervention, respectively, and let $x_{\text{pre}}^*$ and $x_{\text{post}}^*$ denote the corresponding product characteristics. Furthermore, let $m$ denote the monetary attribute (typically price). Under the RC logit model with i.i.d. Type I Extreme Value errors, the corresponding population mean welfare change is given by
\begin{equation}\notag
\begin{aligned}
\tau_{\text{welfare},0} &= \int \Bigg\{ \int \frac{1}{-\beta_m} \Bigg[
\log\!\Big( 1+\sum_{j=1}^{J_{\text{post}}} \exp\!\big((x_{\text{post},j,1}^*)'\delta_0
+(x_{\text{post},j,2}^*)'\beta\big) \Big) \\
&\qquad\qquad - \log\!\Big( 1+\sum_{j=1}^{J_{\text{pre}}} \exp\!\big((x_{\text{pre},j,1}^*)'\delta_0 +(x_{\text{pre},j,2}^*)'\beta\big) \Big) \Bigg] \,\mathrm dF_0(\beta) \Bigg\} \,\mathrm dG(x_{\text{pre}}^*,x_{\text{post}}^*).
\end{aligned}
\end{equation}
If the pre- and post-intervention product characteristics are fixed, the outer integral with respect to $G$ drops out. Positive values of $\tau_{\text{welfare},0}$ indicate welfare gains under the post-intervention scenario. See \citet[pp.~64--67]{train2009} and \citet{karlstrom2005} for further discussion.
\end{example}

\subsection{Fixed-grid estimator}\label{sec:estimator}
		
A common approach in applications of RC models is to restrict $F_0$ to a prespecified parametric family, such as the normal or log-normal family. Although computationally convenient, such a restriction may misrepresent the underlying preference heterogeneity and thereby distort estimates of the functionals of interest.
			
To avoid parametric shape restrictions, \citet{bajari2007} and \citet{fox2011} propose a computationally tractable nonparametric estimator of the RC distribution. Let $R_N$ denote the number of fixed grid points, and let $\calB_{R_N}\equiv\{\beta_1,\ldots,\beta_{R_N}\}\subset\mathbb{R}^{d_\beta}$ denote the grid specified by the researcher before estimation. The probability mass $\theta_r$ associated with grid point $\beta_r$, $r=1,\ldots,R_N$, is estimated from the data. The estimator approximates $F_0$ by the discrete distribution $F(\cdot)=\sum_{r=1}^{R_N}\theta_r\mathbf{1}\{\beta_r\leq\cdot\}$.\footnote{Unless otherwise stated, vector inequalities such as ``$\leq$'' are understood componentwise.} As $R_N$ increases with $N$, a suitably chosen grid can provide increasingly accurate approximations to $F_0$.\footnote{{\citet{fox2016} show that distributions supported on the grid can approximate any RC distribution arbitrarily closely in the weak topology as $R_N\rightarrow\infty$ and the grid becomes dense in the coefficient space.}}

To estimate the probability weights, consider the linear probability approximation that arises from replacing $F_0$ by a discrete distribution supported on the finite grid $\mathcal{B}_{R_N}$ \cite[e.g.,][Equation 11]{heiss2022}
\begin{equation}\label{eq:linear_prob}
	y_{i,j} = \sum_{r=1}^{R_N} g_j(x_i, \delta, \beta_r)\theta_r+\epsilon_{i,j},
\end{equation}
where $\epsilon_{i,j}=y_{i,j}-P_j(x_i)$ is the linear probability error.  
Building on \eqref{eq:linear_prob}, \citet{heiss2022} propose the penalized estimator (referred to as the ENet estimator in later discussions)
			\begin{align}\label{eq:enet}
				\begin{split}
					(\widehat{\delta}, \widehat{\theta}) = \argmin_{\delta, \theta} \ & \frac{1}{2NJ} \sum_{i=1}^N (y_i - \sum\limits_{r=1}^{R_N}g(x_i,\delta, \beta_r)\theta_r)'\Sigma(x_i)^{-1}(y_i - \sum\limits_{r=1}^{R_N}g(x_i,\delta, \beta_r)\theta_r) + \frac{\mu_N}{2} \sum_{r=1}^{R_N} \theta_r^2\\
					\text{s.t.} \quad & \theta_r \geq 0 \quad \forall r, \quad \sum_{r=1}^{R_N} \theta_r = 1,
				\end{split}
			\end{align} 
where $\theta=(\theta_1,\ldots,\theta_{R_N})'$, $y_i=(y_{i,1},\ldots,y_{i,J})'$, and $g(x_i, \delta,\beta_r)=(g_1(x_i,\delta, \beta_r),\ldots,g_J(x_i,\delta, \beta_r))'$ for all $r=1,\ldots,R_N$. The $J\times J$ matrix $\Sigma(x_i)^{-1}$ is a positive-definite weighting matrix (e.g., \citet{heiss2022} use the identity matrix).\footnote{For the construction of a variance-based weighting matrix, see Section~5 of \citet{fox2011}.} The tuning parameter $\mu_N\geq0$ controls the strength of the $\ell_2$-penalty. Setting $\mu_N=0$ recovers the estimator of \citet{fox2011} (referred to as the FKRB estimator in later discussions), whereas larger values imply stronger regularization. \citet{heiss2022} recommend selecting $\mu_N$ by cross-validation. For models with both fixed and random coefficients, \citet{heiss2022} implement \eqref{eq:enet} with an iterative algorithm that alternates between updating the probability weights through constrained least squares and updating the fixed coefficients through a weighted conditional logit likelihood. We use this algorithm in the empirical application; Appendix~B of \citet{heiss2022} provides details.

The estimated RC distribution is thus obtained from the estimated probability weights $\widehat{\theta}$ as
\begin{equation*}
\widehat{F}(\beta) = \sum\limits_{r = 1}^{R_N} \widehat{\theta}_r \ \mathbf{1}\left\{\beta_r\leq \beta\right\}.
\end{equation*}   		
Given the fixed-grid estimator $(\widehat{\delta}, \widehat{F})$ and the object of interest in the form of Equation (\ref{eq:avg_functionals}), Section \ref{sec:theory} derives a bias-aware inference approach based on the plug-in estimator $\tau(\widehat{\delta}, \widehat{F})$.

\subsection{Regularization in the fixed-grid estimator}\label{sec:bias}

Although the estimator described in the previous subsection is computationally attractive, its statistical properties are governed by a fundamental trade-off between approximation bias and regularization bias. On the one hand, the grid should be sufficiently dense so that the discrete approximation accurately represents the unknown RC distribution. On the other hand, increasing the number of grid points makes the estimation problem increasingly ill-conditioned, requiring regularization.

As the grid becomes denser, the resulting design matrix becomes increasingly collinear, because, e.g., neighboring grid points generate very similar choice probabilities and thus similar regressors in the linear probability model in Equation~(\ref{eq:linear_prob}), rendering the constrained least squares problem ill-conditioned. This issue has been documented for the fixed-grid estimator by \citet{fox2016}, \citet{nevo2016}, and \citet{heiss2022}, and more generally for structural models with nonparametric unobserved heterogeneity by \citet{escanciano2023}. 

A natural way to address this issue is through regularization. In the context of the fixed-grid estimator, regularization permits the use of substantially denser grids, thereby reducing sieve approximation error and improving estimation of the underlying RC distribution. Specifically, \citet{heiss2022} augment the estimator of \citet{fox2011} with the $\ell_2$ penalty in Equation~(\ref{eq:enet}), while \citet{escanciano2023} develop a related regularized estimator.

However, regularization comes at the cost of inducing regularization bias. Consequently, researchers face an inherent trade-off. A coarse grid reduces ill-conditioning but may induce substantial approximation bias because the support is too sparse to accurately represent the true distribution. Conversely, a dense grid reduces approximation bias but requires stronger regularization, thereby increasing regularization bias. The following section develops a valid inference for average functionals of the form in Equation (\ref{eq:avg_functionals}) that accounts for the regularization bias of the penalized fixed-grid estimator in Equation (\ref{eq:enet}).

		
\section{Theoretical results}\label{sec:theory} 

This section derives conservative confidence intervals for plug-in estimators of functionals, including the average functionals in \eqref{eq:avg_functionals}. To accommodate the near-singular design matrices that arise with dense grids \citep{nevo2016,fox2016}, we use the penalized fixed-grid estimator of \citet{heiss2022}. The main result establishes asymptotic normality around a penalized pseudo-true parameter. We then bound the regularization bias induced by the $\ell_2$-penalty in \eqref{eq:enet} and incorporate that bound into the confidence interval.
				
The analysis builds on the sieve-inference frameworks of \citet{chen2014sieve} and \citet{chen2015sieve}, adapted to the fixed-grid estimator of \citet{fox2011} and the penalized extension of \citet{heiss2022}. The central idea is to linearize the functional locally around the penalized pseudo-true parameter and characterize its first-order term with a sieve Riesz representer. This yields an asymptotically normal representation, an explicit asymptotic variance, and a tractable regularization-bias bound.

\paragraph{Notation.}
Let $\operatorname{clsp}(A)$ denote the closed linear span of $A$ under the norm in \eqref{eq:norm.inner.product}. For positive deterministic sequences $(a_N)$ and $(b_N)$, write $a_N\lesssim b_N$ if $a_N/b_N\leq C$ for some finite constant $C$ and all sufficiently large $N$; write $a_N\asymp b_N$ if both $a_N\lesssim b_N$ and $b_N\lesssim a_N$. For positive random sequences $(U_N)$ and positive deterministic sequences $(V_N)$, write $U_N=O_p(V_N)$ if $\lim_{M\to\infty}\limsup_{N\to\infty}P(U_N/V_N>M)=0$, and write $U_N=o_p(V_N)$ if $\lim_{N\to\infty}P(U_N/V_N>\varepsilon)=0$ for every $\varepsilon>0$.
We omit the subscript ``$p$'' for nonrandom sequences. Let $|A|$ denote the cardinality of a set $A$, and let $\mathbb{R}_+=[0,\infty)$. For a vector $a\in\mathbb{R}^K$, let $a'$ denote its transpose and $\|a\|_e=(a'a)^{1/2}$. For a matrix $A$, let $\|A\|_e=\{\operatorname{tr}(A'A)\}^{1/2}$ denote the Frobenius norm. For a map $\phi:\boldsymbol{B}\to\boldsymbol{H}$ between Banach spaces, define
\[
D\phi(\alpha)[v]\equiv\left.\frac{\partial}{\partial t}\phi(\alpha+tv)\right|_{t=0}.
\]
Finally, let $E(\cdot)$ denote expectation/expected value, and $\bar{E}_N\{h(z)\}\equiv N^{-1}\sum_{i=1}^N\{h(z_i)-E[h(z_i)]\}$ denote the centered empirical process indexed by $h$ for a sample $z=(z_1,\cdots, z_N)$.
			
\subsection{Model and estimator revisited}

Let $\calB\subset\mathbb{R}^{d_\beta}$ be the compact support of the RCs, and let $\calF$ be the class of distributions supported on $\calB$. The parameter space is $\calA\equiv\calC\times\calF$, with
\[
\alpha=(\delta,F),\qquad\delta\in\calC\subset\mathbb{R}^{d_\delta},\qquad F\in\calF.
\]
Let $\alpha_0=(\delta_0,F_0)$ denote the true parameter.
				
For each sample size $N$, let $\mathcal{B}_N=\{\beta_1,\ldots,\beta_{R_N}\}\subset\mathcal{B}$ denote the fixed grid of $R_N$ support points. Define the simplex
\[
\calD_N=\{\theta\in[0,1]^{R_N}:\iota_{R_N}'\theta=1\},
\]
where $\iota_k$ denotes the $k$-vector of ones. For $\theta\in\calD_N$, let
\[
F_\theta(\cdot) =\theta'\mathbf{1}_N(\cdot) =\sum_{r=1}^{R_N}\theta_r\mathbf{1}\{\beta_r\leq\cdot\},
\]
where $\mathbf{1}_N(\cdot)$ stacks $\{\mathbf{1}(\beta_r\leq\cdot):r=1,\ldots,R_N\}$.
The sieve parameter space is
\[
\calA_N=\{(\delta,F_\theta):\delta\in\calC,\ \theta\in\calD_N\}.
\]

For each $\alpha\in\mathcal A_N$, we write $\alpha=(\delta_\alpha,F_{\theta_\alpha})$ to distinguish from the corresponding components of the sieve Reisz representer introduced below. The notation $F_{\theta_\alpha}$ emphasizes that the distribution component of $\alpha$ is determined by the probability weights $\theta_\alpha$, which is subject to the simplex constraint, assigned to the fixed grid points.

For $\alpha=(\delta_\alpha,F_{\theta_\alpha})\in\calA_N$, define the strong norm
\[
\|\alpha\|_s\equiv\|\delta_\alpha\|_e+\|\theta_\alpha\|_e,
\]				
 
and define the unpenalized weighted least-squares criterion
\[
\ell_N^0(z_i,\alpha)\equiv\frac12\rho(z_i,\alpha)'\Sigma(x_i)^{-1}\rho(z_i,\alpha),
\]
with $\rho(z_i,\alpha)\equiv y_i-P(x_i,\alpha)$, where $z_i=(x_i,y_i)$ are i.i.d., $P(x_i,\alpha)$ is the model-implied vector of choice probabilities, and $\Sigma(x)$ is positive definite with eigenvalues uniformly bounded above and away from zero. Moreover, let 
\[
Q_N^0(\alpha)\equiv E[\ell_N^0(z,\alpha)],
\qquad
\ell_N^\mu(z,\alpha)\equiv\ell_N^0(z,\alpha)+\frac{\mu_N}{2}\|\theta_\alpha\|_e^2,
\qquad
Q_N^\mu(\alpha)\equiv E[\ell_N^\mu(z,\alpha)].
\]
define the unpenalized and penalized population objective function, respectively. In comparison to \eqref{eq:enet}, we omit the constant factor $1/J$, which has no effect once the penalty parameter is rescaled accordingly. 

When $\mu_N=0$, the population criterion may have multiple minimizers because a dense grid can produce a singular design matrix. When $\mu_N>0$, we assume that the penalized population criterion has a unique minimizer. 

Let $\alpha_N^0$ denote a best fixed-grid approximation to the population target,\footnote{When $F_0$ has a density, one common benchmark assigns to each grid point the probability of an associated grid cell. If $F_0$ is discrete and its mass points are included in the grid, the approximation error can vanish.} and let $\alpha_N^\mu$ denote the penalized pseudo-true parameter:
\[
\alpha_N^0\in\arg\min_{\alpha\in\calA_N}Q_N^0(\alpha),
\qquad
\alpha_N^\mu=\arg\min_{\alpha\in\calA_N}Q_N^\mu(\alpha).
\]

The sample analogue of $Q_N^\mu(\alpha)$ is
\[
\widehat Q_N^\mu(\alpha)=\frac1N\sum_{i=1}^N\ell_N^\mu(z_i,\alpha)
\]
and the corresponding estimator is defined by
\[
\widehat{\alpha}_N^\mu\in\arg\min_{\alpha\in\calA_N}\widehat Q_N^\mu(\alpha).
\qquad
\]
The unpenalized sample objective and estimator are obtained by setting $\mu_N=0$ and are denoted by $\widehat Q_N^0$ and $\widehat{\alpha}_N^0$, respectively.
		
Since some probability weights may be zero at the penalized optimum, we define the following restricted sieve spaces. For any $S\subseteq\{1,\ldots,R_N\}$, let 
\[
\calA_N(S)\equiv\{(\delta,F_{\theta})\in\calA_N:\theta_r=0\text{ for }r\notin S\}.
\]
The active set of the penalized pseudo-true parameter is
\[
S_N=\{r:\theta_{\alpha_N^\mu,r}>0\},
\]
where $\theta_{\alpha_N^\mu,r}$ is the $r$th entry of penalized probability weights,  $\theta_{\alpha_N^\mu}$ such that $\alpha_N^\mu=(\delta_{\alpha_N^\mu}, F_{\theta_{\alpha_N^\mu}})$. 

The subsequent analysis is local to this active set. For a positive sequence $\xi_N$, define the local neighborhood
\begin{equation}
\calN_N^\mu=\left\{\alpha\in\calA_N(S_N):Q_N^\mu(\alpha)-Q_N^\mu(\alpha_N^\mu)\le \xi_N\log N\right\}.
\label{eq:xi_N neighborhood}
\end{equation}
The positive sequence $\xi_N$ typically converges to zero. Lemma~\ref{lem: consistency} in Appendix \ref{appen;proof main context} provides one admissible rate.
To characterize local perturbations around the penalized pseudo-true parameter, define the local direction space
\[
\calV_N^\mu \equiv \operatorname{clsp}\bigl(\calN_N^\mu-\{\alpha_N^\mu\}\bigr).
\]
		
\begin{assumption}[Active set stability]\label{assum:ri}
(i) $\alpha_N^\mu$ and $\alpha_N^0$ lie in the relative interior of $\calN_N^\mu$ under the strong norm $\|\cdot\|_s$.\footnote{For a set $S$, $\operatorname{ri}(S)\equiv\{x\in S:\text{there exists }\varepsilon>0\text{ such that }B_\varepsilon(x)\cap\operatorname{aff}(S)\subseteq S\}$, where $B_\varepsilon(x)\equiv\{y:\|y-x\|_s\leq\varepsilon\}$ and $\operatorname{aff}(S)$ is the affine hull of $S$. See Section~6 of \citet{rockafellar1997convex}.}
(ii) For every $v\in\calV_N^\mu$, there exists a deterministic $\rho_N(v)>0$ such that, for every positive deterministic sequence $\varepsilon_N=o(\min\{N^{-1/2},\rho_N(v)\})$,
\[
\widehat{\alpha}_N^\mu\pm\varepsilon_Nv\in\calN_N^\mu
\quad\text{with probability approaching one.}
\]
\end{assumption}
Assumption~\ref{assum:ri} requires the best approximation and the penalized target to share an active face of the simplex and ensures that the estimator admits the local perturbations used below. Theorem~1 in \citet{heiss2022} establishes support recovery by the penalized sample estimator relative to the best fixed-grid weights under sufficient conditions. Lemma~\ref{lem:verification of Assumption ri} in Appendix \ref{assum:ri} gives related sufficient conditions for Assumption~\ref{assum:ri}. 

These conditions permit $\mu_N=0$ when the unpenalized population criterion has sufficient local curvature. When the design is nearly singular, a positive $\ell_2$-penalty provides the curvature required for local identification, albeit around a biased target. The simulation results below illustrate this trade-off.
		
Following the local quadratic expansion used in \citet[Section 3.1]{chen2014sieve}, consider directions $v_1,v_2\in\calV_N^\mu$ and define the inner product
\begin{equation}\label{eq:inner.product}
\langle v_1,v_2\rangle_\mu \equiv \left. \frac{\partial^2}{\partial t_1\partial t_2} E\!\left[\ell_N^\mu\!\left(z,\alpha_N^\mu+t_1v_1+t_2v_2\right) \right] \right|_{t_1=t_2=0}.
\end{equation}
The induced norm is
\begin{equation}\label{eq:norm.inner.product}
\|v\|_\mu^2 \equiv \left. \frac{\partial^2}{\partial t^2} E\!\left[\ell_N^\mu\!\left(z,\alpha_N^\mu+tv\right) \right] \right|_{t=0}.
\end{equation}
Under $\mu_N>0$ and the maintained regularity conditions, this norm remains well defined even when the design matrix is near singular. We also assume that differentiation and expectation can be interchanged through the second order, so that
\begin{equation}\label{eq:inner.product 2}
\langle v_1,v_2\rangle_\mu=E\!\left[D^2\ell_N^\mu(z,\alpha_N^\mu)[v_1,v_2]\right],
\end{equation}
where
\begin{align*}
D\ell_N^\mu(z,\alpha_N^\mu)[v_1]&\equiv\left.\frac{\partial}{\partial t}\ell_N^\mu\!\left(z,\alpha_N^\mu+tv_1\right)\right|_{t=0},\\
D^2\ell_N^\mu(z,\alpha_N^\mu)[v_1,v_2]&\equiv\left.\frac{\partial^2}{\partial t_1\partial t_2}\ell_N^\mu\!\left(z,\alpha_N^\mu+t_1v_1+t_2v_2\right)\right|_{t_1=t_2=0}.
\end{align*}
Equipped with the inner product in Equation (\ref{eq:inner.product}), $\calV_N^\mu$ is a finite-dimensional Hilbert space. Its elements respect the active-set restriction: directions outside $S_N$ have zero grid-weight components. Thus, any $v\in\calV_N^\mu$ can be represented as
\[
v=(\delta_v,\theta_v'\mathbf{1}_{S_N}(\cdot)),
\]
where $\delta_v \in \mathbb{R}^{d_\delta}, \theta_v\in\mathbb{R}^{|S_N|}$ and $\mathbf{1}_{S_N}(\cdot)$ stacks the indicator functions $\{\mathbf{1}(\beta_r\le\cdot)\}_{r\in S_N}$. Moreover, for any $\alpha\in\calN_N^\mu$ and $v\in\calV_N^\mu$, sufficiently small local perturbations $\alpha+tv$ remain in $\calN_N^\mu$ whenever the path stays in the relative interior. Thus, by restricting to the interior point within $\calN_N^\mu$, we avoid the boundary issue raised in constrained optimizations. 
		
Having characterized the local geometry of the penalized estimator, we now turn to the scalar functional $\tau(\alpha)$, which is the object of primary interest. 
For a scalar functional $\tau:\calA_N\to\mathbb{R}$, define the directional derivative at the penalized target in the direction $v\in\calV_N^\mu$ by
\[
D\tau(\alpha_N^\mu)[v]\equiv\left.\frac{\partial}{\partial t}\tau\!\left(\alpha_N^\mu+tv\right)\right|_{t=0}.
\]
If $D\tau(\alpha_N^\mu)[\cdot]$ is linear on $\calV_N^\mu$, then finite-dimensionality of $\calV_N^\mu$ implies the existence of a sieve Riesz representer $v_N^\tau\in\calV_N^\mu$ such that
\begin{equation}
D\tau(\alpha_N^\mu)[v]=\langle v_N^\tau,v\rangle_\mu
\quad\text{for all }v\in\calV_N^\mu,\qquad D\tau(\alpha_N^\mu)[v_N^\tau] =\|v_N^\tau\|_\mu^2 = \sup _{\substack{v \in \calV_N^\mu, v \neq 0}} \frac{|D\tau(\alpha_N^\mu)[v]|^2}{\|v\|_\mu^2}.
\label{eq:Riesz}
\end{equation}
			
Let $\gamma=(\delta',\theta')'\in\mathbb{R}^{d_{\vdelta}}\times\mathbb{R}^{|S_N|}$, and define
\begin{align*}
{D}_{N}^\mu&\equiv \left. \frac{\partial \tau({\alpha}_N^\mu+(\delta,\theta'\mathbf{1}_{S_N}(\cdot)))}{\partial \gamma}\right|_{\gamma=0},\\
{H}_{N}^\mu&\equiv E\left[\left.\frac{\partial^2 \ell_{N}^\mu(z,{\alpha}_N^\mu+(\delta,\theta'\mathbf{1}_{S_N}(\cdot)))}{\partial \gamma \partial \gamma'}\right|_{\gamma=0}\right].
\end{align*}
The Riesz representer solving Equation (\ref{eq:Riesz}) is 
\[
v_N^\tau=(\delta^\tau,\theta^{\tau'}\mathbf{1}_{S_N}(\cdot)),
\qquad(\delta^{\tau'},\theta^{\tau'})'=({H}_{N}^\mu)^{-1}{D}_{N}^\mu,
\]

whose derivation follows the same steps as the solution of Equation~(\ref{eq:empirical sieve riesz representer equation system}) and is therefore omitted. Example \ref{ex:wtp 3} gives the corresponding expression for the mean WTP functional introduced in Example \ref{ex:wtp}.
			
\begin{example}[Willingness-to-Pay (continued)]\label{ex:wtp 3}
Consider the WTP functional without fixed coefficients. For simplicity, suppose that $S_N=\{1,\ldots,R_N\}$. If $\beta_{r,m}\neq0$ for all relevant grid points, then $D\tau(\alpha_N^\mu)[h]=\sum_{r=1}^{R_N}h_r\left(-\frac{\beta_{r,k}}{\beta_{r,m}}\right),$ and thus $$D_N^\mu=\left(-\frac{\beta_{1,k}}{\beta_{1,m}},\ldots,-\frac{\beta_{R_N,k}}{\beta_{R_N,m}}\right)'.$$ Let  $${H}_{N}^\mu=E\left[\left.\frac{\partial^2 \ell_{N}^\mu(z,{\alpha}_N^\mu+\theta'\mathbf{1}_{S_N}(\cdot))}{\partial \theta \partial \theta'}
\right|_{\theta=0}\right],$$ the corresponding sieve Riesz representer is $v_N^\tau=\sum_{r=1}^{R_N}\gamma_{N,r}^\mu\mathbf{1}(\beta_r\le\cdot)$ with $\gamma_N^\mu=(H_N^\mu)^{-1}D_N^\mu$  and $\|v_N^\tau\|_\mu^2=(D_N^{\mu})' (H_N^\mu)^{-1}D_N^\mu$. 
\end{example}
~\\
			
For any $v=(\delta_v,\theta_v'\mathbf{1}_{S_N}(\cdot))\in\mathcal V_N^\mu$, define the regularized score norm
\[
\|v\|_{sd,\mu}^2\equiv\operatorname{Var}\!\left(D\ell_N^\mu(z,\alpha_N^\mu)[v]\right)+\mu_N\|\theta_v\|_e^2.
\]
For the Riesz representer, define the score variance
\[
\sigma_N^2\equiv\operatorname{Var}\!\left(D\ell_N^\mu(z,\alpha_N^\mu)[v_N^\tau]\right).
\]
and, without loss of generality, we only consider $\sigma_N>0$.
            
\begin{assumption}[Sieve variance]\label{assum:norm ratio}
For every $v\in\calV_N^\mu$, $\|v\|_\mu\asymp\|v\|_{sd,\mu}$.
\end{assumption}
Assumption~\ref{assum:norm ratio} requires the local curvature of the penalized population criterion and the sampling variation of its first-order derivative (plus the $\ell_2$-penalty term) to be uniformly comparable over the sieve direction space. 
It is a regularized analog of Assumption 3.2 in \citet{chen2014sieve}. The additional deterministic $\ell_2$ component in $\|v\|_{sd,\mu}$ is induced by the $\ell_2$-penalty term.

To establish asymptotic normality around the penalized pseudo-true value and derive asymptotically conservative confidence intervals for the best fixed-grid target, we impose local smoothness conditions on the criterion function and the functional. These conditions parallel those in \citet{chen2014sieve}, but here they are centered at $\alpha_N^\mu$ for each fixed grid and penalty. Their formal statements are collected in the appendix.

\begin{theorem}\label{theo:3.1 v3 0}
    Under Assumptions  \ref{assum:ri}, \ref{assum:norm ratio}, \ref{assum:average_functional}-\ref{assum:clt v2}, then \[
					\frac{\sqrt{N}\{\tau(\widehat{\alpha}_N^\mu)
					-\tau(\alpha_N^\mu)\}}{\sigma_N}
					\ \longrightarrow_d\ \mathcal{N}(0,1).
				\]
\end{theorem}
\begin{proof}
See Appendix~\ref{Proof of Theorem ref{theo:3.1 v3 0}}.
\end{proof}
 
Theorem \ref{theo:3.1 v3 0} establishes asymptotic normality of the penalized plug-in estimator around the penalized pseudo-true value $\tau(\alpha_N^\mu)$. The following corollary combines this result with a bound on the regularization bias to obtain an asymptotically conservative confidence interval for the unpenalized fixed-grid target $\tau(\alpha_N^0)$.

\begin{corollary}\label{theo:3.1 v3}
Under Assumptions  \ref{assum:ri}, \ref{assum:average_functional},
\[
|\tau(\alpha_N^0)-\tau(\alpha_N^\mu)|\le\|v_N^\tau\|_\mu (\|\alpha_N^0-\alpha_N^\mu\|_\mu+o(N^{-1/2})).
\]
If, in addition, Assumptions \ref{assum:norm ratio}, \ref{assum:Local Behavior of Criterion v2}, and \ref{assum:clt v2} hold, let			
\[
c_N(q)\equiv N^{-1/2}\Phi^{-1}(1-q/2)\sigma_N +\|v_N^\tau\|_\mu \|\alpha_N^0-\alpha_N^\mu\|_\mu,
\]
then
\begin{equation}\label{eq:confidence interval}
\liminf_{N\to\infty}P\!\left(\tau(\alpha_N^0)\in[\tau(\widehat{\alpha}_N^\mu)-c_N(q),\tau(\widehat{\alpha}_N^\mu)+c_N(q)]\right)\geq1-q,
\end{equation}
where \(\Phi^{-1}(u)\) denotes the \(u\)-quantile of the standard normal distribution. 
\end{corollary}
\begin{proof}
See Appendix~\ref{Proof of Theorem ref{theo:3.1 v3}}.
\end{proof}

Corollary~\ref{theo:3.1 v3} separates sampling uncertainty, measured by $\sigma_N$, from regularization bias, bounded by $\|v_N^\tau\|_\mu\|\alpha_N^0-\alpha_N^\mu\|_\mu$.
A larger $\mu_N$ may reduce sampling variance but generally increases regularization bias. For a fixed penalty, a denser grid may reduce approximation bias while increasing ill-conditioning and sampling uncertainty. The coverage statement concerns the best fixed-grid target $\tau(\alpha_N^0)$. It extends to $\tau(\alpha_0)$ under an additional approximation condition such as negligible approximation bias $|\tau(\alpha_N^0)-\tau(\alpha_0)|=o(N^{-1/2}\sigma_N)$. The theorem also allows irregular functionals, for which $\sigma_N$ and $\|v_N^\tau\|_\mu$ may diverge. The following subsections construct feasible estimators of the sampling and bias components.

\subsection{Sieve Riesz representer and bias bounds}\label{sec:riesz.representer}

We first estimate the sieve Riesz representer by its empirical counterpart and then use it to construct feasible sampling-variance and bias terms.
Let
\[
\widehat S_N=\{r:\theta_{\widehat{\alpha}_N^\mu,r}>0\},
\qquad
\widehat{\calV}_N^\mu=
\operatorname{clsp}\bigl(\calA_N(\widehat S_N)-\{\widehat\alpha_N^\mu\}\bigr).
\]
For $v_1,v_2\in\widehat{\calV}_N^\mu$, define the empirical inner product
\[
\langle v_1,v_2\rangle_N\equiv\frac1N\sum_{i=1}^N D^2\ell_N^\mu(z_i,\widehat\alpha_N^\mu)[v_1,v_2],
\]
with induced norm $\|\cdot\|_N$. The empirical directional derivative is
\[
D\tau(\widehat{\alpha}_N^\mu)[v]\equiv\left.\frac{\partial}{\partial t}\tau\!\left(\widehat{\alpha}_N^\mu+tv\right)\right|_{t=0}.
\] 
The empirical sieve Riesz representer $\widehat v_N^\tau$ solves
\begin{equation}\label{eq:empirical sieve riesz representer equation system}
	D\tau(\widehat{\alpha}_N^\mu)[\widehat v_N^\tau]=\sup_{v\in\widehat{\calV}_N^\mu,\ v\neq0}\frac{|D\tau(\widehat{\alpha}_N^\mu)[v]|^2}{\|v\|_N^2},\qquad D\tau(\widehat{\alpha}_N^\mu)[v]=\langle v,\widehat v_N^\tau\rangle_N.
\end{equation}
Lemma \ref{lem:empirical sieve Riesz representer closed solution} gives its closed form.

\begin{lemma}\label{lem:sieve variance estimator} 
	(i)
	Under Assumptions~\ref{assum:ri} and \ref{assum:riesz representer estimator I}(i)--(iii),
	\[
	\left|\frac{\|\widehat v_N^\tau\|_\mu}{\|v_N^\tau\|_\mu}-1\right|=O_p(\epsilon_N^*),
	\qquad
	\frac{\|\widehat v_N^\tau-v_N^\tau\|_\mu}{\|v_N^\tau\|_\mu}=O_p(\epsilon_N^*),
	\]
	where $\epsilon_N^*$ is specified in Assumption~\ref{assum:riesz representer estimator I}.
	
	(ii) If, in addition, Assumptions~\ref{assum:norm ratio} and \ref{assum:riesz representer estimator I}(iv) hold, then
	\[
	\left|
	\frac{\|\widehat v_N^\tau\|_{sd,\mu,N}}
	{\|v_N^\tau\|_{sd,\mu}}-1
	\right|=O_p(\epsilon_N^*),
	\]
	where, continuing with the notation from Lemma \ref{lem:empirical sieve Riesz representer closed solution},
	\[
	\|\widehat v_N^\tau\|_{sd,\mu,N}^2
	=
	\widehat\sigma_N^2+\mu_N\|\widehat\theta^\tau\|_{e}^2,
	\qquad
	\widehat\sigma_N^2=
	\frac1N\sum_{i=1}^N
	\left(D\ell_N^\mu(z_i,\widehat\alpha_N^\mu)[\widehat v_N^\tau]\right)^2.
	\]
\end{lemma}
\begin{proof}
	See Appendix~\ref{proof of Lemma {lem:sieve variance estimator}}.
\end{proof}
The difference between $\alpha_N^0$ and $\alpha_N^\mu$ is driven by the $\ell_2$ penalty applied to the probability weights. Lemma \ref{lem:bias bound} bounds this difference in the pseudo-norm $\|\cdot\|_\mu$. The corresponding difference in the strong norm $\|\cdot\|_s$ need not be small.

\begin{lemma}\label{lem:bias bound}
	Under Assumptions~\ref{assum:ri} and \ref{assum:riesz representer estimator I}(v),
	\[
	\|\alpha_N^0-\alpha_N^\mu\|_\mu
	\le
	\sqrt{\mu_N}+O(\epsilon_N^*\xi_N),
	\]
	where $\xi_N$ is as in Equation (\ref{eq:xi_N neighborhood}) and $\epsilon_N^*$ is as in Assumption \ref{assum:riesz representer estimator I}.
\end{lemma}
\begin{proof}
	See Appendix~\ref{proof of {lem:bias bound}}.
\end{proof}
If $\epsilon_N^*\xi_N=o(N^{-1/2})$ (see, for example, Assumption~5.2 of \citet{chen2014sieve}), Corollary \ref{theo:3.1 v3} and Lemmas~\ref{lem:sieve variance estimator}--\ref{lem:bias bound} yield the feasible half-length
\[
\widehat c_N(q)
=\Phi^{-1}(1-q/2)N^{-1/2}\widehat\sigma_N
(\sqrt{\mu_N}+N^{-1/2})
(1+N^{-1/2}\xi_N^{-1})
\|\widehat v_N^\tau\|_\mu.
\]
Thus,
$[\tau(\widehat{\alpha}_N^\mu)-\widehat c_N(q),
\tau(\widehat{\alpha}_N^\mu)+\widehat c_N(q)]$
is a feasible conservative interval; Corollary~\ref{lem:append ci 1} gives the formal statement. It does not require the unpenalized design matrix to be nonsingular and can accommodate irregular functionals. Its implementation, however, requires the estimated sieve Riesz representer and $\widehat\sigma_N^2$. The next subsection gives a simpler bootstrap procedure for regular functionals.

\subsection{Bootstrap}

This subsection provides a direct bootstrap procedure for regular functionals. We call $\tau(\cdot)$ regular at $\alpha_N^\mu$ when the Riesz representer of $D\tau(\alpha_N^\mu)[\cdot]$ is uniformly bounded, as specified in Assumption~\ref{assum:regular functional}. Boundedness permits a nonstudentized bootstrap approximation without the need to compute any variance estimator (otherwise, the nonstudentized bootstrapped distribution may not serve as a good approximation, see also, e.g., Theorem 5.2.(2) in \cite{chen2015sieve}).
			
\begin{assumption}[Bounded Riesz representer]\label{assum:regular functional} 
For the given functional $\tau(\cdot)$ 
, there exists a finite constant $C_\tau>0$ for all chosen grids and penalty $\mu_N$, such that $\|v_N^\tau\|_\mu<C_\tau.$
\end{assumption}
Assumption~\ref{assum:regular functional} and Lemma~\ref{lem:bias bound} imply that any $b_N>0$ satisfying $\sqrt{\mu_N}=o(b_N)$ bounds $|\tau(\alpha_N^0)-\tau(\alpha_N^\mu)|$ up to an $o(N^{-1/2})$ remainder under the conditions stated below. A sharper option is available when the functional places little weight on directions whose curvature is generated primarily by the penalty, and we discuss one such case in Assumption \ref{assume:regular functional II}. 
			
Let $v_N^\tau =(\delta^\tau,\theta^{\tau'}\mathbf{1}_{S_N}(\cdot))$ be the corresponding Riesz representer of $D\tau(\alpha_N^\mu)[\cdot]$, $\gamma_\tau=(\delta^{\tau'},\theta^{\tau'})'=(H_N^\mu)^{-1}D_N^\mu,$ and $H_N^\mu=\sum_iw_{N,i}h_{N,i}h_{N,i}'$ be the spectral decomposition, with eigenvalues $w_{N,i}$ in descending order. 
The next assumption requires $D_N^\mu$ to place negligible weight on the lower spectral band $H_N^\mu$, that is, on directions whose curvature is close to the $\ell_2$-induced lower edge $\mu_N$.
			
\begin{assumption}[Negligible weight on the lower spectral band]\label{assume:regular functional II}
For every $i$, $w_{N,i} \ge\mu_N$, and there exist a fixed $\eta>0$ such that  
$$\frac{N}{\mu_N}\sum_{i: |w_{N,i} - \mu_N|\le (1+\eta)\mu_N }
{(h_{N,i}'D_N^\mu)^2} = O(1).$$
\end{assumption}
In the well-conditioned case, all eigenvalues of $H_N^\mu$ are uniformly bounded away from zero. The lower spectral band is then empty for all sufficiently large $N$ whenever $\mu_N=o(1)$, so Assumption~\ref{assume:regular functional II} holds automatically. More generally, the assumption ensures that the sampling variance $\sigma_N^2$ is sufficiently large relative to the $\ell_2$ component $\mu_N\|\theta^\tau\|_e^2$, leading to the following result.

\begin{lemma}\label{lem:bound1}
Suppose that $\epsilon_N^*\xi_N=o(N^{-1/2})$, Assumptions~\ref{assum:ri}, \ref{assum:regular functional}, \ref{assume:regular functional II}, \ref{assum:average_functional} and \ref{assum:riesz representer estimator I}.(v) hold, $\|v_N^\tau\|_{sd,\mu}^2=\|v_N^\tau\|_{\mu}^2+O(N^{-1})$ and $\mu_N=o(1)$, then $$\sqrt{2\mu_N}\sigma_N\ge|\tau(\alpha_N^0)-\tau(\alpha_N^\mu)|+o(N^{-1/2}).$$
\end{lemma}
			
\begin{proof}
See Appendix~\ref{Proof of Lemma ref{lem:bound1}}.
\end{proof}

Lemma~\ref{lem:bound1} bounds the regularization bias in the functional up to an asymptotically negligible remainder by $\sqrt{2\mu_N}$ times its asymptotic standard deviation. Without calculating the Riesz representer case by case, we show below that $\sigma_N$ can be consistently estimated by bootstrap under conditions stated in Corollary \ref{coro}, and the result also yields a feasible bias adjustment without requiring direct estimation of the distance between the penalized and unpenalized pseudo-true parameters.

We use the multinomial bootstrap to estimate $\sigma_N$. In each bootstrap iteration, observations are reweighted by random weights $(\omega_{n,N})_{n=1}^N$ drawn from the multinomial distribution such that
\[
  (\omega_{n,N})_{n=1}^N
  \sim \operatorname{Multinomial}
  \left(N; N^{-1},\ldots,N^{-1}\right).
\]
 See, for example, Assumption Boot.2 in \citet{chen2015sieve}.  For a realization $\omega=(\omega_{1,N},\ldots,\omega_{N,N})'$, define
\[
\widehat Q_N^{\mu,B}(\alpha)= \frac1N\sum_{n=1}^N \omega_{n,N}\ell^0_N(z_n,\alpha) +\frac{\mu_N}{2}\|\theta_\alpha\|_e^2.
\]
Let $\widehat\alpha_N^{\mu,B}$ minimize $\widehat Q_N^{\mu,B}$ over $\calA_N$, and let $s^B$ denote the conditional bootstrap standard deviation of $\sqrt{N}\{\tau(\widehat{\alpha}_N^{\mu,B}) -\tau(\widehat{\alpha}_N^\mu)\}$.
			
			\begin{corollary}\label{coro}
				Suppose that Theorems \ref{theo:3.1 v3}, \ref{theo:bootstrap}, Lemmas \ref{lem:bias bound}, \ref{lem:bound1} hold, then
				\[
				\liminf_{N\to\infty}
				P\left(
				\tau(\alpha_N^0) \in
				\left[
				\tau(\widehat{\alpha}_N^\mu)
				\pm (\Phi^{-1}(1-q/2) N^{-1/2}s^B+\sqrt{2\mu_N} s^B)
				\right]
				\right)
				\ge 1-q.
				\]
			\end{corollary}
			\begin{proof}
				See Appendix \ref{proof of Corollary {coro}}.
			\end{proof}
			Next, we demonstrate the performance of the proposed intervals in the above corollary in the following two sections.

\section{Simulation study}\label{sec:monte_carlo}
		
We conduct a Monte Carlo study to evaluate the finite-sample performance of the proposed confidence-conservative intervals for different average functionals. The data are generated from an RC logit model in which each individual $i$ chooses among $J=3$ mutually exclusive alternatives and an outside option. We use a linear utility function $u_{i,j}=x_{i,j,1}\beta_{i,1}+x_{i,j,2}\beta_{i,2}+\varepsilon_{i,j}$ to model the utility individual $i$ derives from alternative $j$, where $x_{i,j,k}$ are observed covariates, $\beta_i=(\beta_{i,1},\beta_{i,2})'$ are individual-specific RCs whose distribution we aim to estimate from the data, and $\varepsilon_{i,j}$ are unobserved error terms. Under the assumption that $\varepsilon_{i,j}$ are i.i.d.\ type-I extreme value, the model implies multinomial logit choice probabilities conditional on the individual-specific coefficients. We draw the covariates $x_{i,j,k}$ independently from $\mathcal{U}[0,3]$ across individuals, alternatives, and attributes $k\in\{1,2\}$.

To assess whether the proposed intervals attain nominal coverage in settings where standard parametric approaches fail, the RCs are drawn from a two-component mixture of bivariate normal distributions,
\begin{equation*}
0.5\,\mathcal{N}\!\big((-6,6)',\Omega\big)+0.5\,\mathcal{N}\!\big((-2,2)',\Omega\big),\qquad \Omega=
\begin{pmatrix}
0.7 & 0.15\\
0.15 & 0.7
\end{pmatrix}.
\end{equation*}	
We estimate this RC distribution using both the unpenalized fixed-grid estimator (FKRB) and the penalized fixed-grid estimator (ENet). We consider sample sizes $N\in\{1000, \, 10000\}$ and repeat each design $M=1000$ times. Because the true RC distribution is continuous and bimodal, it requires a relatively dense grid to accurately approximate the underlying distribution.
We therefore estimate the model using a uniform grid with $R_N\in\{25,49,81,289\}$ grid points over the domain $[-8,-0.1]\times[0.1,8]$ which covers the support of the true distribution with coverage probability close to one.\footnote{As a robustness check, we also consider a substantially larger grid domain, $[-10,0]\times[0,10]$, placing many grid points in regions with near-zero probability. This is close to the special scenario, placing grids outside the support, used to verify Assumption \ref{assum:ri} in Appendix \ref{appen:verification of assum 1}. The corresponding results are reported in Appendix \ref{app:tables_simulation}.
This additional robustness check leads to similar conclusions as the baseline specification: the ENet estimator continues to outperform FKRB, and denser grids generally improve finite-sample performance.}
 As the grid becomes denser, however, neighboring support points generate increasingly similar choice probabilities, leading to severe multicollinearity in the fixed-grid regression. This design, therefore, provides a natural setting to evaluate whether regularization improves both estimation and inference.
Figure \ref{fig:small_large_grid} in Appendix \ref{app:tables_simulation} illustrates the grid for $R_N=25$ and $R_N=81$ support points together with contour lines of the \st{true} mixture distribution.

For the fixed-grid estimators, we use the weighting matrix as suggested by Section 5 of \cite{fox2011}. 
For the ENet estimator, the regularization parameter is selected separately in each estimation step by five-fold cross-validation over 101 candidate values generated with the R package \texttt{glmnet}, including $\mu_N=0$ to allow for the unpenalized FKRB estimator.

For each Monte Carlo replication, we compute plug-in estimates of two functionals: the mean of the first RC,
\[
\tau_{\beta_1,0}
=
\int \beta_1\,dF_0(\beta),
\]
and the own-price elasticity evaluated at the mean covariate vector $\bar{x}=(1.5,1.5)$,
\[
\tau_{\eta,0}
=
\int
\left\{
\frac{\bar{x}_{j,k}}{P_j(\bar{x})}
\frac{\partial g_j(\bar{x},\beta)}{\partial x_{j,k}}
\right\}
dF_0(\beta).
\]
Confidence intervals for FKRB and ENet estimators are calculated using the results in Corollary \ref{coro} with $\mu_N=0$ for the FKRB estimator. We use a multinomial block bootstrap with $B=500$ bootstrap replications to estimate $s^B$. 

As a benchmark, we additionally estimate a parametric RC logit model that assumes a correlated multivariate normal distribution for the RCs. Confidence intervals for the parametric estimator are constructed using a percentile bootstrap with the same resampling design. The parametric model is estimated using the R package \texttt{logitr} \citep{logitr} with $S=250$ simulation draws. All confidence intervals are constructed at the 95\% confidence level. 

We evaluate the finite-sample performance in terms of the average absolute relative bias of the plug-in estimators and the empirical coverage of the confidence intervals. Let $\tau_{\beta_1,0}$ and $\tau_{\eta,0}$ denote the true population values of the corresponding functionals implied by the continuous mixture distribution, and let $\hat{\tau}^{(m)}$ denote the corresponding plug-in estimator in Monte Carlo replication $m=1,\ldots,M$. The average absolute relative bias is computed as
\[
\widehat{\mathrm{Bias}}
=
\frac{1}{M}
\sum_{m=1}^{M}
\left|
\frac{\hat{\tau}^{(m)}-\tau_0}{\tau_0}
\right|,
\]
whereas the empirical coverage is computed as
\[
\widehat{\mathrm{Cov}}
=
\frac{1}{M}
\sum_{m=1}^{M}
\mathbbm{1}
\left\{
\tau_0 \in CI^{(m)}
\right\},
\]
where $CI^{(m)}$ denotes the confidence interval obtained in replication $m$. Since the objective in practice is to recover the true population functionals of the underlying continuous RC distribution rather than the pseudo-true functionals implied by a finite-grid approximation, both bias and coverage are evaluated relative to the average functionals evaluated at the true RC distribution.\\

\begin{figure}[h]
\centering
\caption{Estimated distributions of the standardized plug-in estimators of $\tau_{\eta,0}$.}
\label{fig:kde_ela}
\includegraphics[width=0.9\linewidth]{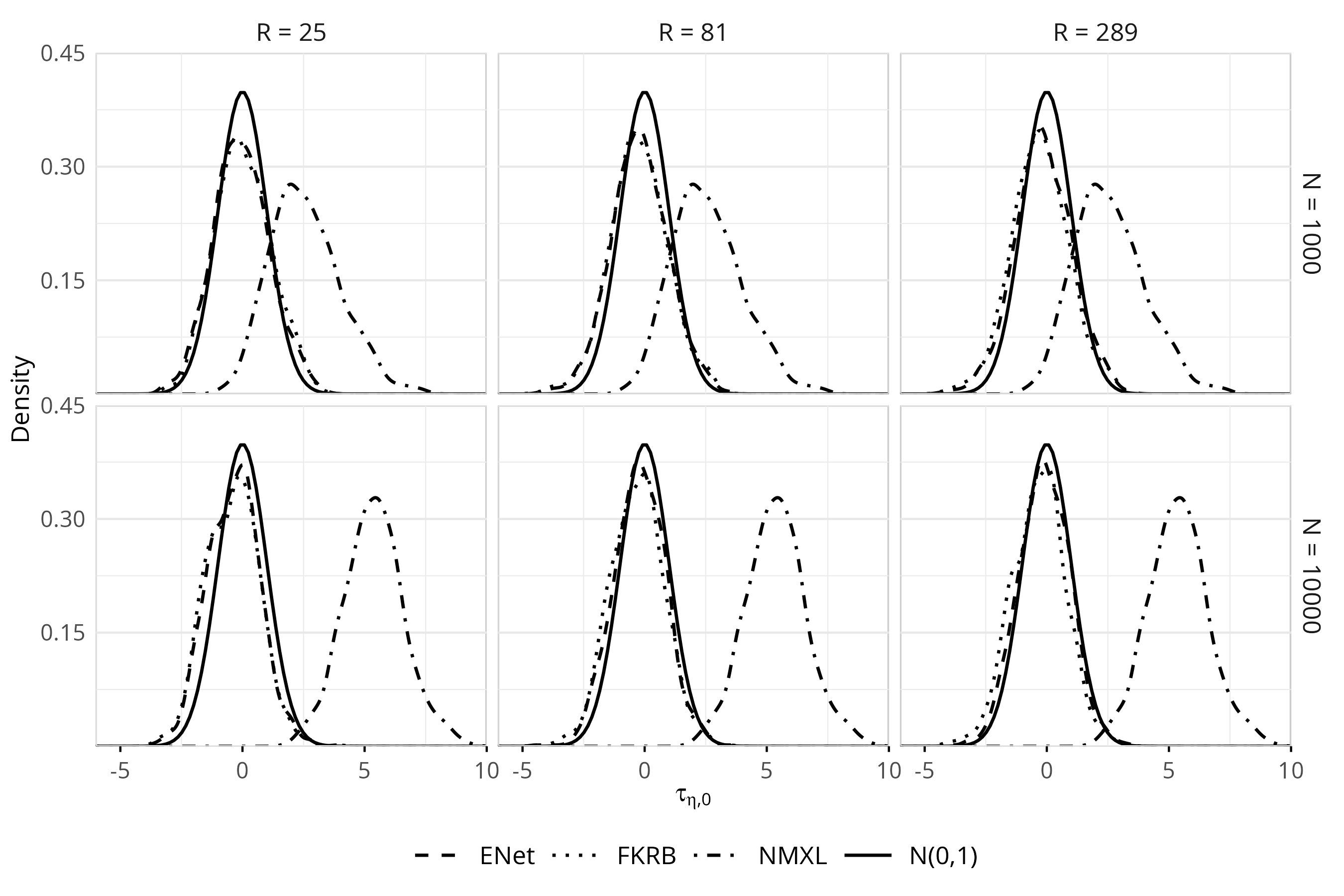}
\vspace{-0.1em}
\captionsetup{justification=justified, singlelinecheck=false, width=0.9\textwidth, font={footnotesize, stretch=0.9}}
\caption*{\textit{Notes:} The figure displays kernel density estimates of the standardized plug-in estimates of $\tau_{\eta,0}$ using the $1000$ Monte Carlo replications. The solid line shows the standard normal density. Panels correspond to different sample sizes ($N$) and numbers of grid points ($R$), using the grid domain $[-8,-0.1]\times[0.1,8]$.}
\end{figure}

Figure \ref{fig:kde_ela} plots the kernel density estimates of the standardized plug-in estimators for the own-elasticity at the mean against the standard normal distribution. The corresponding results for the mean RC are reported in Figure \ref{fig:kde_beta1} in Appendix \ref{app:tables}. For the fixed-grid estimators, and in particular for the ENet estimator, the standardized statistics are centered close to zero and resemble the shape of a standard normal distribution, especially for $N=10000$. In contrast, the standardized NMXL estimator is clearly shifted away from zero. This indicates that imposing an incorrect parametric distribution for the RCs leads to substantial finite-sample bias in the estimated functionals.

Table \ref{tab:mc_results_1} reports the average relative bias, the empirical coverage of nominal 95\% confidence intervals, and the average standard errors across 1000 Monte Carlo replications for the mean RC, \(\tau_{\beta_1,0}\), and the own-elasticity at the mean, \(\tau_{\eta,0}\). 
Consistent with the severe bias of the parametric estimator displayed in Figure~\ref{fig:kde_ela}, the corresponding confidence intervals substantially undercover the true value. The problem becomes even more pronounced as the sample size increases because the standard errors shrink while the specification bias remains unchanged. These findings illustrate that imposing an incorrect parametric distribution for the RCs can lead to severely misleading inferential conclusions, highlighting the importance of nonparametric estimators.

\begin{table}[h]
\centering
\caption{Average bias and coverage of 95\% confidence interval.}
\label{tab:mc_results_1}
\centering
\small
\setlength{\tabcolsep}{10pt}
\renewcommand{\arraystretch}{1.3}
\begin{tabular}[t]{rrccccccccc}
\toprule\toprule
\multicolumn{2}{c}{ } & \multicolumn{9}{c}{Mean RC ($\tau_{\beta_1,0}$)} \\
\cmidrule(l{3pt}r{3pt}){3-11}
& & \multicolumn{3}{c}{$\widehat{\mathrm{Bias}}$} & \multicolumn{3}{c}{$\widehat{\mathrm{Cov}}$} & \multicolumn{3}{c}{$s.e.(\hat{\tau})$}  \\
\cmidrule(l{3pt}r{3pt}){3-5} \cmidrule(l{3pt}r{3pt}){6-8} \cmidrule(l{3pt}r{3pt}){9-11}
N & R & NMXL & FKRB & ENet & NMXL & FKRB & ENet & NMXL & FKRB & ENet\\
\midrule
1000 & 25 & 0.139 & 0.108 & 0.099 & 0.662 & 0.878 & 0.891 & 0.393 & 0.402 & 0.346\\
1000 & 49 & -- & 0.122 & 0.101 & -- & 0.847 & 0.886 & -- & 0.437 & 0.343\\
1000 & 81 & -- & 0.121 & 0.096 & -- & 0.863 & 0.912 & -- & 0.434 & 0.337\\
1000 & 289 & -- & 0.128 & 0.091 & -- & 0.840 & 0.936 & -- & 0.441 & 0.322\\
\addlinespace
10000 & 25 & 0.136 & 0.067 & 0.062 & 0.012 & 0.685 & 0.730 & 0.116 & 0.171 & 0.161\\
10000 & 49 & -- & 0.054 & 0.048 & -- & 0.864 & 0.880 & -- & 0.194 & 0.168\\
10000 & 81 & -- & 0.058 & 0.044 & -- & 0.837 & 0.893 & -- & 0.196 & 0.160\\
10000 & 289 & -- & 0.061 & 0.040 & -- & 0.829 & 0.929 & -- & 0.201 & 0.153\\
\addlinespace\midrule
& & \multicolumn{9}{c}{Mean Own-Elasticity ($\tau_{\eta,0}$)} \\
\cmidrule(l{3pt}r{3pt}){3-11}
& & \multicolumn{3}{c}{$\widehat{\mathrm{Bias}}$} & \multicolumn{3}{c}{$\widehat{\mathrm{Cov}}$} & \multicolumn{3}{c}{$s.e.(\hat{\tau})$}   \\
\cmidrule(l{3pt}r{3pt}){3-5} \cmidrule(l{3pt}r{3pt}){6-8} \cmidrule(l{3pt}r{3pt}){9-11}
N & R & NMXL & FKRB & ENet & NMXL & FKRB & ENet & NMXL & FKRB & ENet \\
\midrule
1000 & 25 & 0.188 & 0.103 & 0.088 & 0.367 & 0.909 & 0.920 & 0.335 & 0.455 & 0.385\\
1000 & 49 & -- & 0.113 & 0.089 & -- & 0.893 & 0.921 & -- & 0.482 & 0.374\\
1000 & 81 & -- & 0.112 & 0.087 & -- & 0.899 & 0.931 & -- & 0.484 & 0.367\\
1000 & 289 & -- & 0.114 & 0.083 & -- & 0.893 & 0.949 & -- & 0.490 & 0.350\\
\addlinespace
10000 & 25 & 0.180 & 0.046 & 0.042 & 0.001 & 0.919 & 0.929 & 0.142 & 0.193 & 0.182\\
10000 & 49 & -- & 0.047 & 0.041 & -- & 0.923 & 0.933 & -- & 0.213 & 0.181\\
10000 & 81 & -- & 0.050 & 0.039 & -- & 0.917 & 0.943 & -- & 0.215 & 0.174\\
10000 & 289 & -- & 0.051 & 0.037 & -- & 0.912 & 0.959 & -- & 0.221 & 0.164\\
\bottomrule\bottomrule
\end{tabular}
\vspace{-0.8em}
\captionsetup{position=bottom, justification=justified, singlelinecheck=false, width=0.98\textwidth,font={footnotesize, stretch=0.9}}
\caption*{\textit{Notes:} Grid points for the nonparametric fixed-grid estimators of \citet{fox2011} (FKRB) and \citet{heiss2022} (ENet) are generated from a uniform grid on $[-8,-0.1]\times[0.1,8]$. The ENet tuning parameter is selected by 5-fold cross-validation, minimizing the mean squared error. The parametric estimator with normally distributed random coefficients (NMXL) is estimated by simulated maximum likelihood using 250 Halton draws. Confidence intervals for the mean random coefficient are based on a multinomial block bootstrap with 500 replications, while those for the mean own-elasticity use a nonparametric block bootstrap with 250 replications. 
All are calculated based on 1000 Monte Carlo replicates.}
\end{table}
 
Compared with the misspecified parametric estimator, both nonparametric estimators lead to substantially lower bias and more reliable inference. However, their performance depends critically on the density of the grid. With a coarse grid ($R=25$), approximation bias remains substantial, resulting in confidence intervals that still undercover, particularly for the mean RC. As the grid becomes denser, the ENet estimator becomes increasingly accurate, and the corresponding confidence intervals achieve coverage close to the desired nominal 95\% level for both functionals for $R=289$ -- while the confidence intervals for the mean RC still exhibit slight undercoverage (0.936 for $N=1000$ and 0.929 for $N=10000$), the corresponding intervals for the own-elasticity attain coverage essentially equal to the nominal level. Hence, these results demonstrate that sufficiently dense grids are essential for an accurate approximation and , therefore, reliable inference about the true parameter when applying the fixed-grid estimator.

The comparison with the FKRB estimator highlights the importance of regularization. While FKRB also reduces the bias relative to the parametric estimator, its bias does not decrease systematically as the grid becomes denser. Consequently, the corresponding confidence intervals continue to undercover for both functionals across all sample sizes. This behavior is consistent with the ill-conditioned problem discussed in Section \ref{sec:bias}. As the grid becomes denser, neighboring grid points become increasingly collinear, causing the FKRB estimator to suffer from selection inconsistency and inaccurate estimation of the probability weights: the estimator assigns positive probability mass to grid points outside the true distribution's support while setting weights within the true support to zero, preventing it from fully exploiting the improved approximation afforded by a denser grid \citep{heiss2022}.
In contrast, the additional \(\ell_2\)-penalty stabilizes the ENet estimator, 
leading to a coverage close to the 95\% level when the regularization bias is properly handled.

Across all simulation designs, the ENet estimator consistently exhibits smaller standard errors than the FKRB estimator. Moreover, the precision of the ENet estimator improves as the grid becomes denser, whereas the opposite pattern is observed for FKRB. This demonstrates that the additional \(\ell_2\)-regularization effectively mitigates the multicollinearity induced by dense grids

\begin{figure}[h]
\centering
\caption{Average 95\% confidence interval.}
\label{fig:ci_ela_beta}
\begin{subfigure}[t]{0.495\textwidth}
\includegraphics[width=\linewidth]{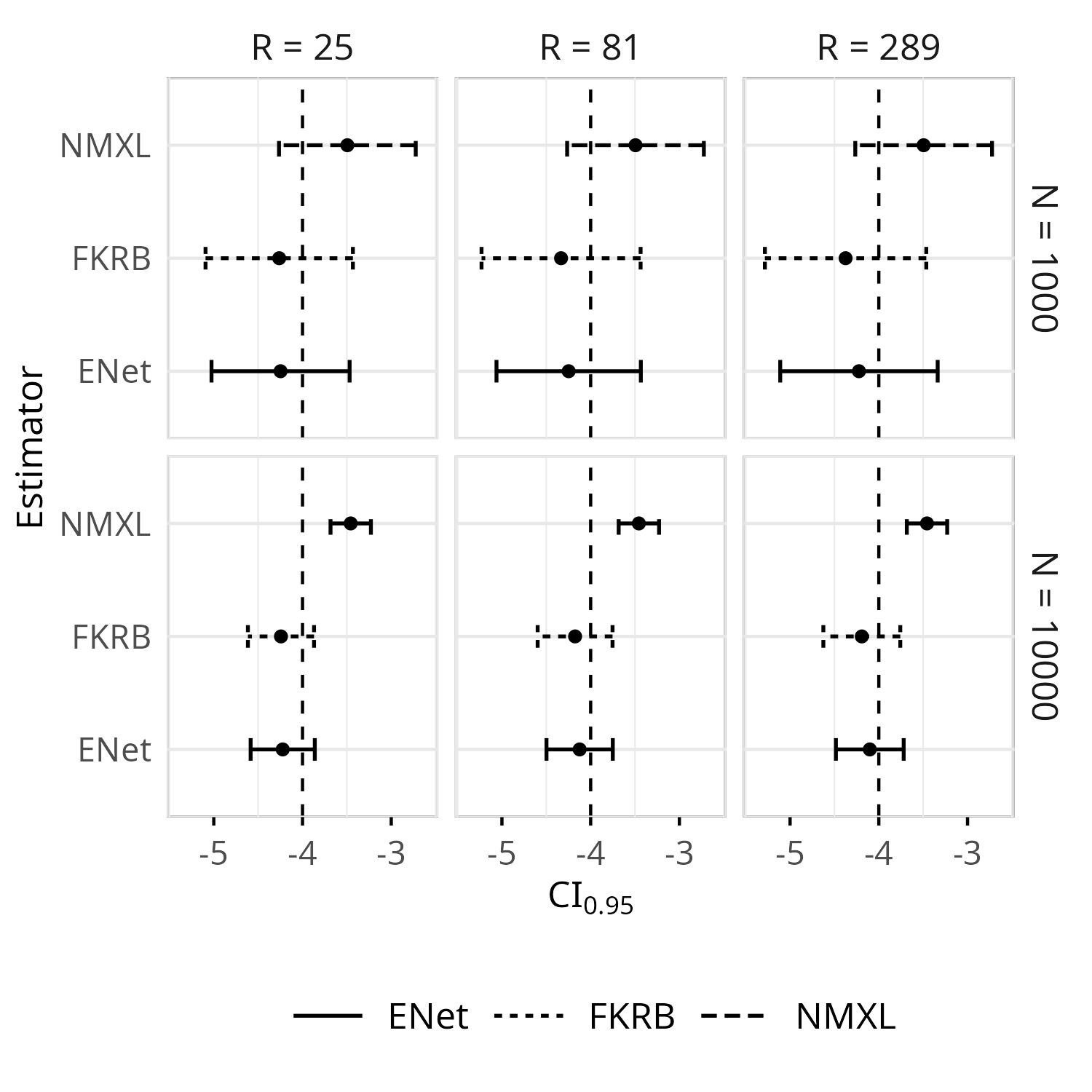}
\caption{Mean RC ($\tau_{\beta_1,0}$)}
\end{subfigure}\hfill
\begin{subfigure}[t]{0.495\textwidth}
\includegraphics[width=\linewidth]{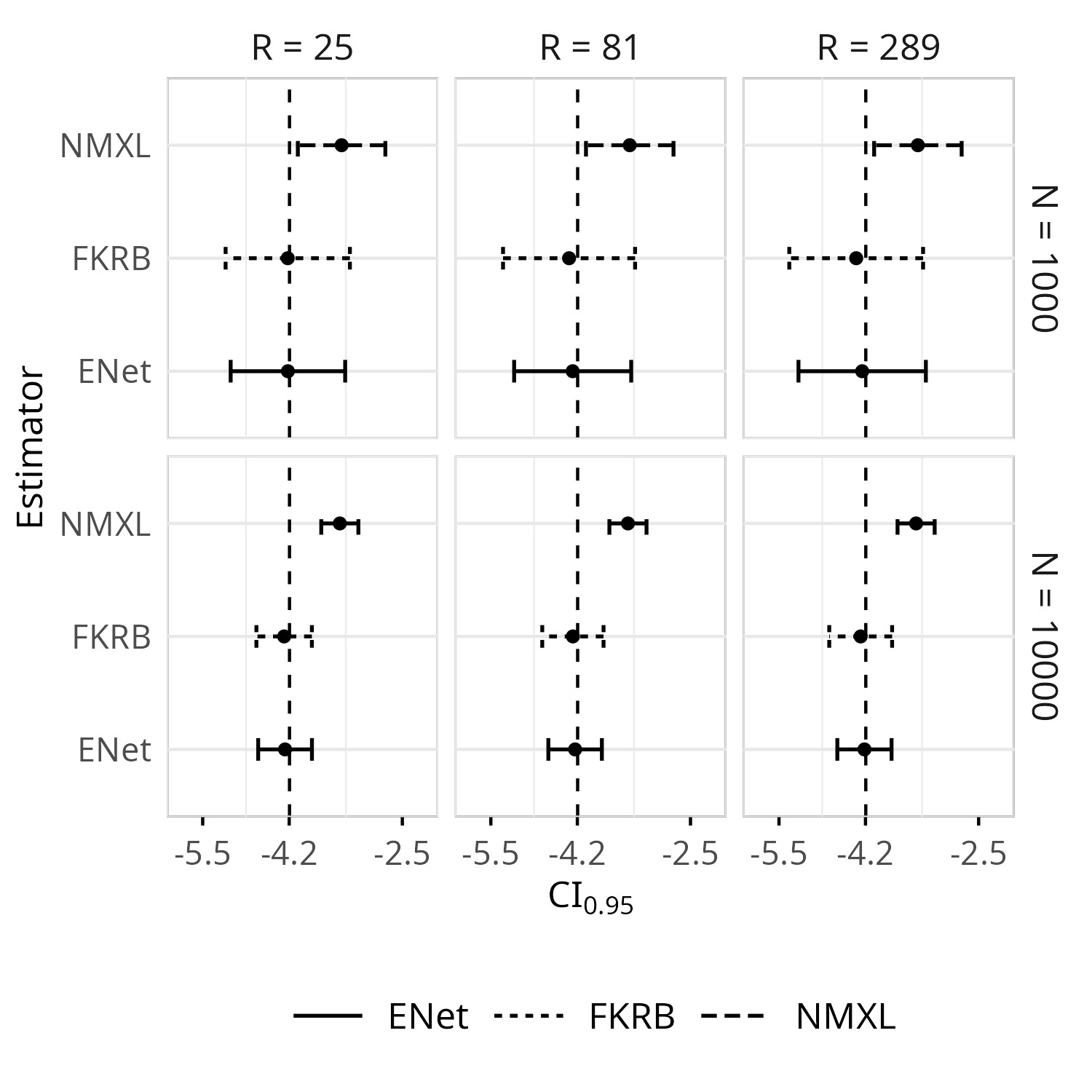}
\caption{Mean Own-Elasticity ($\tau_{\eta,0}$)}
\end{subfigure}
\vspace{-0.1em}
\captionsetup{justification=justified, singlelinecheck=false, width=0.99\textwidth, font={footnotesize, stretch=0.9}}
\caption*{\textit{Notes:} The figure displays the average 95\% confidence intervals over $1000$ Monte Carlo replications. The confidence-conservative intervals are constructed using a multinomial block bootstrap with equal weights and $500$ iterations. The conservative intervals for the parametric approach are calculated using a nonparametric block bootstrap with $500$ iterations. The vertical dotted lines correspond to the values. }
\end{figure}

Finally, Figure \ref{fig:ci_ela_beta} reports the average 95\% confidence intervals across Monte Carlo replicates for the mean RC (left panel) and the own-elasticity (right panel). As the comparison with the FKRB and parametric estimators reveals, the bias correction does not lead to overly conservative confidence intervals. Across all simulation designs, the ENet intervals are consistently narrower than the corresponding FKRB intervals and only moderately wider than the parametric NMXL intervals, despite the substantially greater flexibility of the nonparametric specification. Moreover, the ENet confidence intervals become substantially narrower as the sample size increases, reflecting that the estimator becomes more efficient with increasing sample size. In addition, the intervals also decrease slightly in width as the grid becomes denser, indicating that the additional $\ell_2$-regularization successfully controls the multicollinearity induced by dense grids while allowing the estimator to benefit from the improved approximation of the RC distribution.

		
\section{Application}\label{sec:application}
		
To illustrate the proposed inference procedure, we revisit the stated-preference travel-choice data analyzed by \citet{meijer2006}.\footnote{The data are available via the \emph{Journal of Applied Econometrics} \href{https://journaldata.zbw.eu/dataset/measuring-welfare-effects-in-models-with-random-coefficients}{Data Archive}.} 
Using survey data on travel-mode choices, they study travelers’ value of time (VoT) -- their willingness to pay for reduced travel time -- under normal, log-normal, and gamma specifications of the RC distribution in RC logit models. The VoT corresponds to the willingness-to-pay functional in Example~\ref{ex:wtp}, with the RC for travel time in the numerator and that for ticket fare in the denominator.
They find that the implied VoT distribution is highly sensitive to the assumed parametric specification, making this dataset well suited for illustrating the benefits of our nonparametric inference procedure.

The data originate from a stated-preference survey conducted for the Dutch railway company NS in 1987, in which respondents repeatedly chose between two hypothetical train connections that differed in fare (Euro), travel time (minutes), number of interchanges (0--4), and comfort level ($\{0,1,2\}$, where $2$ denotes the lowest comfort). The sample comprises 235 respondents who completed, on average, 12.5 choice tasks, yielding 2929 observed choices. To keep the empirical illustration simple, we treat these choice tasks as independent cross-sectional observations and abstract from the panel structure.\footnote{Extensions of the fixed-grid estimator to panel data are straightforward using Section 5 of \citet{fox2011}. Moreover, we verified that accounting for the panel structure has little effect on the corresponding parametric estimates.}

We specify respondent $i$'s utility from train connection $j\in\{1,2\}$ as
\begin{equation}\notag
u_{i,j}=\text{fare}_{i,j}\beta_{i,\text{fare}}+\text{time}_{i,j}\beta_{i,\text{time}}+\text{interchanges}_{i,j}\delta_{\text{int}}+\text{comfort}_{i,j}\delta_{\text{comf}}+\varepsilon_{i,j},
\end{equation}
where $\beta_i=(\beta_{i,\text{fare}},\beta_{i,\text{time}})'$ contains the RCs on fare and travel time, $\delta=(\delta_{\text{int}},\delta_{\text{comf}})'$ contains fixed coefficients on interchanges and comfort, and $\varepsilon_{i,j}$ are i.i.d.\ Type~I extreme value errors. In contrast to \citet{meijer2006}, we treat only the fare and travel-time coefficients as random. Allowing these two coefficients to vary jointly enables us to flexibly estimate the VoT distribution -- our primary object of interest -- while keeping the RC vector two-dimensional enables to specify a sufficiently dense grid.\footnote{The total number of grid points grows exponentially with the number of RCs. Allowing additional coefficients to be random would make a sufficiently dense grid infeasible given the available sample size.} To assess the implications of restricting heterogeneity to the fare and travel-time coefficients, we compare the parametric log-normal benchmark used in our analysis with the corresponding four-dimensional log-normal specification considered by \citet{meijer2006}. The resulting estimates are very similar, supporting the use of two RCs in both our parametric benchmark and nonparametric specifications.

We estimate the joint RC distribution using the penalized and unpenalized fixed-grid estimators with $R\in\{25,81,289\}$ uniformly spaced support points. The grid domain is determined from preliminary estimates of a parametric bivariate normal model. For each $k\in\{\text{fare},\text{time}\}$, the lower bound is $\underline{\beta}_k=\widehat{\mu}_k-4\widehat{\sigma}_k$, where $\widehat{\mu}_k$ and $\widehat{\sigma}_k$ denote the estimated mean and standard deviation under the parametric normal model. The upper bounds are $\overline{\beta}_{\text{fare}}=-0.1$ and $\overline{\beta}_{\text{time}}=0$, imposing negative marginal utility of fare and travel time. 

The model is estimated using the iterative algorithm of \citet{heiss2022}, which alternates between updating the probability weights via a linear probability model estimated with constrained least squares and updating the fixed coefficients via a weighted conditional logit model until convergence. 
In line with Section \ref{sec:monte_carlo}, we follow \citet{fox2011} for the construction of the weighting matrix. The ENet tuning parameter is selected by five-fold cross-validation over 101 candidate values generated by \texttt{glmnet}, including zero. We construct 95\% confidence intervals for the mean RCs and the implied mean VoT using the proposed confidence-conservative procedure with a multinomial block bootstrap based on 500 replications.

As a benchmark, we estimate a parametric model in which the fare and travel-time coefficients follow a correlated bivariate log-normal distribution. We focus on the log-normal specification because it provided the best fit among the parametric models considered by \citet{meijer2006} while ensuring a finite mean VoT. The model is estimated by simulated maximum likelihood using 250 Halton draws per respondent. Since the ratio of two jointly log-normal variables is itself log-normal, we calculate the mean VoT from its closed-form mean using the estimated mean vector and covariance matrix of the underlying bivariate log-normal distribution.\footnote{Since the coefficients are jointly log-normal, the logarithm of their ratio is normally distributed. Let $\mu_{\text{time}}$ and $\mu_{\text{fare}}$ denote the means of the log-coefficients and $\sigma_{\text{time}}^2$, $\sigma_{\text{fare}}^2$, and $\sigma_{\text{tf}}$ the corresponding variances and covariance. The mean VoT is then given by $\mathbb{E}[VoT]=-\exp\!\left((\mu_{\text{time}}-\mu_{\text{fare}})+1/2(\sigma_{\text{time}}^2+\sigma_{\text{fare}}^2-2\sigma_{\text{tf}})
\right)$.} The corresponding confidence intervals are obtained via a block bootstrap with 500 replications.

Table \ref{tab:results_application} reports the estimated parameters and implied mean VoT estimates for the parametric benchmark, FKRB, and ENet estimators. For the nonparametric fixed-grid estimators, we report the results for \(R=289\), which provides the best fit to the data in terms of the log-likelihood among the considered grid specifications. The corresponding results for \(R=25\) and \(R=81\) are reported in Table \ref{tab:application_extended} in Appendix \ref{app:tables}.\footnote{The results show that the estimated means and standard deviations are relatively robust across grid specifications for both FKRB and ENet, while the absolute mean VoT decreases as $R$ increases, with more pronounced differences for FKRB than for ENet.}

\begin{table}[h]
\centering
\caption{Estimated RC moments and VoT across estimators.}
\label{tab:results_application} 
\small
\setlength{\tabcolsep}{7pt}
\renewcommand{\arraystretch}{1.3}
\begin{tabular}{lccccccccc}
\toprule\toprule
& \multicolumn{3}{c}{Log-N} & \multicolumn{3}{c}{FKRB} & \multicolumn{3}{c}{ENet} \\
\cmidrule(l{3pt}r{3pt}){2-4} \cmidrule(l{3pt}r{3pt}){5-7} \cmidrule(l{3pt}r{3pt}){8-10}
& Coef & se & CI$_{0.95}$ & Coef & se & CI$_{0.95}$ & Coef & se & CI$_{0.95}$\\
\midrule
\multicolumn{10}{l}{\textit{mean RCs}}  \\[0.2em]
Fare & -1.31 & 0.15 & {}[-1.60, -1.01] & -1.56 & 0.17 & {}[-1.90, -1.22] & -1.64 & 0.12 & {}[\: -1.98, -1.30]\\
Time & -2.67 & 0.48 & {}[-3.62, -1.73] & -7.81 & 0.90 & {}[-9.56, -6.05] & -8.29 & 0.71 & {}[-10.16, -6.42]\\
Comfort & -0.37 & 0.04 & {}[-0.43, -0.32] & -0.96 & 0.23 & {}[-1.28, -0.65] & -1.18 & 0.21 & {}[\: -1.56, -0.81]\\
Interchanges & -1.05 & 0.03 & {}[-1.12, -0.97] & -2.13 & 0.16 & {}[-2.57, -1.69] & -2.47 & 0.16 & {}[\: -3.02, -1.92]\\
\addlinespace
\multicolumn{10}{l}{\textit{standard deviation RCs}}  \\[0.2em]
Fare & 4.72 & -- & -- & 1.63 & -- & -- & 1.53 & -- & --\\
Time & 2.75 &  -- & -- & 9.55 & -- & -- & 9.07 & -- & --\\
\addlinespace
\multicolumn{10}{l}{\textit{value-of-time}}  \\[0.2em]
Mean & -9.36 & 4.75 & {}[-18.68, -0.05] & -16.25 & 6.21 & {}[-28.43, -4.07] & -25.77 & 4.64 & {}[-39.23, -12.31]\\
\bottomrule\bottomrule
\end{tabular}
\vspace{-0.8em}
\captionsetup{position=bottom, justification=justified, singlelinecheck=false, width=0.995\textwidth,font={footnotesize, stretch=0.9}}
\caption*{\textit{Notes:} The table reports the estimated parameter, standard errors (se), and the $95\%$ confidence-(conservative) intervals for the penalized (ENet) and unpenalized (FKRB) estimators for $R=289$ grid points, and the parametric estimator with a bivariate log-normal assumption (Log-N). We used a multinomial block bootstrap with $500$ iterations and equal weights to construct the confidence intervals. The tuning parameter for the ENet estimator is selected by $5$-fold cross-validation based on the lowest mean squared error. The parametric estimator uses $250$ Halton draws. }
\end{table}

The results reveal substantial differences between the parametric and nonparametric estimates of the RC distribution and the implied mean VoT estimates. The differences are relatively small for the mean fare coefficient, although the parametric specification implies substantially greater dispersion in fare sensitivities among respondents. In contrast, the differences are much more pronounced for the travel-time coefficient: both fixed-grid estimators imply substantially stronger preferences for travel-time reductions and markedly greater heterogeneity across respondents. As a consequence, the implied mean VoT for a one-hour reduction in travel time is considerably larger under the nonparametric specifications. While the parametric benchmark implies a mean VoT of \(9.36\) EUR per hour, the corresponding estimates equal \(16.25\) EUR for FKRB and \(25.77\) EUR for ENet.\footnote{\citet{meijer2006} report a mean VoT estimate of \(9.53\) EUR per hour for a specification in which all four coefficients are assumed to follow a correlated log-normal distribution. This similarity suggests that restricting heterogeneity to fare and travel time has only a minor effect on the implied mean VoT.}
		
The confidence intervals provide further information about the precision of the estimates across estimators. All reported coefficients and mean VoT estimates are statistically different from zero at the \(5\%\) level. The intervals under the nonparametric estimators are generally wider than those under the log-normal specification, reflecting the usual efficiency advantage of parametric estimation. Importantly, despite accounting for regularization bias, the ENet intervals are only moderately wider than the FKRB intervals because regularization reduces sampling variance. The proposed intervals therefore remain informative rather than being excessively conservative. Moreover, the log-normal and ENet point estimates of mean VoT lie outside each other's confidence intervals. Although this does not constitute a formal test of equality, together with the simulation evidence that parametric misspecification can generate persistent bias, it indicates that the VoT estimates are sensitive to the imposed distributional assumptions.

Figure \ref{fig:spike_tail_289} compares the VoT distributions implied by the parametric benchmark and the two nonparametric fixed-grid estimators for $R=289$ grid points to further motivate the difference in the estimated mean VoT between the three estimators. The plot groups the VoT grid points into ten equally spaced bins, where the reported mass corresponds to the sum of probability weights associated with all grid points inside the respective bin. In addition, we zoom in on the left tail of the distribution to examine those more closely and compare the parametric benchmark against the nonparametric estimators. Although all three estimators allocate considerable mass near zero, the estimated log-normal distribution decays much more rapidly as the VoT increases, whereas FKRB and ENet retain substantial mass at larger VoT values. Among the nonparametric estimators, FKRB produces a sparser distribution than ENet, as is particularly visible in the left-tail zoom. These thicker tails explain the differences in the estimated mean VoT reported in Table \ref{tab:results_application}. Figure \ref{fig:hist.wtp} illustrates that the results are robust across different numbers of grid points.

\begin{figure}[!t]
\centering
\caption{Estimated VoT distributions and their left-tails.}
\includegraphics[width=0.9\linewidth]{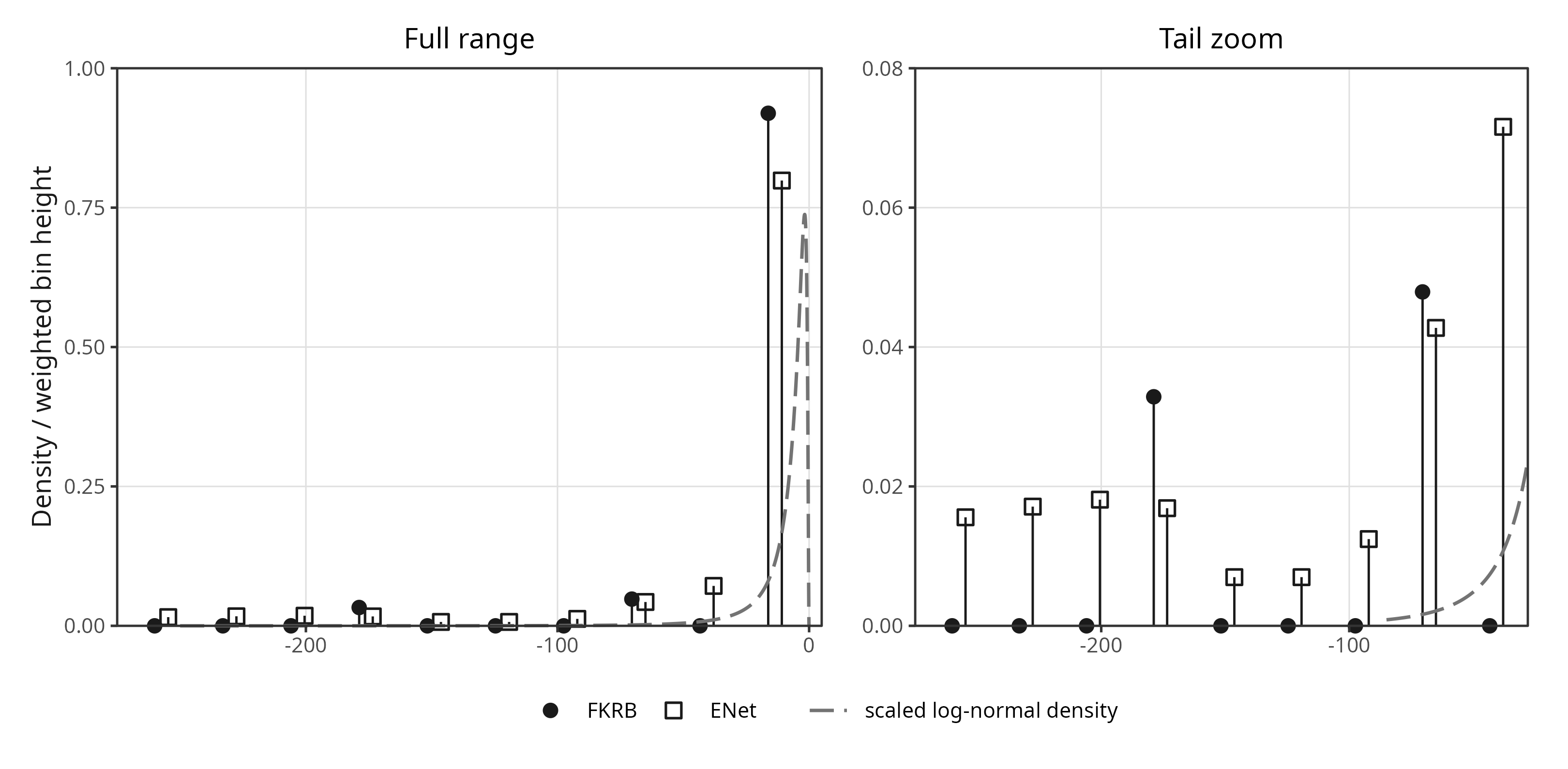}
\label{fig:spike_tail_289}
\captionsetup{justification=justified, singlelinecheck=false, width=0.9\textwidth, font={footnotesize, stretch=0.9}}
\vspace{-0.1em}
\caption*{\footnotesize \textit{Notes:} The distribution is shown for the parametric log-normal distribution and the nonparametric FKRB and ENet estimators for $R=289$ grid points. The spikes are created by 10 equally sized bins across the grid for the VoT. We overlay the spikes by the curve of the parametric log-normal distribution. Finally, we include a zoom for the left tail of the VoT distribution to study tail behavior in more detail.}
\end{figure}

\begin{figure}[!t]
\centering
\caption{Estimated joint distributions of RCs on travel time and fare.}
\includegraphics[width=\linewidth]{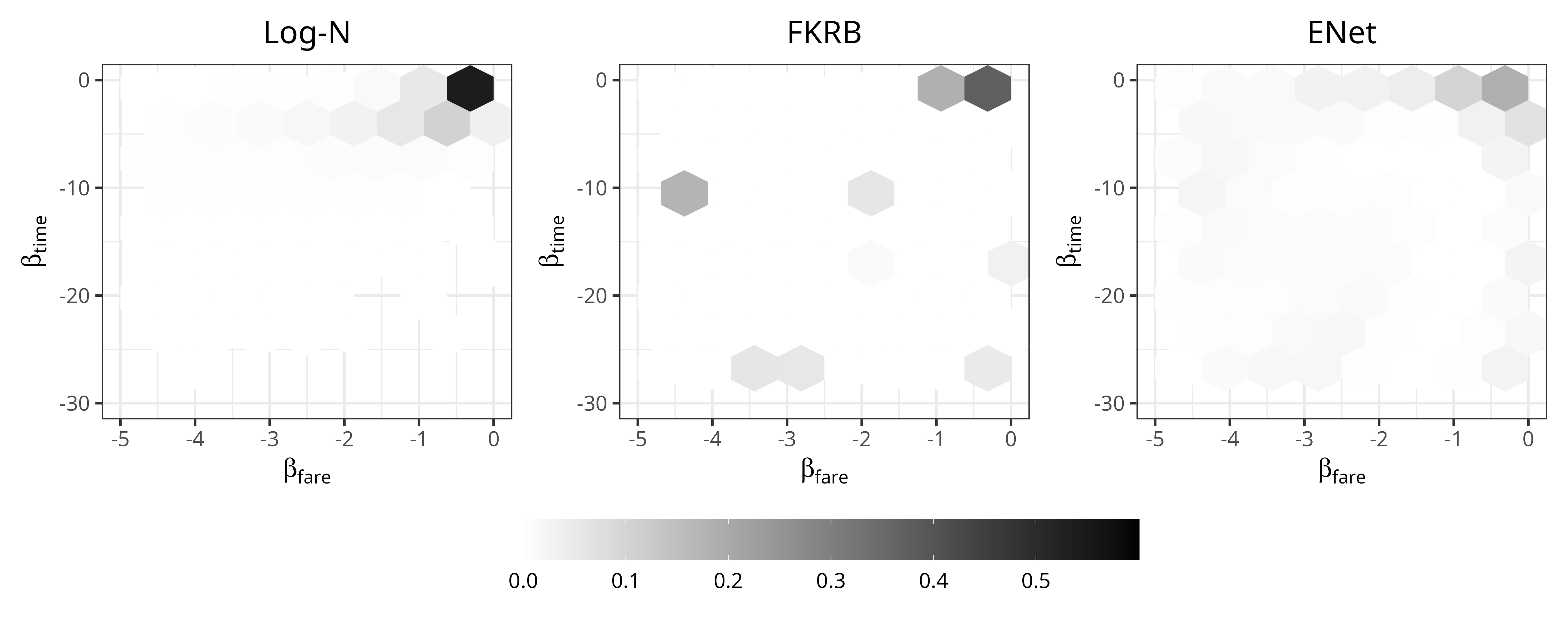}
\captionsetup{justification=justified, singlelinecheck=false, width=\textwidth, font={footnotesize, stretch=0.9}}
\vspace{-0.1em}
\caption*{\footnotesize \textit{Notes:} The heatmap for the parametric model with log-normal fare and time coefficients is created using $100 \ 000$ draws from the estimated log-normal distribution. For the non-parametric approaches, we used the specification with $R=289$ grid points, whereby each hexagon corresponds to the sum of weights falling into the two-dimensional bin. We used $8$ bins per dimension. }
\label{fig:heatmap.application}
\end{figure}

To better understand the sources of the differences in the estimated VoT distributions, Figure \ref{fig:heatmap.application} displays the underlying joint distributions of the fare and travel-time coefficients. Similarly to the VoT distribution in Figure \ref{fig:spike_tail_289}, the heatmap is constructed using eight equally spaced bins per dimension. We compare the parametric log-normal distribution with FKRB and ENet for $R=289$ grid points.

Under the parametric specification, the coefficients are more strongly positively correlated, so small travel-time coefficients tend to be accompanied by small fare coefficients, leading to a large share of travelers with small VoT values. In contrast, FKRB and ENet recover a non-negligible share of travelers with small fare and large travel-time coefficients, and hence larger VoT values in absolute terms. Consistent with this pattern, the estimated correlations are $0.560$, $0.413$, and $0.291$ for Log-N, FKRB, and ENet, respectively (see Table \ref{tab:correlation.application} in Appendix \ref{app:tables}). These differences explain the thicker VoT tails and larger mean VoT estimates reported in Table \ref{tab:results_application}. Figure \ref{fig:heatmap_all} shows that the qualitative findings are robust across different numbers of grid points.

		
\section{Conclusion}\label{sec:conclusion}

Nonparametric estimators of RC models are increasingly used to avoid restrictive distributional assumptions about unobserved preference heterogeneity. Among these methods, the fixed-grid estimator of \citet{bajari2007} and \citet{fox2011} is particularly attractive because it is computationally tractable and easy to implement. Accurately approximating the RC distribution requires a sufficiently dense grid, which, however, generates severe multicollinearity and unstable estimates of the probability weights. The penalized fixed-grid estimator of \citet{heiss2022} addresses this problem through $\ell_2$-regularization, but the resulting gains in stability come at the cost of regularization bias.

This paper develops bias-aware inference for economically relevant functionals based on the penalized fixed-grid estimator. We derive the asymptotic distribution of the penalized plug-in estimator and a feasible upper bound on the resulting regularization bias. Combining these components yields confidence intervals that remain valid when regularization bias is non-negligible. The procedure applies to scalar functionals commonly reported in empirical applications, including mean RCs, elasticities, willingness-to-pay measures, and predicted choice probabilities.

The Monte Carlo simulations demonstrate the importance of addressing both regularization and sieve approximation. When the grid is sufficiently dense, the proposed intervals achieve coverage close to the nominal level while remaining informative in finite samples. In contrast, a coarse grid generates persistent approximation bias and substantial undercoverage, even as the sample size increases. The simulations therefore show that regularization is important for making estimation on dense grids sufficiently stable, while the proposed bias adjustment is needed to account for the bias introduced by that regularization. In addition, our travel-mode application shows that the nonparametric estimators imply substantially thicker tails of the VoT distribution and considerably larger mean VoT estimates than the parametric log-normal specification. At the same time, the bias-aware ENet intervals are only moderately wider or even narrower than the FKRB intervals, indicating that accounting for regularization bias does not result in excessively conservative inference.

Several extensions provide promising directions for future research. First, the simulation results highlight the importance of the grid domain and resolution. Developing data-driven methods that jointly select the location and number of grid points would be valuable. And second, the proposed intervals account for regularization bias but require the sieve approximation bias to be sufficiently small for inference on the true population functional. An interesting extension would be to construct confidence intervals that explicitly account for sieve approximation bias, particularly in high-dimensional applications where the exponential growth of the grid limits a dense resolution in each dimension.

		
		
\bibliography{ref.bib} 
\newpage

		
\begin{appendices}
\appendixnumbering
\section{Additional tables \& figures}\label{app:tables}

\subsection{Additional tables \& figures for Section \ref{sec:monte_carlo}}\label{sec:tables.mc}

\begin{figure}[h]
\centering
\caption{Uniform grids on $[-8,-0.1]\times[0.1,8]$ with contours of the two-component bivariate normal mixture.}
\label{fig:small_large_grid}
\begin{subfigure}[t]{0.49\textwidth}
\centering
\includegraphics[width=0.8\linewidth]{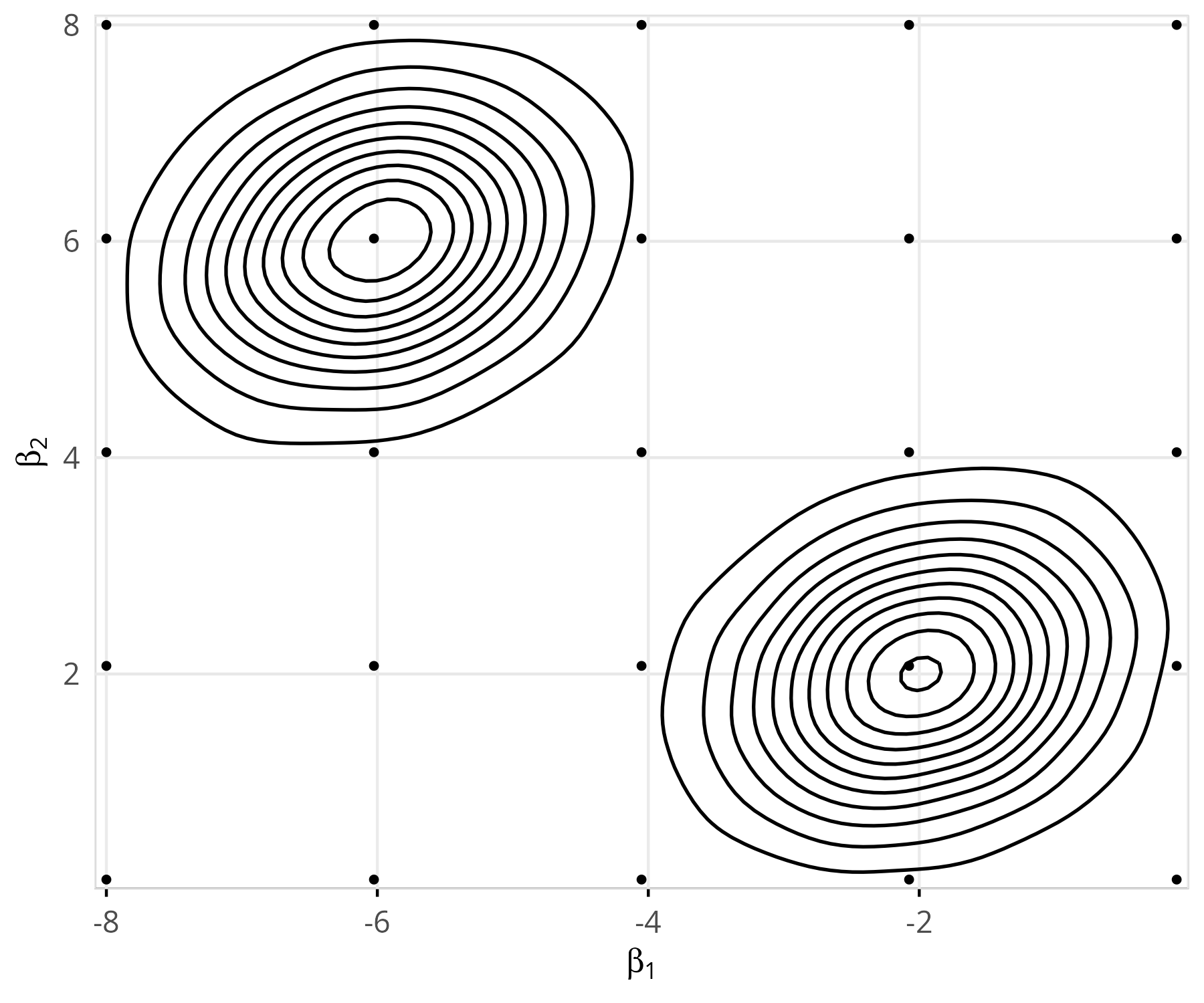}
\caption{$R = 25$}
\end{subfigure}
\begin{subfigure}[t]{0.49\textwidth}
\centering
\includegraphics[width=0.8\linewidth]{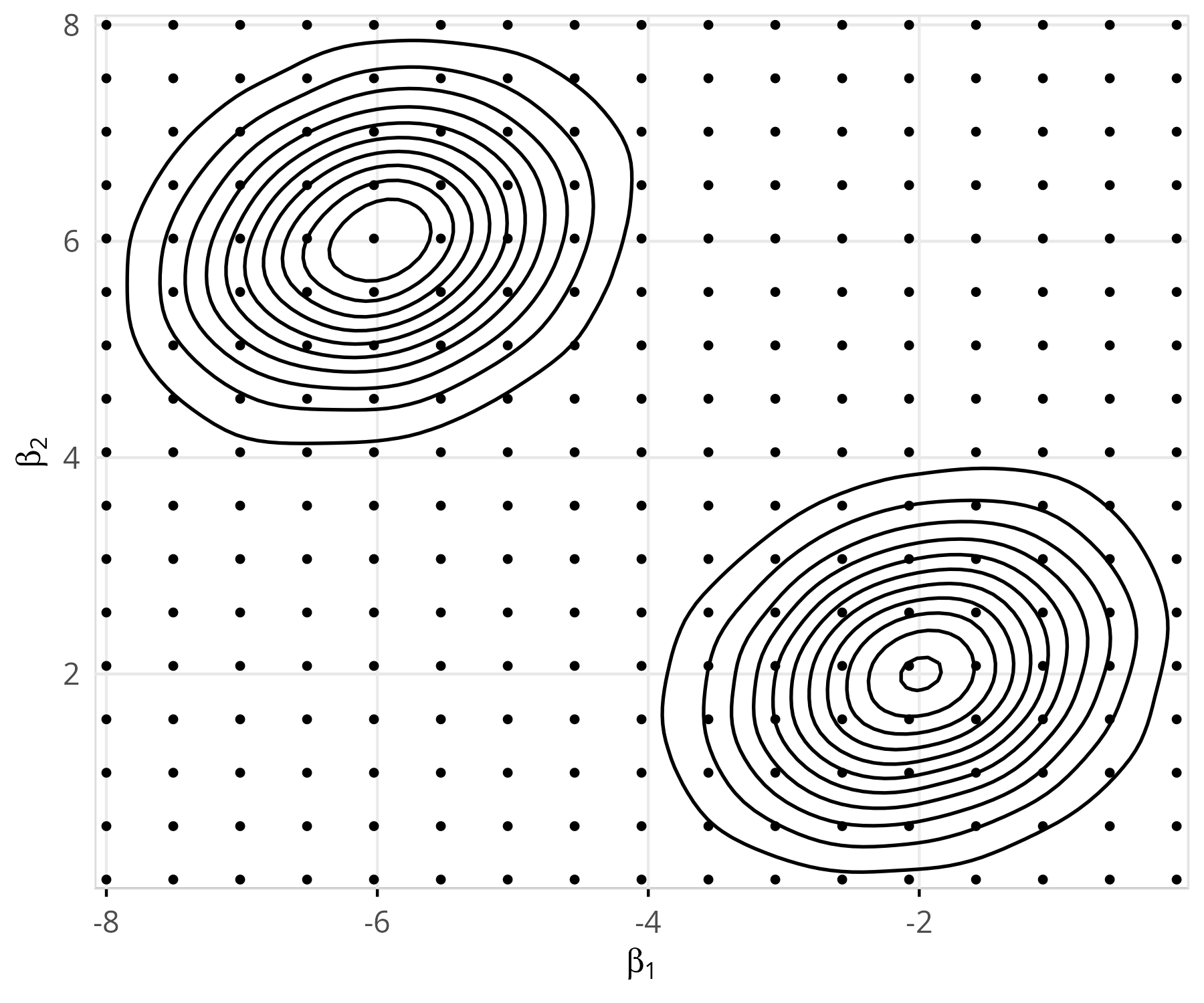}
\caption{$R = 289$}
\end{subfigure}
\end{figure}

\begin{figure}[H]
\centering
\caption{Monte Carlo Simulated distributions of the standardized plug-in estimator of $\tau_{\beta_1,0}$.}
\label{fig:kde_beta1}
\includegraphics[width=0.9\linewidth]{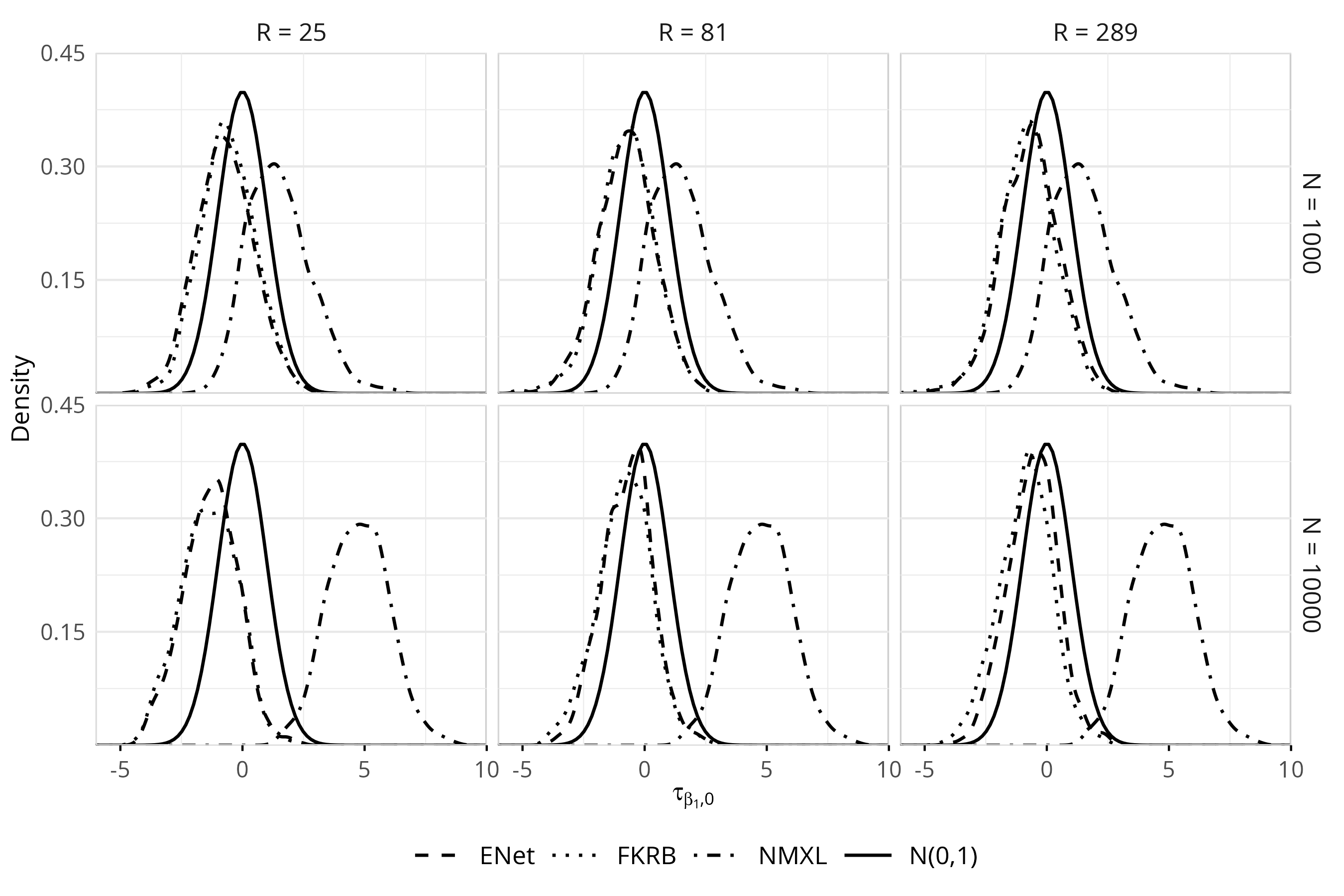}
\vspace{0.1em}
\captionsetup{justification=justified, singlelinecheck=false, width=0.9\textwidth, font={footnotesize, stretch=0.9}}
\caption*{\textit{Notes:} The figure displays kernel density estimates of the standardized plug-in estimates of $\tau_{\beta_1,0}$ using the $1000$ Monte Carlo replications. The solid line shows the standard normal density. Panels correspond to different sample sizes ($N$) and numbers of grid points ($R$), using the grid domain $[-8,-0.1]\times[0.1,8]$.}
\end{figure}

\subsubsection{Simulation study: sensitivity to the grid domain}\label{app:tables_simulation}

\begin{figure}[H]
\centering
\caption{Uniform grids on $[-10,-0.1]\times[0.1,10]$ with contours of the two-component bivariate normal mixture.}
				\label{fig:large_grid_appendix}
				\begin{subfigure}[t]{0.495\textwidth}
					\centering
					\includegraphics[width=0.8\linewidth]{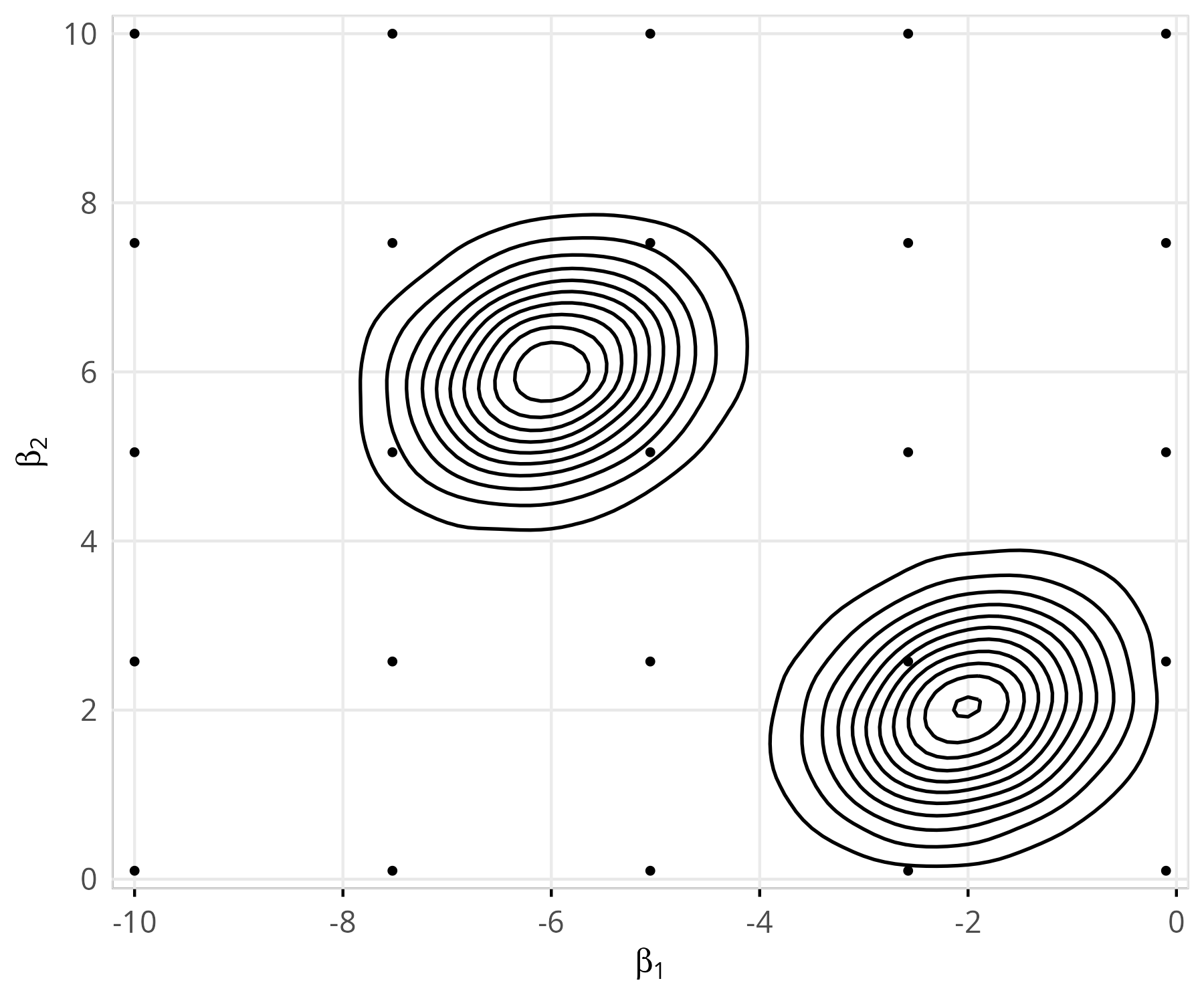}
					\caption{$R = 25$}
				\end{subfigure}
				\begin{subfigure}[t]{0.495\textwidth}
					\centering
					\includegraphics[width=0.8\linewidth]{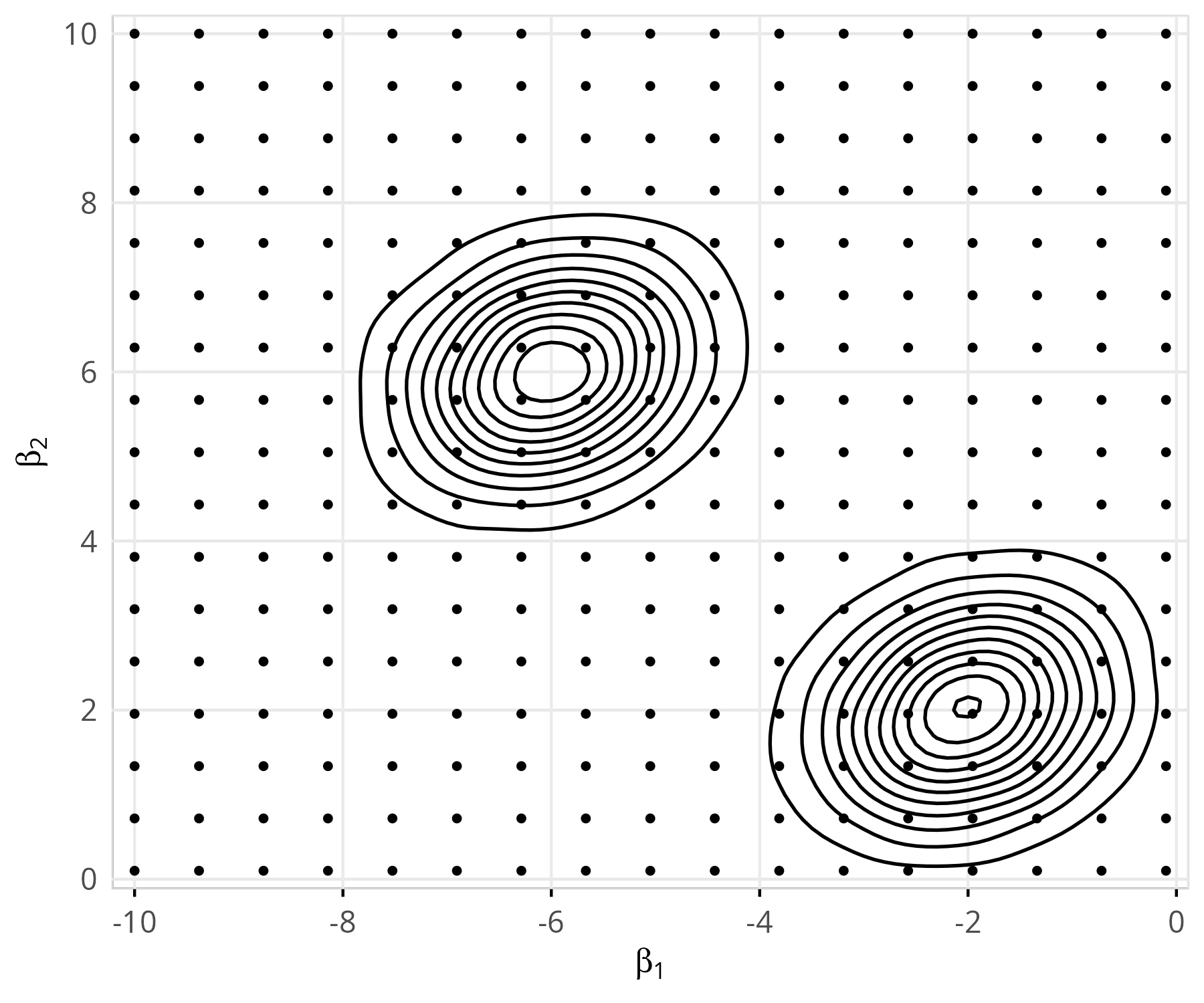}
					\caption{$R = 289$}
				\end{subfigure}
			\end{figure}

To assess the sensitivity of the nonparametric estimators to the choice of the grid domain, we repeat the Monte Carlo experiment using an enlarged grid domain given by \([-10,0]\times[0,10]\). Relative to the baseline specification, this larger domain allocates a substantial share of grid points to regions where the true mixture distribution has little probability mass. Figure \ref{fig:large_grid_appendix} illustrates the resulting grids for \(R=25\) and \(R=289\) together with contour lines of the true RC distribution. All remaining aspects of the simulation design, including the sample sizes, bootstrap procedure, and tuning-parameter selection, are identical to those used in the baseline specification discussed in Section \ref{sec:monte_carlo}. Table \ref{tab:mc_results_2} reports the corresponding Monte Carlo results.

\begin{table}[H]
\centering
\caption{Average bias and coverage of 95\% confidence interval under extended grid.}
\label{tab:mc_results_2}
\centering
\small
\setlength{\tabcolsep}{10pt}
\renewcommand{\arraystretch}{1.05}
\begin{tabular}[t]{rrccccccccc}
\toprule\toprule
\multicolumn{2}{c}{ } & \multicolumn{9}{c}{Mean RC ($\tau_{\beta_1,0}$)} \\
\cmidrule(l{3pt}r{3pt}){3-11}
& & \multicolumn{3}{c}{$\widehat{\mathrm{Bias}}$} & \multicolumn{3}{c}{$\widehat{\mathrm{Cov}}$} & \multicolumn{3}{c}{$s.e.(\hat{\tau})$}  \\
\cmidrule(l{3pt}r{3pt}){3-5} \cmidrule(l{3pt}r{3pt}){6-8} \cmidrule(l{3pt}r{3pt}){9-11}
N & R & NMXL & FKRB & ENet & NMXL & FKRB & ENet & NMXL & FKRB & ENet\\
\midrule
1000 & 25 & 0.191 & 0.116 & 0.116 & 0.336 & 0.922 & 0.926 & 0.333 & 0.584 & 0.538\\
1000 & 49 & - & 0.139 & 0.127 & - & 0.905 & 0.930 & - & 0.621 & 0.537\\
1000 & 81 & - & 0.143 & 0.129 & - & 0.902 & 0.936 & - & 0.647 & 0.535\\
1000 & 289 & - & 0.149 & 0.125 & - & 0.893 & 0.953 & - & 0.660 & 0.519\\
\addlinespace
10000 & 25 & 0.182 & 0.050 & 0.050 & 0.004 & 0.924 & 0.936 & 0.143 & 0.241 & 0.240\\
10000 & 49 & - & 0.068 & 0.067 & - & 0.902 & 0.912 & - & 0.293 & 0.272\\
10000 & 81 & - & 0.070 & 0.065 & - & 0.890 & 0.923 & - & 0.298 & 0.268\\
10000 & 289 & - & 0.076 & 0.064 & - & 0.875 & 0.943 & - & 0.313 & 0.260\\
\addlinespace\midrule
& & \multicolumn{9}{c}{Mean Own-Elasticity ($\tau_{\eta,0}$)} \\
\cmidrule(l{3pt}r{3pt}){3-11}
& & \multicolumn{3}{c}{$\widehat{\mathrm{Bias}}$} & \multicolumn{3}{c}{$\widehat{\mathrm{Cov}}$} & \multicolumn{3}{c}{$s.e.(\hat{\tau})$}   \\
\cmidrule(l{3pt}r{3pt}){3-5} \cmidrule(l{3pt}r{3pt}){6-8} \cmidrule(l{3pt}r{3pt}){9-11}
N & R & NMXL & FKRB & ENet & NMXL & FKRB & ENet & NMXL & FKRB & ENet \\
\midrule
1000 & 25 & 2.127 & 1.096 & 1.097 & 0.476 & 0.920 & 0.926 & 30.146 & 1.113 & 1.098\\
1000 & 49 & - & 0.821 & 0.827 &- & 0.966 & 0.979 & - & 1.024 & 1.007\\
1000 & 81 & - & 0.711 & 0.702 & - & 0.968 & 0.987 & - & 0.966 & 0.941\\
1000 & 289 & - & 0.597 & 0.559 & - & 0.967 & 0.993 & - & 0.896 & 0.835\\
\addlinespace
10000 & 25 & 3.691 & 0.791 & 0.789 & 0.441 & 0.264 & 0.289 & 4.551 & 0.388 & 0.387\\
10000 & 49 & - & 0.560 & 0.552 & - & 0.824 & 0.857 & -& 0.467 & 0.459\\
10000 & 81 & -& 0.436 & 0.433 & - & 0.907 & 0.941 &-& 0.467 & 0.457\\
10000 & 289 & - & 0.352 & 0.326 & - & 0.934 & 0.970 & - & 0.452 & 0.431\\
\bottomrule\bottomrule
\end{tabular}
\vspace{-0.8em}
\captionsetup{position=bottom, justification=justified, singlelinecheck=false, width=0.98\textwidth,font={footnotesize, stretch=0.9}}
\caption*{\textit{Notes:} Grid points for the nonparametric fixed-grid estimators of \citet{fox2011} (FKRB) and \citet{heiss2022} (ENet) are generated using a uniform grid on $[-10,-0.1]\times[0.1,10]$. Confidence intervals are constructed using a multinomial block bootstrap with $500$ iterations and equal weights. The tuning parameter for the ENet estimator is selected by $5$-fold cross-validation based on the lowest mean squared error. The parametric estimator with normal RCs (NMXL) is estimated by simulated maximum likelihood using $250$ Halton draws. Confidence intervals for the mean own-elasticity are constructed using a nonparametric block bootstrap with 250 iterations. All are calculated based on 1000 Monte Carlo replicates.}
			\end{table}


\subsection{Additional tables \& figures for Section \ref{sec:application}}       			
\begin{figure}[H]
\centering
\caption{Estimated distributions of VoT.}
\includegraphics[width=0.9\linewidth]{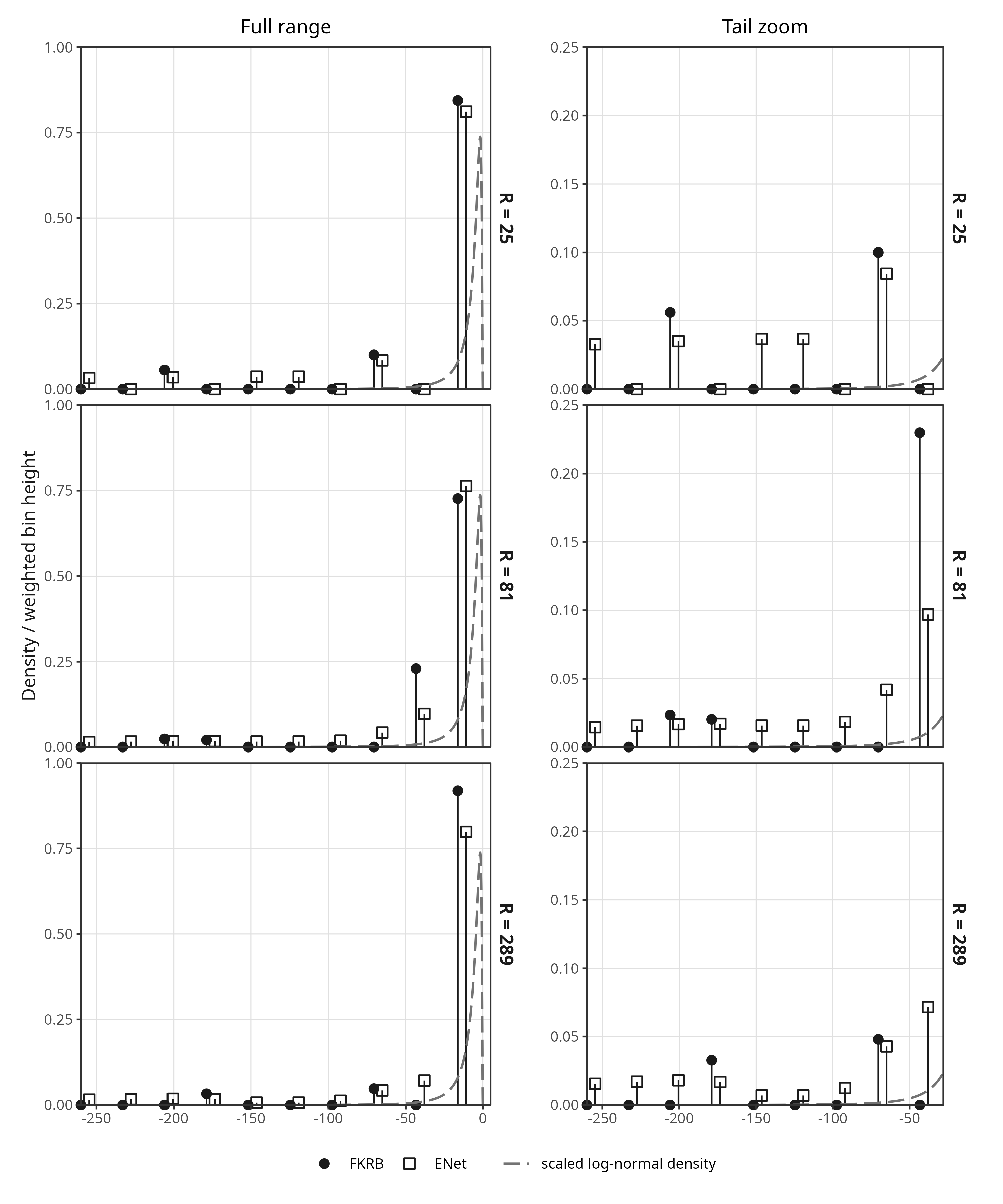}
\captionsetup{justification=justified, singlelinecheck=false, width=0.9\textwidth, font={footnotesize, stretch=0.9}}
\vspace{-0.1em}
\caption*{\footnotesize \textit{Notes:} The Figure shows the histogram for the FKRB and ENet estimator for $R=25,81,289$ grid points and the parametric log-normal distribution overlaid in dashed line. }
\label{fig:hist.wtp}
\end{figure}

\begin{figure}
\centering
\caption{Estimated joint distributions of random coefficients on travel time and fare.}
\includegraphics[width=0.8\linewidth]{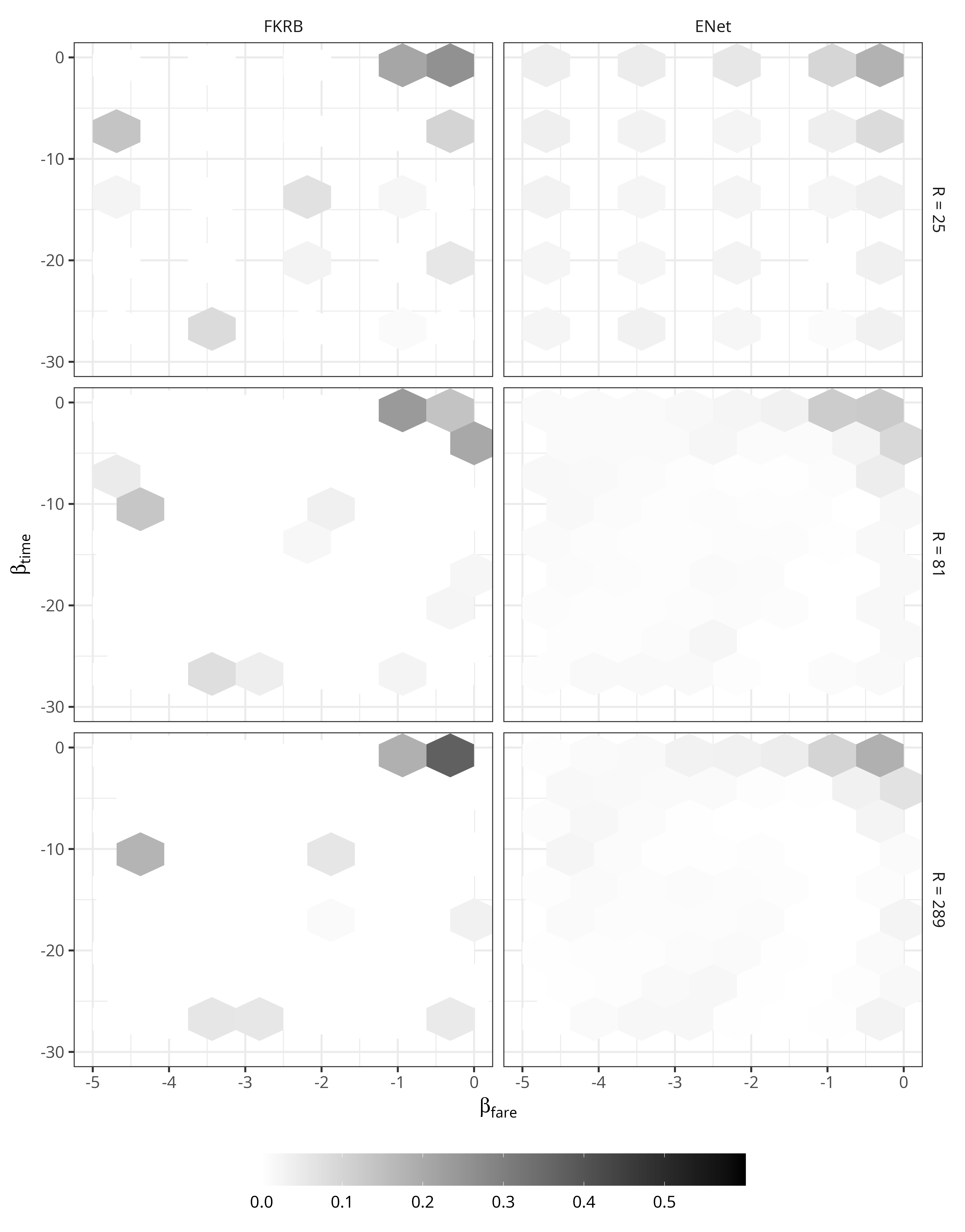}
\label{fig:heatmap_all}
\captionsetup{justification=justified, singlelinecheck=false, width=0.8\textwidth, font={footnotesize, stretch=0.9}}
\vspace{-0.1em}
\caption*{\footnotesize \textit{Notes:} The heatmap for the parametric model with log-normal fare and time coefficients is created using $100 \ 000$ draws from the estimated log-normal distribution. For the non-parametric approaches, we used the specification with $R=25,81,289$ grid points, whereby each hexagon corresponds to the sum of weights falling into the two-dimensional bin. We used $8$ bins per dimension.}
\end{figure}

\begin{table}[H]
\caption{Parameter estimates and VoT for different grid sizes.}
\label{tab:application_extended}
\centering
\small
\setlength{\tabcolsep}{5.5pt}
\renewcommand{\arraystretch}{1.3}
\begin{tabular}[t]{lccccccccc}
\toprule\toprule
        
					\multicolumn{10}{c}{{FKRB}} \\
                    \midrule
				\multicolumn{1}{c}{ } & \multicolumn{3}{c}{R=25} & \multicolumn{3}{c}{R=81} & \multicolumn{3}{c}{R=289} \\
\cmidrule(l{3pt}r{3pt}){2-4} \cmidrule(l{3pt}r{3pt}){5-7} \cmidrule(l{3pt}r{3pt}){8-10}
Parameter & Coef & se & CI$_{0.95}$ & Coef & se & CI$_{0.95}$  & Coef & se & CI$_{0.95}$ \\
\midrule
				\multicolumn{10}{l}{\textit{mean}}  \\[0.2em]
Fare & -1.54 & 0.17 & {}[-1.88, -1.2] & -1.55 & 0.17 & {}[-1.88, -1.22] & -1.56 & 0.17 & {}[-1.9, -1.22]\\
Time & -7.44 & 0.91 & {}[-9.22, -5.67] & -7.59 & 0.92 & {}[-9.39, -5.79] & -7.81 & 0.90 & {}[-9.56, -6.05]\\
Comfort & -0.99 & 0.22 & {}[-1.31, -0.67] & -0.95 & 0.23 & {}[-1.25, -0.65] & -0.96 & 0.23 & {}[-1.28, -0.65]\\
Interchanges & -2.17 & 0.17 & {}[-2.6, -1.74] & -2.11 & 0.15 & {}[-2.56, -1.66] & -2.13 & 0.16 & {}[-2.57, -1.69]\\
				\addlinespace
				\multicolumn{10}{l}{\textit{standard deviation}}  \\[0.2em]
				Fare & 1.59 & -- & -- & 1.64 & -- & -- & 1.63 & -- & --\\
				Time & 8.99 & -- & -- & 9.29 & -- & -- & 9.55 & -- & --\\
				\addlinespace
				\multicolumn{10}{l}{\textit{value-of-time}}  \\[0.2em]
				wtp& -20.30 & 6.79 & {}[-33.6, -7] & -17.87 & 6.89 & {}[-31.38, -4.35] & -16.25 & 6.21 & {}[-28.43, -4.07]\\
                \midrule
                
                LL & \multicolumn{3}{c}{-1674.81} & \multicolumn{3}{c}{-1672.62}  & \multicolumn{3}{c}{-1672.04}  \\
				\midrule\\[-0.5em]
					\multicolumn{10}{c}{{ENet}} \\
					\midrule
					\multicolumn{1}{c}{ } & \multicolumn{3}{c}{R=25} & \multicolumn{3}{c}{R=81} & \multicolumn{3}{c}{R=289} \\
					\cmidrule(l{3pt}r{3pt}){2-4} \cmidrule(l{3pt}r{3pt}){5-7} \cmidrule(l{3pt}r{3pt}){8-10}
					Parameter & Coef & se & CI$_{0.95}$  & Coef & se & CI$_{0.95}$  & Coef & se & CI$_{0.95}$ \\
					\midrule
					\multicolumn{10}{l}{\textit{mean}}\\
					Fare & -1.77 & 0.07 & {}[-1.96, -1.58] & -1.64 & 0.10 & {}[-1.91, -1.37] & -1.64 & 0.12 & {}[-1.98, -1.3]\\
Time & -8.83 & 0.57 & {}[-10.16, -7.49] & -8.15 & 0.69 & {}[-10.16, -6.14] & -8.29 & 0.71 & {}[-10.16, -6.42]\\
Comfort & -1.31 & 0.18 & {}[-1.73, -0.9] & -1.17 & 0.19 & {}[-1.53, -0.8] & -1.18 & 0.21 & {}[-1.56, -0.81]\\
Interchanges & -2.67 & 0.16 & {}[-3.2, -2.14] & -2.45 & 0.14 & {}[-2.9, -1.99] & -2.47 & 0.16 & {}[-3.02, -1.92]\\
					\addlinespace
					\multicolumn{10}{l}{\textit{standard deviation}}\\
					Fare & 1.60 & -- & -- & 1.55 & -- & -- & 1.53 & -- & --\\
					Time  & 9.55 & -- & -- & 9.05 & -- & -- & 9.07 & -- & --\\
					\addlinespace
					\multicolumn{10}{l}{\textit{value-of-time}} \\
					wtp  & -29.07 & 4.58 & {}[-41.19, -16.95] & -26.15 & 4.86 & {}[-39.05, -13.25] & -25.77 & 4.64 & {}[-39.23, -12.31]\\
                    \midrule
                    LL & \multicolumn{3}{c}{-1683.03} & \multicolumn{3}{c}{-1678.69}  & \multicolumn{3}{c}{-1679.09}  \\
					\bottomrule\bottomrule
				\end{tabular}
                \vspace{-0.8em}
\captionsetup{position=bottom, justification=justified, singlelinecheck=false, width=1\textwidth,font={footnotesize, stretch=0.9}}
\caption*{\textit{Notes:} The table reports the estimated parameter, standard errors (se), and the $95\%$ confidence-(conservative) intervals for the penalized (ENet) and unpenalized (FKRB) estimators for $R=25,81,289$ grid points. Confidence intervals are constructed using a multinomial block bootstrap with $500$ iterations and equal weights. The tuning parameter for the ENet estimator is selected by $5$-fold cross-validation based on the lowest mean squared error.}
\end{table}

\begin{table}[H]
\centering
\caption{Estimated standard deviations and time--fare correlations for the mean RC's.}
\centering
\small
\setlength{\tabcolsep}{10pt}
\renewcommand{\arraystretch}{1.3}
\begin{tabular}[t]{lccccccc}
\toprule\toprule
\multicolumn{1}{c}{ } & \multicolumn{1}{c}{Log-N} & \multicolumn{3}{c}{FKRB} & \multicolumn{3}{c}{ENet} \\
\cmidrule(l{3pt}r{3pt}){2-2} \cmidrule(l{3pt}r{3pt}){3-5} \cmidrule(l{3pt}r{3pt}){6-8}
Parameter &   & R=25 & R=81 & R=289 & R=25 & R=81 & R=289\\
\midrule
\multicolumn{8}{l}{\textit{standard deviation}} \\
Fare & 4.723 & 1.595 & 1.641 & 1.634 & 1.600 & 1.547 & 1.534\\
Time & 2.745 & 8.991 & 9.287 & 9.547 & 9.551 & 9.052 & 9.075\\
\addlinespace
\multicolumn{8}{l}{\textit{correlation}} \\
Fare-Time & 0.560 & 0.412 & 0.472 & 0.413 & 0.245 & 0.305 & 0.294\\
\bottomrule\bottomrule
\end{tabular}
				\label{tab:correlation.application}
                \vspace{-0.8em}
\captionsetup{position=bottom, justification=justified, singlelinecheck=false, width=0.75\textwidth,font={footnotesize, stretch=0.9}}
\caption*{\textit{Notes:} The table presents the estimated standard deviation and correlation of the corresponding RCs, calculated from the estimated RC distributions of the three approaches.

}
\end{table}

\section{Auxiliary results: proofs and verification of assumptions}\label{app:proofs}

\paragraph{Additional notation.} We introduce additional notations used in the section before we formally describe intermediate results and the associated assumptions.
For $\alpha =(\delta_\alpha, F_{\theta_\alpha}) \in \mathcal{A}_N$, let $\gamma_{\alpha} = (\vdelta_{\alpha}, \theta_{\alpha}) \in \mathcal{C} \times \mathbb{R}^{R_N}$.
Using this notation, we can rewrite $P(X,\alpha)$ as ${P}_{\gamma}(X,\gamma_{\alpha})$, and define ${\ell_{N,\gamma}^{\mu}}(\cdot),{Q_{N,\gamma}^{\mu}}(\cdot)=\mathbb{E}\! [{\ell_{N,\gamma}^{\mu}}(z,\cdot)]$ accordingly. In addition, let $T_{\mathcal{A}_N}({\alpha_N^\mu})$ denote the feasible tangent cone in $\mathcal{A}_N$ at ${\alpha_N^\mu}$, 
$$T_{\mathcal{A}_N}({\alpha_N^\mu}) = \{v_{\alpha}: \textit{~there exists a sequence of $\alpha_t\in \mathcal{A}_N$ such that~} v_{\alpha}=\lim_{t\downarrow 0} \frac{\alpha_t -\alpha^{\mu}_N}{t} \},$$ 
and by construction we see that $\alpha - \alpha^{\mu}_N \in T_{\mathcal{A}_N}({\alpha_N^\mu})$.
For $v=(\delta_{v}, \mathbf{1}_{N}(\cdot)'\theta_{v}) \in T_{\mathcal{A}_N}({\alpha_N^\mu})$,  let $h_{v} = (\delta_{v}, \theta_{v}) \in \mathcal{C} \times \mathbb{R}^{R_N}$,
and let
$$T_{\mathcal{A}_N,h}({\alpha}_{N}^\mu) \equiv \left\{ h_{v}  \in \mathcal{C} \times \mathbb{R}^{R_N} : v \in  T_{\mathcal{A}_N}({\alpha_N^\mu}) \right\}.$$  
The simplex constraints imply that
\begin{align*}
T_{\mathcal{A}_N,h}({\alpha_N^\mu}) = \left\{ (\delta_{v}, \theta_{v}) \in \mathcal{C} \times \mathbb{R}^{R_N} : \iota^{\prime} \theta_{v}=0, \theta_{v, r} \geq 0 \text { whenever } {\theta}_{\alpha^{\mu}_N, r}=0\right\},
\end{align*}
where $\theta_{v, r}$ and $\theta_{\alpha^{\mu}_N, r}$ are the $r$th entry of $\theta_{v}$ and $\theta_{\alpha^{\mu}_N}$ respectively, and $\iota$ is a $R_N \times 1$ vector with all entries equal to 1. For $v\in T_{\mathcal{A}_N}(\alpha^{\mu}_N)$, denote
\begin{equation*} \label{eq:inner.product 3}
\begin{aligned}& D_+{Q_{N}^{\mu}}(\alpha^{\mu}_N)[v] 
					=\lim _{t \downarrow 0}\frac{{Q_{N}^{\mu}}( \alpha^{\mu}_N+tv)-{Q_{N}^{\mu}}( \alpha^{\mu}_N)}{t},\\
					& D^2_+{Q_{N}^{\mu}}(\alpha^{\mu}_N)[v,v] =
					\lim _{t \downarrow 0}\frac{D_+{Q_{N}^{\mu}}(\alpha^{\mu}_N+tv)[v]-D_+{Q_{N}^{\mu}}(\alpha^{\mu}_N)[v]}{t}.
\end{aligned}
\end{equation*}
Equivalently, the above can also be represented via $h_{v}$ for $v\in T_{\mathcal{A}_N,\gamma}(\alpha_{\mu_N, N})$,
\begin{equation*}  
\begin{aligned}& D_+{Q_{N,\gamma}^{\mu}}(\gamma_{\alpha^{\mu}_N})[h] 
					=\lim _{t \downarrow 0}\frac{{Q_{N,\gamma}^{\mu}}( \gamma_{\alpha^{\mu}_N}+th)-{Q_{N,\gamma}^{\mu}}( \gamma_{\alpha^{\mu}_N})}{t},\\
					& D^2_+{Q_{N,\gamma}^{\mu}}(\gamma_{\alpha^{\mu}_N})[h,h] =
					\lim _{t \downarrow 0}\frac{D_+{Q_{N}^{\mu}}(\gamma_{\alpha^{\mu}_N}+th)[h]-D_+{Q_{N}^{\mu}}(\gamma_{\alpha^{\mu}_N})[h]}{t}.
\end{aligned}
\end{equation*}
Consider a class of real-valued measurable functions $\mathcal{G}$, and an $\epsilon$-bracket $[\mathrm{l}, \mathrm{u}]$ consisting of two functions $\mathrm{l}, \mathrm{u} \in \mathrm{L}_2({P})$ such that $\mathrm{l}(z) \leq f(z) \leq \mathrm{u}(z)$ for some $f \in \mathcal{G}$, and the $\mathrm{L}_2({P})$ distance satisfies $\|u-\mathrm{l}\|_{\mathrm{L}_2({P})} \equiv (\mathbb{E}(\mathrm{l}(z)-\mathrm{u}(z))^{2})^{1/2} < \epsilon$. The bracketing number $\mathbb{N}_{[]}(\epsilon, \mathcal{G},\|\cdot \|_{\mathrm{L}_2({P})})$ is the minimum number of such brackets required to cover any $f\in \mathcal{G}$ under the norm $\|\cdot \|_{\mathrm{L}_2({P})}$. We also use the letters $ c$ and $ C$ (with or without subscripts) to denote generic positive numbers, which may take different values depending on the context. To characterize the bootstrap, we adopt the discussions in Section 5 of \cite{chen2012} and consider $(\mathcal{Z}^{\infty} \times \Omega, \sigma(z^{\infty}) \times \sigma(\Omega), P_{\mathcal{Z}^{\infty} \times \Omega})$ with $\Omega=\left\{\omega_{n, N}: n=1, \ldots, N; N=1, \ldots\right\}$ being the space of weights defined as a triangle array with elements in $\mathbb{R}_+$, $\mathcal{Z}^{\infty}$ being the space of samples, $\sigma(z^{\infty}) \times \sigma(\Omega) $ and $ P_{\mathcal{Z}^{\infty} \times \Omega}$ being the corresponding $\sigma$-algebra and the joint probability over $\mathcal{Z}^{\infty} \times \Omega$. We use $P(\cdot)$ as the marginal distribution with respect to $z^{\infty}$ (so it is coherent with the notations used in the previous discussions), and use $P^*(\cdot)$ as an abbreviation of the bootstrap conditional distribution given the sample, i.e., $P^*(\cdot) \equiv P_{\mathcal{Z}^{\infty} \times \Omega |\mathcal{Z}^{\infty}}(\cdot \; \rvert \{z_n\}_{n=1}^N),$ and we say $X$ is $X= o_{p^*}(1)$ wpa1 ($P$ ), if, for any $\epsilon>0, P(P^*(\left|X\right|>\epsilon )>\epsilon) \rightarrow 0$ as $n \rightarrow \infty$ (note that $P^*(\left|X\right|>\epsilon )$ is a measurable function of $\{z_n\}_{n=1}^N$), and $X=O_{p^*}(1)$ wpa1 ($P$ ), if, for any $\epsilon>0,$ there exists $M$ such that $ P(P^*(\left|X\right|>M )>\epsilon) \rightarrow 0$ as $n \rightarrow \infty$. 
In addition, for two random variables $A_N,B_N$ measurable with the associated $\sigma$-algebra, let $\mathcal{L}^*(A_N),\mathcal{L}(B_N)$ be the conditional law of $A_N$, i.e., $P^*(A_N\leq \cdot)$, and the law of $B_N$, , i.e., $P(B_N\leq \cdot)$, respectively, and we denote $\left|\mathcal{L}^*(A_N)-\mathcal{L}(B_N)\right|=o_{P}(1)$ if, for any $\delta>0$, there exists a $N(\delta)$ such that
$$
P(\sup _{f:\mathbb{R}\mapsto \mathbb{R}, s.t., |f(z_1)-f(z_2)|\leq |z_1-z_2|, \|f\|_{\mathrm{L}^\infty}\leq 1}\left|E_{P^*}\left[f(A_N)\right]-E[f(B_N)]\right| \leq \delta) \geq 1-\delta \quad \text { for all } N \geq N(\delta),
$$
where $E_{P^*}(\cdot)$ is calculated based on the bootstrap conditional distribution.

\subsection{Proofs of results in the main context}\label{appen;proof main context}

\begin{lemma}[Convergence rate]\label{lem: consistency}
Given i.i.d. $z_i$, and suppose that the conclusions of Lemmas \ref{lem:appendix local variance}-\ref{lem:append,Entropy} hold, then $\|{\widehat{\alpha}_N^\mu}- \alpha^{\mu}_N\| =O_p(\xi_N)$ with $\xi_N =\mathrm{d}_N$ with $\mathrm{d}_N$ specified in Lemma \ref{lem:append,Entropy}  
\end{lemma}
\begin{proof}[Proof of Lemma \ref{lem: consistency}]
This proof essentially follows the proof of Theorem 3.1 in \cite{chen1998} (or Theorem 3.1 in \cite{chen2007large}). We only need to verify if the conditions A.1-A.4 for Theorem 3.1 from \cite{chen1998} hold. The condition A.1 is directly implied by i.i.d. $z_i$, while the remaining conditions A.2-A.4 are indicated by the conclusions of Lemmas \ref{lem:appendix local variance}-\ref{lem:append,Entropy}.	
\end{proof} 
\begin{remark}
Results from \cite{chen1998} indicate that we can further relax the dependence conditions in Lemma \ref{lem: consistency}.
\end{remark} 
		
\begin{assumption}[Local properties of the functional]\label{assum:average_functional}
(i) $v\mapsto D\tau(\alpha_N^\mu)[v]$ is a linear functional from $\calV_N^\mu$ to $\mathbb{R}$.
(ii)
\[
\sup_{\alpha\in\calN_N^\mu}
\frac{
\left|
\tau(\alpha)-\tau(\alpha_N^\mu)-D\tau(\alpha_N^\mu)[\alpha-\alpha_N^\mu]
\right|}
{\|v_N^\tau\|_\mu}
=o(N^{-1/2}).
\]
\end{assumption}
		 
\begin{assumption}[Local properties of the criterion]\label{assum:Local Behavior of Criterion v2}
(i) $D\ell_N^\mu(z,\alpha_N^\mu)[v]$ is linear in all $v\in\calV_N^\mu$;
(ii) Suppose that for some positive sequence $e_N=o(\min\{N^{-1/2},\rho_N(u_N^\tau)\})$ with $u_N^\tau\equiv\frac{v_N^\tau}{\|v_N^\tau\|_{sd,\mu}}$:
\[
\sup_{\alpha\in\calN_N^\mu}
\bar{E}_N\!\left\{
\ell_N^\mu(z,\alpha\pm e_Nu_N^\tau) -\ell_N^\mu(z,\alpha) -D\ell_N^\mu(z,\alpha_N^\mu)[\pm e_Nu_N^\tau]\right\} =O_p(e_N^2);
\]
and
\[
\sup_{\alpha\in\calN_N^\mu}
\left|
E[\ell_N^\mu(z,\alpha)-\ell_N^\mu(z,\alpha\pm e_Nu_N^\tau)]-\frac{\|\alpha\pm e_Nu_N^\tau-\alpha_N^\mu\|_\mu^2-\|\alpha-\alpha_N^\mu\|_\mu^2}{2}\right|=O(e_N^2).
\]
\end{assumption}
		
\begin{assumption}[CLT]\label{assum:clt v2}
\[
\frac{1}{\sigma_N \sqrt N} \sum_{i=1}^N D\ell_N^\mu(z_i,\alpha_N^\mu)[v_N^\tau] \rightarrow_d N(0,1).
\]
\end{assumption}

\begin{proof}[Proof of Theorem \ref{theo:3.1 v3 0}] \label{Proof of Theorem ref{theo:3.1 v3 0}}
Assumption \ref{assum:norm ratio} implies that $u_N^\tau=\frac{v_N^\tau}{\|v_N^\tau\|_{sd,\mu}} =O(1)$. Under Assumption \ref{assum:ri}, we have that for some $\varepsilon=o(\min\{N^{-1/2},\rho_N(u_N^\tau)\})$,   $\widehat{\alpha}_N^\mu\pm|\varepsilon|u_N^\tau\in\calN_N^\mu$, and given Assumption \ref{assum:norm ratio}, the proof of Theorem 3.1 from \cite{chen2014sieve} implies that 
\begin{align}
&  \left|\sqrt{N}\left\langle{\widehat{\alpha}_N^\mu} -\alpha^{\mu}_N, u_N^\tau\right\rangle-\sqrt{N} \bar{E}_N\left\{D{\ell_{N}^{\mu}} (z, \alpha^{\mu}_N )[u_N^\tau]\right\}\right|=o_p(1), \label{eq:theo1,1} \\
& \frac{\tau ({\widehat{\alpha}_N^\mu} )-\tau ({\alpha_N^\mu})}{\|v^\tau_N \|_{s d,\mu}} -  \langle{\widehat{\alpha}_N^\mu} -\alpha^{\mu}_N, u_N^\tau \rangle=o_p(N^{-1/2}). \label{eq:theo1,2}
\end{align}
Hence, we can conclude as the asymptotic distribution of $ \langle{\widehat{\alpha}_N^\mu} -\alpha^{\mu}_N, u_N^\tau \rangle$ is indicated by Assumption \ref{assum:clt v2}.
\end{proof}
			
\begin{proof}[Proof of Corollary \ref{theo:3.1 v3}]
\label{Proof of Theorem ref{theo:3.1 v3}} 
Assumption  \ref{assum:average_functional}.(ii) implies that 
$${ |\tau(\alpha^0_N)-\tau (\alpha^{\mu}_N)- D\tau(\alpha^{\mu}_N) [\alpha^0_N-\alpha^{\mu}_N ] |}/{\left|v^\tau_N \right|\mu}=o(N^{-1/2}).$$
Moreover, Assumption \ref{assum:average_functional}.(i), together with the definition of $v^\tau_N$, implies that 
$$|D\tau(\alpha^{\mu}_N) [\alpha^0_N-\alpha^{\mu}_N ]| \leq |v^\tau_N |\mu,|\alpha^0_N-\alpha^{\mu}N|\mu.$$ 
These arguments directly yield the first statement in this corollary, and the remaining claim follows directly from the asymptotic normality established in Theorem \ref{theo:3.1 v3 0}.
\end{proof} 
\begin{lemma}\label{lem:empirical sieve Riesz representer closed solution}
Let $\mathbf{1}_{\widehat S_N}(\cdot)$ stack the indicator functions $\{\mathbf{1}(\beta_r\le\cdot)\}_{r\in\widehat S_N}$, let $\gamma=(\delta',\theta')'\in\mathbb{R}^{d_{\vdelta}}\times\mathbb{R}^{|\widehat S_N|}$, and define
\begin{align*}
\widehat{{D}}_{N}^\mu&= \left.\frac{\partial \tau(\widehat{\alpha}_N^\mu+(\delta,\theta'\mathbf{1}_{\widehat S_N}(\cdot)))}{\partial \gamma}\right|_{\gamma=0},\\
\widehat{{H}}_{N}^\mu &=\frac{1}{N}\sum_{i=1}^N\left. \frac{\partial^2 \ell_{N}^\mu(z_i,\widehat{\alpha}_N^\mu+(\delta,\theta'\mathbf{1}_{\widehat S_N}(\cdot)))}{\partial \gamma \partial \gamma'}\right|_{\gamma=0}.
\end{align*}
Then the empirical sieve Riesz representer in Equation \eqref{eq:empirical sieve riesz representer equation system} is
\[
\widehat v_N^\tau
=
(\widehat\delta^\tau,\widehat\theta^{\tau'}\mathbf{1}_{\widehat S_N}(\cdot)),
\qquad
(\widehat\delta^{\tau'},\widehat\theta^{\tau'})'
=
(\widehat{{H}}_{N}^\mu)^{-1}\widehat{{D}}_{N}^\mu.
\]
\end{lemma}

\begin{proof}[Proof of Lemma \ref{lem:empirical sieve Riesz representer closed solution}]
\label{Proof of lem:empirical sieve Riesz representer closed solution}
For each $v \in \widehat \calV_N^\mu, v \neq 0$, let $\gamma_{v} = (\delta_{v}, \theta_{v}) $ such that $v=(\delta_{v}, \mathbf{1}_{\widehat S_N}(\cdot)' \theta_{v})$, and the result directly follows from the fact that the empirical sieve Riesz representer specified in Equation(s)
\eqref{eq:empirical sieve riesz representer equation system} can be rewritten as: 
\begin{equation}\label{eq:empirical sieve riesz representer equation system 2}
\begin{aligned}
& D\tau({\widehat{\alpha}_N^\mu}) [\widehat v_N^\tau]=\sup _{v \in \widehat \calV_N^\mu, v \neq 0} \frac{ \gamma_{v}' 	\widehat{{D}}_{N}^\mu 	(\widehat{{D}}_{N}^{\mu})'  \gamma_{v} }{\gamma_{v}' \widehat{{H}}_{N}^\mu \gamma_{v} } ,\\
& D\tau({\widehat{\alpha}_N^\mu})[v]= \gamma_{v}' \widehat{{H}}_{N}^\mu \gamma_{{\widehat{v}_N^{\mu}}}, \quad \text { for any } v \in \widehat \calV_N^\mu.
\end{aligned}
\end{equation}
\end{proof} 

\begin{assumption}[Smoothness]\label{assum:riesz representer estimator I}
For a positive sequence $(\epsilon_N^*)=o(1)$:\newline
			(i)
			\[
			\sup_{\substack{\alpha\in\calN_N^\mu\\ v_1,v_2\in\calV_N^\mu,\ \|v_1\|_\mu=\|v_2\|_\mu=1}}
			\left|
			E\{D^2\ell_N^\mu(z,\alpha)[v_1,v_2]
			-D^2\ell_N^\mu(z,\alpha_N^\mu)[v_1,v_2]\}
			\right|
			=O(\epsilon_N^*);
			\]
			(ii)
			\[
			\sup_{\substack{\alpha\in\calN_N^\mu\\ v_1,v_2\in\calV_N^\mu,\ \|v_1\|_\mu=\|v_2\|_\mu=1}}
			\bar{E}_N\{D^2\ell_N^\mu(z,\alpha)[v_1,v_2]\}
			=O_p(\epsilon_N^*);
			\]
			(iii)
			\[
			\sup_{\substack{\alpha\in\calN_N^\mu\\ v\in\calV_N^\mu,\ \|v\|_\mu=1}}
			\left|D\tau(\alpha)[v]-D\tau(\alpha_N^\mu)[v]\right|
			=O(\epsilon_N^*);
			\]
			(iv)
			\[
			\sup_{\substack{v\in\calV_N^\mu,\ \|v\|_\mu=1}}
			\left|
			\frac1N\sum_{i=1}^N
			\left\{
			\bigl(D\ell_N^\mu(z_i,\widehat \alpha_N^\mu)[v]\bigr)^2
			-E\bigl[\bigl(D\ell_N^\mu(z,\alpha_N^\mu)[v]\bigr)^2\bigr]
			\right\}
			\right|
			=O_p(\epsilon_N^*).
			\]
			(v) 
							\[
			\sup_{\substack{\alpha\in\calN_N^\mu\cup\{\alpha_N^0\}\\ v\in\calV_N^\mu,\ \|v\|_\mu=1}}
			\left|
			E\!\left[
			D\ell_N^\mu(z,\alpha)[v]
			-D\ell_N^\mu(z,\alpha_N^\mu)[v]
			-D^2\ell_N^\mu(z,\alpha_N^\mu)[v,\alpha-\alpha_N^\mu]
			\right]
			\right|
			=O(\epsilon_N^*\xi_N).
			\] 
		\end{assumption}

\begin{proof}[Proof of Lemma \ref{lem:sieve variance estimator}]
\label{proof of Lemma {lem:sieve variance estimator}} Given  Assumption \ref{assum:ri}.(ii), without losing generality, we can safely work with  $\mathcal{V}_{N}^\mu$ in the proof. The proof of Lemma \ref{lem:sieve variance estimator}.(i) essentially follows the proof of Lemma 5.1 in \cite{chen2014sieve}.  Note that Assumptions \ref{assum:riesz representer estimator I}.(i)-(ii) imply that,
\begin{align}
\left|\frac{\langle\widehat{v}_N^\mu, v\rangle_N-\langle v_N^\tau, v\rangle_\mu}{ \|\widehat{v}_N^\tau\|_\mu \|v\|_\mu}\right|=O_p(\epsilon_N^*). \label{eq:appen lem sieve variance estimator 1}
\end{align}     
Assumption \ref{assum:riesz representer estimator I}.(iii) and (\ref{eq:appen lem sieve variance estimator 1}) jointly imply that 
\begin{align*}
O(\epsilon_N^*) & =  \left|D\tau(\widehat \alpha_N^\mu)[v]-D\tau(\alpha_N^\mu)[v]\right|\\
& = \sup_{\substack{v\in\calV_N^\mu,\ \|v\|_\mu=1}} \left|\frac{\langle\widehat{v}_N^\tau -v_N^\tau, v\rangle_\mu}{ \|v\|_\mu}+\frac{\langle\widehat{v}_N^\tau, v\rangle_N-\langle_\mu\widehat v_N^\tau, v\rangle}{ \|\widehat{v}_N^\tau\|_\mu\|v\|_\mu} \|\widehat{v}_N^\tau\|_\mu \right|\\
& =\sup_{\substack{v \in \mathcal{V}_{N}^\mu, \;\|v\|=1}} \left|\frac{\langle\widehat{v}_N^\tau -v_N^\tau, v\rangle_\mu}{ \|v\|_\mu}\right|+ O_p(\epsilon_N^* \|\widehat{v}_N^\tau\|_\mu),
\end{align*}
and therefore, we have that
\begin{align}
\sup_{\substack{v \in \mathcal{V}_{N}^\mu, \;\|v\|=1}} \left|\frac{\langle\widehat{v}_N^\tau -v_N^\tau, v\rangle_\mu}{ \|v\|_\mu}\right|=_p(\epsilon_N^* \|\widehat{v}_N^\tau\|_\mu ). \label{eq:appen lem sieve variance estimator 2}
\end{align}
Based on the above (\ref{eq:appen lem sieve variance estimator 2}) and (B.13) in \cite{chen2014sieve}, we can deduce Lemma \ref{lem:sieve variance estimator}(i). 
~\\
				
Next, we discuss Lemma \ref{lem:sieve variance estimator}.(ii). Note that 
\begin{align*}
\left|\frac{\|\widehat{v}_{N}^\tau\|_{sd,\mu,N}}{\|v^\tau_N \|_{sd,\mu}}-1\right| & \leq \left|\frac{\|\widehat{v}_{N}^\tau\|_{sd,\mu,N}-\|\widehat{v}_{N}^\tau\|_{sd,\mu}}{\|v^\tau_N \|_{sd,\mu}} \right|  + \frac{\|\widehat{v}_{N}^\tau - {v}_{N}^\tau\|_{sd,\mu}}{\|v^\tau_N \|_{sd,\mu}} \\
& \leq \left|\frac{\|\widehat{v}_{N}^\tau\|_{sd,\mu,N}-\|\widehat{v}_{N}^\tau\|_{sd,\mu}}{\|v^\tau_N \|_{sd,\mu}} \right|  + o_p(1)\\
& = \left|\frac{\|\widehat{v}_{N}^\tau\|_{sd,\mu,N}-\|\widehat{v}_{N}^\tau\|_{sd,\mu}}{\|\widehat{v}_{N}^\tau\|_{sd,\mu}} \right| O_p(1) + o_p(1)
\end{align*}
where the second inequality and the third equality are due to the result in Lemma \ref{lem:sieve variance estimator}.(i)  and Assumption \ref{assum:norm ratio}. Therefore, let $\tilde{w}= \widehat{v}_{N}^\tau/\|\widehat{v}_{N}^\tau\|_\mu$. Assumptions \ref{assum:ri} and \ref{assum:norm ratio} then imply that we only need to establish
\begin{align*}
\left| \|\tilde{w}\|_{sd,\mu,N}-\|\tilde{w} \|_{sd,\mu}
\right|=O_p(\epsilon_N^*),
\end{align*}
which is indicated by Assumption \ref{assum:riesz representer estimator I}.(iv).
\end{proof} 
			
\begin{proof}[Proof of Lemma \ref{lem:bias bound}]
\label{proof of {lem:bias bound}} Let  ${\alpha^0_N}=(\delta_{\alpha_N^0}, $ $ \theta_{\alpha_N^0}' \mathbf{1}_{S_N}(\cdot))$ and $\alpha^{\mu}_N$$=(\delta_{\alpha_N^\mu},$$\theta_{\alpha_N^\mu}' \mathbf{1}_{S_N}(\cdot))$ with $\delta_{\alpha_N^0},\delta_{\alpha_N^\mu}\in \mathcal{C}$ and $\theta_{\alpha_N^0}, \theta_{\alpha_N^\mu} \in (0,1]^{|S_N|}$. In addition, let ${v}= {\alpha^0_N}-\alpha^{\mu}_N=(\delta_{\alpha_N^0}-\delta_{\alpha_N^\mu}, ( \theta_{\alpha_N^0}-\theta_{\alpha_N^\mu})' \mathbf{1}_{S_N}(\cdot))$ and $\bar{v}=({\delta}_{\bar v}, {\theta}_{\bar{v}}'\mathbf{1}_{S_N}(\cdot))={v}/{\|v\|}.$
Assumption \ref{assum:ri}(i) and \ref{assum:riesz representer estimator I}.(v) imply that 
\begin{align}
& \langle \bar{v}, {\alpha^0_N}-\alpha^{\mu}_N \rangle =   
{E}\left[D^2\ell_{N}^\mu(z,{\alpha^\mu_N})[\bar{v},{\alpha^0_N} - \alpha^{\mu}_N] \right]\notag \\
\qquad\qquad\qquad & \leq E\Bigl[(D\ell_{N}^\mu(z,{\alpha^0_N})[\bar{v}]) - (D\ell_{N}^\mu(z,\alpha^{\mu}_N)[\bar{v}])  \Bigr] + O(\epsilon_N^* \xi_N). 
\label{eq:lemma4-1}
\end{align}  
Note that since $\bar{v}$ is in the feasible direction under Assumption \ref{assum:ri}, so that by the definition of ${\alpha^0_N}$ and $\alpha^{\mu}_N$ we have that 
\begin{align*}
E\Bigl[(D\ell_{N}^\mu(z,{\alpha^0_N})[\bar{v}]) - (D\ell_{N}^\mu(z,\alpha^{\mu}_N)[\bar{v}])  \Bigr] = \mu_N \theta_{\alpha_N^0}' {\theta}_{\bar{v}}\leq \mu_N \|\theta_{\alpha_N^0}\|_e \|{\theta}_{\bar{v}}\|_e.
\end{align*} 
Note that $\mu_N\|{\theta}_{\bar{v}}\|_e^2 \leq \|\bar{v}\|_\mu^2=1$, which then leads to the conclusion noting that $ \|\theta_{\alpha_N^0}\|_e \leq 1$.  
\end{proof}

\begin{corollary}\label{lem:append ci 1}
Suppose that  $\epsilon_N^*\xi_N=o(N^{-1/2})$, Assumption \ref{assum:norm ratio}, Corollary \ref{theo:3.1 v3} and Lemmas \ref{lem:sieve variance estimator}--\ref{lem:bias bound} hold,  then
\[
\liminf_{N\to\infty}P\left(\tau(\alpha_N^0) \in\left[\tau(\widehat{\alpha}_N^\mu)\pm (N^{-1/2}\Phi^{-1}(1-q/2)\widehat\sigma_N+(\sqrt{\mu_N} +\epsilon_N^*\xi_N)(1+N^{-1/2}\xi_N^{-1}) \|\widehat v_N^\tau\|_\mu)\right]\right)\ge 1-q. 
\]
\end{corollary}
\begin{proof}[Proof of Corollary \ref{lem:append ci 1}]\label{Proof of Lem {lem:append ci 1}}
To conclude, Corollary \ref{theo:3.1 v3} indicates that  we only need to show
\begin{align}
&  N^{-1/2}\Phi^{-1}(1-q/2)\widehat\sigma_N+(\sqrt{\mu_N} +\epsilon_N^*\xi_N)(1+N^{-1/2}\xi_N^{-1})\|\widehat v_N^\tau\|_\mu \nonumber \\ & \quad\quad \geq  N^{-1/2}\Phi^{-1}(1-q/2)\sigma_N+(\|\alpha_N^0-\alpha_N^\mu\|_\mu) \|v_N^\tau\|_\mu  +o_p(1). \label{eq:coro b3 1}	
\end{align}
Note that Lemma \ref{lem:sieve variance estimator}.(i) implies that 
\begin{align}
\left\|\widehat{v}_N^\tau-v_N^\tau\right\|_\mu=O_p\left(\epsilon_N^*\left\|v_N^\tau\right\|_\mu\right), \left\|\widehat{v}_N^\tau\right\|_\mu=O_p\left(\left\|v_N^\tau\right\|_\mu\right). \label{eq: coro b3 0}
\end{align}
Let $v_N^\tau=(\delta^\tau,\theta^{\tau'}\mathbf{1}_{S_N}(\cdot))$, then
\begin{align}
\mu_N\left|\left\|\hat{\theta}^\tau\right\|_e^2-\left\|\theta^\tau\right\|_e^2\right| & \leq \sqrt{\mu_N}\left\|\hat{\theta}^\tau-\theta^\tau\right\|_e \sqrt{\mu_N}\left(\left\|\hat{\theta}^\tau\right\|_e+\left\|\theta^\tau\right\|_e\right)\nonumber \\
& \lesssim\left\|\widehat{v}_N^\tau-v_N^\tau\right\|_\mu\left(\left\|\widehat{v}_N^\tau\right\|_\mu+\left\|v_N^\tau\right\|_\mu\right)=O_p\left(\epsilon_N^*\left\|v_N^\tau\right\|_\mu^2\right), \nonumber
\end{align}
where the second inequality is due to Assumption \ref{assum:norm ratio} and the last equality is due to (\ref{eq: coro b3 0}). Hence, we have that 
\begin{align}
\widehat{\sigma}_N^2-\sigma_N^2=\|\widehat{v}_N^\tau \|_{s d, \mu, N}^2-\left\|v_N^\tau\right\|_{s d, \mu}^2-\mu_N\left(\left\|\hat{\theta}_\tau\right\|_e^2-\left\|\theta_\tau\right\|_e^2\right) 	=O_p\left(\epsilon_N^*\left\|v_N^\tau\right\|_\mu^2\right). \label{eq: coro b3 3}
\end{align}
Equation (\ref{eq: coro b3 3}) combined with the fact that $|\sqrt{a}-\sqrt{b}| \leq \sqrt{|a-b|}$ leads to 
\begin{align}
\widehat{\sigma}_N-\sigma_N=O_p\left(\sqrt{\epsilon_N^*}\left\|v_N^\tau\right\|_\mu\right). \label{eq: coro b3 4}
\end{align}
In addition, (\ref{eq: coro b3 0}) also implies that 
\begin{align}
\left\|{v}_N^\tau\right\|_\mu \leq  \|\widehat v_N^\tau\|_\mu + \left\|\widehat{v}_N^\tau-v_N^\tau\right\|_\mu = (1+O_p\left(\epsilon_N^*\right))\left\|\widehat{v}_N^\tau\right\|_\mu. \label{eq: coro b3 5}
\end{align}
Given $\epsilon_N^*\xi_N=o(N^{-1/2})$, we know $\epsilon_N^*=o(N^{-1/2}\xi_N^{-1})$, and hence Lemma \ref{lem:bias bound} and (\ref{eq: coro b3 4})-(\ref{eq: coro b3 5}) lead to (\ref{eq:coro b3 1}). The interval in the main context further replaces the term $\epsilon_N^*\xi_N$ with $N^{-1/2}$.
\end{proof}

\begin{proof}[Proof of Lemma \ref{lem:bound1}]\label{Proof of Lemma ref{lem:bound1}}
Assumption \ref{assum:average_functional} implies that $|\tau(\alpha_N^0)-\tau(\alpha_N^\mu)| \leq  \|\alpha_N^0-\alpha_N^\mu\|_\mu\|v_N^\tau\|_\mu   + \|v_N^\tau\|_\mu o(N^{-1/2})$, and Lemma \ref{lem:bias bound}, Assumption \ref{assum:regular functional} and $\epsilon_N^*\xi_N=o(N^{-1/2})$ then further indicate that 
\begin{align*}
|\tau(\alpha_N^0)-\tau(\alpha_N^\mu)| & \leq  (\sqrt{\mu_N}+O(\epsilon_N^*\xi_N))\|v_N^\tau\|_\mu+ \|v_N^\tau\|_\mu o(N^{-1/2})\\
& = \sqrt{\mu_N}\|v_N^\tau\|_\mu   + o(N^{-1/2}).
\end{align*}
Recall that $\gamma_N^\mu =((\delta_N^{\mu})',(\theta_N^{\mu})')'=(H_N^\mu)^{-1}D_N^\mu$ and $\left\|v_N^\tau\right\|_\mu^2=\left(\gamma_N^\tau\right)^{\prime} H_N^\mu \gamma_N^\tau. $ Once we show  $2\sigma_N^2 \geq \|v_N^\tau\|_{\mu}^2-O(N^{-1})$  provided that $\|v_N^\tau\|_{sd,\mu}^2=\|v_N^\tau\|_{\mu}^2-O(N^{-1})$ and $\mu_N=o(1)$ we can then conclude. 
					
Let $f_{N, i}=h_{N, i}^{\prime} D_N^\mu, i=1,\cdots, d_\delta+|S_N|$, and then we have $\gamma_N^\tau=\sum_{i=1}^{d_\delta+|S_N|} \frac{f_{N, i}}{w_{N, i}} h_{N, i}$ and
\begin{align}
\left\|v_N^\tau\right\|_\mu^2=\sum_{i=1}^{d_\delta+|S_N|} \frac{f_{N, i}^2}{w_{N, i}}. \label{eq: lem4 1}	
\end{align}
Let $W_N=\text{diag}(0_{d_\delta}, I_{|S_N|}),$ and thus 
$$\mu_N \|\theta_N^{\mu}\|_e^2= \mu_N   (\gamma_N^{\mu})' W_N \gamma_N^\mu\leq \mu_N    {(\gamma_N^{\mu})'}\gamma_N^\mu= \mu_N  {\sum_{i=1}^{d_\delta+|S_N|}} \frac{f_{N, i}^2}{w_{N, i}^2}=A_1+A_2,$$ where $A_1= \mu_N \sum_{i\in \mathcal{A}_N} \frac{f_{N, i}^2}{w_{N, i}^2}, A_2=\mu_N \sum_{i\notin \mathcal{A}_N} \frac{f_{N, i}^2}{w_{N, i}^2}, \mathcal{A}_N=\left\{i: |w_{N, i}-\mu_N|<(1+\eta) \mu_N\right\}.$ Using $w_{N, i} \geq \mu_N$ and Assumption \ref{assume:regular functional II},
\begin{align}
A_1= \mu_N \sum_{i \in \mathcal{A}_N} \frac{f_{N, i}^2}{w_{N, i}^2} & \leq \frac{1}{\mu_N} \sum_{i \in \mathcal{A}_N} f_{N, i}^2  =O(N^{-1}). \label{eq: lem4 A1}
\end{align}
For $i \notin \mathcal{A}_N$, we know that $w_{N, i} \geq(2+\eta) \mu_N$ and thus
\begin{align}
A_2=\mu_N \sum_{i \notin  {\mathcal{A}_N}} \frac{f_{N, i}^2}{w_{N, i}^2} & \leq \frac{1}{2+\eta} \sum_{i \notin {\mathcal{A}_N}} \frac{f_{N, i}^2}{w_{N, i}} \leq \frac{1}{2+\eta}\left\|v_N^\tau\right\|_\mu^2. \label{eq: lem4 A2}
\end{align} 
(\ref{eq: lem4 1}) - (\ref{eq: lem4 A2}) thus imply that $\mu_N \|\theta_N^{\mu}\|_e^2 \leq \frac{1}{2+\eta}\left\|v_N^\tau\right\|_\mu^2 +O(N^{-1})$, and this combined with $\|v_N^\tau\|_{sd,\mu}^2=\|v_N^\tau\|_{\mu}^2-O(N^{-1})$ leads to   $2\sigma_N^2 \geq \|v_N^\tau\|_{\mu}^2-O(N^{-1})$.
\end{proof}
				
\begin{proof}[Proof of Corollary \ref{coro}]
\label{proof of Corollary {coro}}
Theorems \ref{theo:3.1 v3} and \ref{theo:bootstrap} imply that the conditional distribution of  $\sqrt{N}(\tau({\widehat{\alpha}_N^{\mu,B}})-\tau({\widehat{\alpha}_N^\mu}))$ is approximated by $\sigma_N N(0,1)$, and to conclude we only need to discuss the bound for the regularization bias term, which is indicated by Lemma \ref{lem:bound1}.
\end{proof}

\subsection{Intermediate results}

In this section, we verify some results for models with a more specific structure, though the conclusions may hold under more relaxed conditions. 
			
\subsubsection{Verification of Assumption \ref{assum:ri}}\label{appen:verification of assum 1}

We use the model analyzed in \cite{heiss2022} and justify Assumption \ref{assum:ri} in such a setting. We start from a small modification of their Equation (11) to allow for a combination of fixed and random coefficients. We consider the following linear approximate probability model
\begin{equation}\label{eq:linear_prob 2}
y_{i,j} = \sum_{r=1}^{R_N} X_{i,j}(\beta_r, \delta^0)\theta^0_r+\epsilon_{i,j},
\end{equation}
where $X_{i,j}(\beta_r, \delta^0)$ denotes the conditional choice probability which depends on $(x_i, \beta_r, \delta^0)$, $\epsilon_{i,j}$ denotes the error of the probability model, and $\alpha^0_N =(\delta^0, \mathbf{1}_{N}(\cdot)'\theta^0)$. In addition, let $\widehat{\alpha}_N^\mu = (\widehat\delta^\mu_N, \mathbf{1}_{N}(\cdot)'\widehat\theta_N^\mu)$. To match their leave-one-grid-out approach, we assume there is an additional $R_N+1$ grid point that is left out, and the non-negative constraint for the corresponding weight is strongly active, which is a reasonable assumption that can be easily justified by the possibility of excluding a point that is located far outside the presumed true support. Let $X(\delta)$ refer to the $NJ\times R_N$ regressor matrix in Equation (\ref{eq:linear_prob 2}) evaluated at $\delta$ and each corresponding $\beta_r$ in $\mathcal{B}_N$, and $X_i(\delta)$ to the corresponding $J\times R_N$ regressor matrix for observation unit $i$. To keep the notation simple, we omit the dependence on $\beta_r$'s (or $\mathcal{B}_N$) in the following. Moreover, let $X_{A}(\delta)$ be the regressor submatrix corresponding to the grids in the index set $S$, and we write $S=S_N, S^C=S_N^C$. For a given penalty $\mu_N \geq 0$, define
$$
A_S(\mu_N, \delta)\equiv \frac{1}{N J} X_S(\delta)^{\prime} X_S(\delta)+\mu_N I_s, \quad \xi_S\left(\mu_N, \delta\right)\equiv\lambda_{\min}\left\{A_S(\mu_N, \delta)\right\},
$$
where 
$I_s$ is the identity matrix and $\lambda_{\min}(A)$ denotes the minimum eigenvalue of $A$. Define 
$$
q_S(\mu_N, \delta)=A_S(\mu_N, \delta)^{-1}\left\{\frac{1}{N J} X_S(\delta)^{\prime} X_S(\delta) \theta_S^{0}-\nu_N \iota_S\right\}, \quad \rho_N\left(\mu_N, \delta\right)=\min _{r \in S}\left|q_{S, r}(\mu_N, \delta)\right|,
$$
where $\iota_S$ is a vector of $|S|$ ones, $\nu_N>0$ denotes the Lagrange multiplier for the simplex constraint in the constrained optimization, and $q_{S, r}$ refers to the $r$th entry of $q_{S, r}$. In addition, let the $|S|\times 1$ vector $\theta_S^{0}$ stack $\theta_r^{0}$ for $r\in S$, the $NJ\times 1$ vector $\epsilon$ stack $\epsilon_{i,j}$,  $U_N(\delta)=A_S(\mu_N, \delta)^{-1} \frac{1}{N J} X_S(\delta)^{\prime} \epsilon$, and 
$$V_N(\delta)=\frac{1}{N J} X_{S^c}(\delta)^{\prime} X_S(\delta) A_S(\mu_N, \delta)^{-1}\left\{\nu_N \iota_S+\mu_N \theta_S^{0}-\frac{1}{N J} X_S(\delta)^{\prime} \epsilon \right\}+\frac{1}{N J} X_{S^c}(\delta)^{\prime} \epsilon.
$$
						
\begin{assumption}\label{assum:verification of assum 1} We make the following assumptions for the model specified in Equation (\ref{eq:linear_prob 2}). 
\begin{enumerate}
\item[(i)]  (a) $\left(\epsilon_i=\left(\epsilon_{i, 1}, \ldots, \epsilon_{i, J}\right)\right)_{i=1}^N$ are independent.
(b) $\epsilon_{i, j}$ is sub-Gaussian: $\mathbb{E}\left[\exp \left(t \epsilon_{i, j}\right)\right] \leq \exp \left(\frac{\sigma^2 t^2}{2}\right) \quad(\forall t \in \mathbb{R})$ for $\sigma>0$.
(c) $\left(X_i(\delta^0)\right)_{i=1}^N$ are i.i.d. with a density bounded from above, and each entry of $X_i(\delta^0)$'s is within the range $[-1,1]$.
(d) $\mathbb{E}\left[\epsilon_i \mid x_1, \cdots, x_N\right]=0$.
\item[(ii)] There exist nonnegative deterministic sequences $\rho_{N,1}, \rho_{N,2}, \rho_{N,3}, \rho_N^*$ such that wpa1, $\|U_N(\widehat{\delta}^\mu_N)-U_N(\delta^{0})\|_{\infty} \leq \rho_{N,1},  \rho_N(\mu_N, \widehat{\delta}_N) \geq \rho_N(\mu_N, \delta^0)-\rho_{N,2},  \|V_N(\widehat{\delta}_N)-V_N(\delta_N^{\star})\|_{\infty} \leq \rho_{N,3}$, $\rho_N(\mu_N, \delta^0)\geq \rho_N^*$.
\item[(iv)] The following rates hold such that: (a) for $\nu_N>0$, there exists a positive constant $\eta>0$ (independent of $N$) such that 
$$
\max _{r \in S^C} e_r' \left[\frac{1}{N J} (X_{S^C}(\delta^0))'  X_S(\delta^0)(\frac{1}{N J} (X_S(\delta^0))'  X_S(\delta^0)+\mu I_S)^{-1}\left(\iota_S+\frac{\mu_N}{\nu_N} \theta_S^0\right)\right] \leq 1-\eta
$$
where $e_r$ denotes the standard basis unit vector with $r$th entry equal to one; (b) for some $c_\rho\geq 1,$ the followings hold such that, $\rho_N^* > c_\rho \rho_{N,1}+ \rho_{N,2},$ $\rho_{N,3}<\eta \nu_N$, 
$$
\lim _{N \rightarrow \infty} 2 |S| J \exp \left(-\frac{N \xi_S(\mu_N, \delta^0)^2 (\rho_N^*-  c_\rho \rho_{N,1}- \rho_{N,2})^2}{2 |S| c_\rho^2}\right)=0,
$$
and
$$\lim _{N \rightarrow \infty} 2 R_N J \exp \left(-N \left((\eta \nu_N-\rho_{N,3})\frac{\xi_S(\mu_N,\delta^0)}{|S| \sqrt{|S|}+\xi_S(\mu_N,\delta^0)}\right)^2 / 2\right)=0.
$$ 
\end{enumerate}
\end{assumption}
~\\
When there is no fixed coefficients or fixed coeeficients are known, $\rho_{N,1}, \rho_{N,2}, \rho_{N,3}$ can be set to be zeros, then in case where $c_\rho=1$ and $\rho_N^*=\rho_N(\mu_N, \delta^0)$, Assumption \ref{assum:verification of assum 1} exactly conincides with Assumption 1, rate conditions NEIC and RCDG considered in  \cite{heiss2022}.
\begin{lemma}\label{lem:verification of Assumption ri}
Suppose that Assumption \ref{assum:verification of assum 1} holds with $c_\rho>1$, for every $v=(\delta_v,\theta_v'\mathbf{1}_{S_N}(\cdot)) \in\calN_N^\mu$, for all $\varepsilon\leq (1-c_\rho^{-1})( \rho_N^{\star}-\rho_{N,2}-c_\rho \rho_{N,1})/\|\theta_v\|_\infty $,
	$\widehat{\alpha}_N^\mu\pm|\varepsilon|v\in\calN_N^\mu\text{, wpa1.}$
\end{lemma}
\begin{proof}[Proof of Lemma \ref{lem:verification of Assumption ri}]
	We first derive a similar sign consistency result as in Theorem 1 from \cite{heiss2022}, the proof of which implies that we only need to show 
	\begin{align}\label{eq:lemma verification 1}
	\lim\limits_{N}	P(M_{\widehat\delta^\mu_N}(U)  \cap M_{\widehat\delta^\mu_N}(V)) =1,
	\end{align}
where $M_{\delta}(U)=\{\|U_N(\delta)\|_{\infty} \leq c_\rho^{-1} \rho_N(\mu_N,\delta)\}$ and $M_\delta(V)=\left\{\max _{r \in S^c} V_{N, r}(\delta) \leq \nu_N\right\}.$

Assumptions \ref{assum:verification of assum 1}.(i) and (iv).(a) imply
$
E\left[V_N\left(\delta^{0}\right) \mid X_1, \ldots, X_N\right] \leq\left(1-\eta_N\right) \nu_N \iota_{S^c},
$
and if $M_{\hat{\delta}_N}^c(V)$ occurs, then
$\max _{r \in S^c} V_{N, r}\left(\delta_N^{\star}\right)>\nu_N-\rho_{N,3}.$  The proof of Lemma 4 from the appendix of \cite{heiss2022} gives
\begin{align}
P((M_{\widehat\delta^\mu_N}(V))^c)	& \leq  {P}\left(\max _{r \in S^c}\left\{V_{N, r}\left(\delta_N^{\star}\right)-{E}\left[V_{N, r}\left(\delta_N^{\star}\right) \mid X_1, \ldots, X_N\right]\right\}>\eta \nu_N-\rho_{N,3}\right) \nonumber \\
	&  \leq  2 R_N J \exp \left(-N \left((\eta \nu_N-\rho_{N,3})\frac{\xi_S(\mu_N,\delta^0)}{|S| \sqrt{|S|}+\xi_S(\mu_N,\delta^0)}\right)^2 / 2\right),\label{eq:lemma verification 1-1}
\end{align}
which goes to zero under Assumption \ref{assum:verification of assum 1}.(iv).(b). If $(M_{\widehat{\delta}_N^\mu}(U))^c$ occurs, then
$
\|U_N(\widehat{\delta}_N^\mu)\|_{\infty} \geq c_\rho^{-1}\rho_N(\mu_N, \widehat{\delta}_N^\mu) \geq c_\rho^{-1}( \rho_N^{\star}-\rho_{N,2}),
$
and thus via Assumption \ref{assum:verification of assum 1}.(ii).
$$
\|U_N(\delta^{0})\|_{\infty} \geq\|U_N(\widehat{\delta}_N^\mu)\|_{\infty}-\|U_N(\widehat{\delta}_N^\mu)-U_N(\delta^{0})\|_{\infty} \geq c_\rho^{-1}( \rho_N^{\star}-\rho_{N,2}-c_\rho \rho_{N,1}).
$$
Thus, the proof of Lemma 5 from the appendix of \cite{heiss2022} gives
\begin{align}
	P((M_{\widehat\delta^\mu_N}(U))^c)	& \leq  {P}\left(\|U_N(\delta_N^{\star})\|_{\infty} \geq c_\rho^{-1}( \rho_N^{\star}-\rho_{N,2}-c_\rho \rho_{N,1})\right) \nonumber \\
	&  \leq 2 |S| J \exp \left(-\frac{N \xi_S(\mu_N, \delta^0)^2 (\rho_N^*(\mu_N)-  c_\rho \rho_{N,1}- \rho_{N,2})^2}{2 |S| c_\rho^2}\right),\label{eq:lemma verification 1-2}
\end{align} 
which goes to zero under Assumption \ref{assum:verification of assum 1}.(iv).(b). Equations (\ref{eq:lemma verification 1-1}) and (\ref{eq:lemma verification 1-2}) together indicates the sign consistency that $\widehat{\alpha}_N^\mu$ and $\alpha^0_N$ agree on which grid weights are zero wpa1.

The remaining result directly follows from Equation (\ref{eq:lemma verification 1-2}) once note that the event $M_{\widehat\delta^\mu_N}(U)$ indicates (see the proof of Lemma 3 in the appendix of \cite{heiss2022}) that $\|\widehat\theta_N^\mu\|_\infty > \rho_N(\mu_N, \widehat{\delta}_N^\mu) - \|U_N(\widehat{\delta}_N^\mu)\|_{\infty} \geq (1-c_\rho^{-1})( \rho_N^{\star}-\rho_{N,2}-c_\rho \rho_{N,1}).$

\end{proof} 
Lemma \ref{lem:verification of Assumption ri} verifies the statements in Assumption \ref{assum:ri}. It gives one set of sufficient conditions under which \(\widehat\alpha_N^\mu\) and \(\alpha_N^0\) agree on which grid weights are zero, which in turn verifies Assumption \ref{assum:ri}.(ii). The sign agreement between \(\alpha_N^\mu\) and \(\alpha_N^0\) can be verified by an analogous population argument, so we omit the details. In addition, the proof of Lemma 
\ref{lem:verification of Assumption ri} shows that too many grid points lead to a relatively small value of $
\rho_N(v)$, while the smoothness assumption in Assumption 
\ref{assum:Local Behavior of Criterion v2} also relates to $
\rho_N(v)$ and requires stronger smoothness conditions as $
\rho_N(v)$ decreases. Together, these assumptions implicitly restrict the number of grid points from increasing too quickly with respect to $N$, while our simulations show that in finite samples, our results hold with relatively large $R_N$.

\subsubsection{Intermediate results related to Lemma \ref{lem: consistency}}

\begin{assumption}\label{assum:append a specific model specification}
The following conditions hold:

\begin{enumerate}
\item[(i)] For $\alpha \in \mathcal{A}_N$, $g(x,\delta_{\alpha},\beta_r) \in[0,1]^J$ and ${P}(x_i,{\alpha}) =P_{\gamma}(x,\gamma_{\alpha})\equiv G(x,\delta_{\alpha})^\top \theta_{\alpha}\equiv \sum_{r=1}^{R_N} g(x,\delta_{\alpha},\beta_r)\theta_r.$ 
\item[(ii)] For almost all $z=(y,x)\in\mathcal{Z}$ with binary $y$, $g(x,\delta_{\alpha},\beta_r)$ and ${Q_{N,\gamma}^{\mu}}(\cdot)$ are twice differentiable with respect to $\gamma=(\delta_{\alpha}', \theta_{{\alpha}}')'$. Moreover, there exists a constant $L_g<\infty$ such that, for all $\delta_1,\delta_2$, 
\[
\left|g(x,\delta_1,\beta_r)-g(x,\delta_2,\beta_r)\right|\le L_g\|\delta_1-\delta_2\|_e .
\]
\item[(iii)] The expectation operator can be interchanged with the first and second derivatives of ${\ell_{N,\gamma}^{\mu}}(\cdot)$ with respect to $\gamma$. That is,
\[
\frac{\partial}{\partial \gamma}\mathbb{E}\!\left[{\ell_{N,\gamma}^{\mu}}(z,\gamma)\right]=\mathbb{E}\!\left[\frac{\partial}{\partial \gamma}{\ell_{N,\gamma}^{\mu}}(z,\gamma)\right],
\]
and
\[
\frac{\partial^2}{\partial \gamma \partial \gamma^\top}\mathbb{E}\!\left[{\ell_{N,\gamma}^{\mu}}(z,\gamma)\right]=\mathbb{E}\!\left[\frac{\partial^2}{\partial \gamma \partial \gamma^\top}{\ell_{N,\gamma}^{\mu}}(z,\gamma)\right].
\]
\item[(iv)] Let $$\|\theta_{v}\|_{W_{\mu}}^2\equiv\theta_{v}^\top W_{\mu}\theta_{v}$$ with $ W_{\mu} =\operatorname{E}(G(x,\delta_{\alpha^{\mu}_N})\Sigma(x) G(x,\delta_{\alpha^{\mu}_N})^\top)$.  
For $v\in T_{\mathcal{A}_N}(\alpha^{\mu}_N)$, whose associated $h_{v} = (\delta_{v},\theta_{v}) \in T_{\mathcal{A}_N,h} ({\alpha_N^\mu}),$ satisfying $\|\delta_{v}\|_e+\|\theta_{v}\|_{W_{\mu}}+\mu_N\|\theta_{v}\|_e=1,$ there exists a constant $1\geq \overline{\rho}>0$ such that for all $t\in[0,\overline{\rho})$, the following expansion holds
\[
\begin{aligned}
&{Q_{N}^{\mu}}({\alpha_N^\mu}+tv)-{Q_{N}^{\mu}}({\alpha_N^\mu}) ={Q_{N,\gamma}^{\mu}}(\gamma_{{\alpha_N^\mu}}+th_{v})-{Q_{N,\gamma}^{\mu}}(\gamma_{{\alpha_N^\mu}})\\
&\quad \qquad \qquad = tD_+{Q_{N,\gamma}^{\mu}} (\gamma_{\alpha^{\mu}_N})[h_{v}] + \frac{1}{2}t^2 D_+^2{Q_{N,\gamma}^{\mu}} (\gamma_{\alpha^{\mu}_N}) [h_{v},h_{v}] + R_t[h_{v}],
\end{aligned}
\]
where  \[\sup_{\substack{{\alpha_N^\mu}+tv \in \mathcal{C} \times \Delta_{R_N} 	1
}}\frac{|R_t[h_{v}]|}{t^2} \rightarrow 0 \text{~~ as~~} t\downarrow 0.
\]
In addition, there exist constants $\overline{\kappa},\underline{\kappa}>0$ such that
\[
\begin{aligned}
& \inf_{\substack{t\geq \bar{\rho},\;\gamma_{{\alpha_N^\mu}}+th_{v}\in \mathcal{C} \times \Delta_{R_N}
}}
[{Q_{N,\gamma}^{\mu}}(\gamma_{{\alpha_N^\mu}}+th_{v})-{Q_{N,\gamma}^{\mu}}(\gamma_{{\alpha_N^\mu}})]\ge \overline{\kappa},
\end{aligned}
\]
and $D_+^2{Q_{N,\gamma}^{\mu}}(\gamma_{\alpha^{\mu}_N})[h_{v},h_{v}]\ge\underline{\kappa}.$ 
\end{enumerate}
\end{assumption}
Assumption \ref{assum:append a specific model specification} can be shown to hold upon the specification considered in \cite{heiss2022} and directly verified for the case where there are no fixed coefficients $\delta_{\alpha}$ in $\alpha$. Note that Assumption \ref{assum:append a specific model specification}. (iv) does not rule out the singular design matrices as $\|\cdot \|_{W_{\mu}}$ depends upon $W_{\mu}$, and we do not restrict the rank conditions upon $G(x,\delta_{\alpha_{\mu}^N})$.

\begin{lemma}\label{lem:lower loss function difference}
Under Assumption \ref{assum:append a specific model specification}, there exists $\underline{c},\overline{\varepsilon}>0$ such that if ${Q_{N}^{\mu}}({\alpha_N^\mu}+ v)-{Q_{N}^{\mu}}({\alpha_N^\mu})\leq \overline{\varepsilon}^2$, then $ {Q_{N}^{\mu}}({\alpha_N^\mu}+ v)-{Q_{N}^{\mu}}({\alpha_N^\mu}) \geq \underline{c} (\|\delta_{v}\|_e^2 +\|\theta_{v}\|_{W_{\mu}}^2)$. 
\end{lemma}
\begin{proof}[Proof of Lemma \ref{lem:lower loss function difference}]
Note that Assumption \ref{assum:append a specific model specification}. (iv) combined with the fact that  $\alpha^\mu_N$ is the constrained minima imply that there exists $0\leq \rho\leq \overline{\rho}$, such that for $t \leq \rho$, 
\[
{Q_{N,\gamma}^{\mu}}(\gamma_{{\alpha_N^\mu}}+th_{v})- {Q_{N,\gamma}^{\mu}} (\gamma_{{\alpha_N^\mu}})\geq  \frac{1}{4} t^2 \underline{\kappa},
\]
where $h_{v} = (\delta_{v},\theta_{v})$ satisfying $\|\delta_{v}\|_e^2 + \|\theta_{v}\|_{W_{\mu}}^2 + \mu_N\|\theta_{v}\|_e = 1.$
				
Now let $h_{v}=(\delta_{v},\theta_{v})$ be the associated weight of $v$, $t_{v}= \|\delta_{v}\|_e+\|\theta_{v}\|_{W_{\mu}}+\mu_N\|\theta_{v}\|_e$ and $\tilde{h}_{v}=h_{v}/t_{v}$. Thus we can rewrite ${Q_{N}^{\mu}}({\alpha_N^\mu}+ v)-{Q_{N}^{\mu}}({\alpha_N^\mu}) $ as ${Q_{N,\gamma}^{\mu}}(\gamma_{{\alpha_N^\mu}}+t_{v}\tilde{h}_{v})-{Q_{N,\gamma}^{\mu}}(\gamma_{{\alpha_N^\mu}})$.  Let $\overline{\varepsilon}^2 \leq \min\{\overline{\kappa}, \frac{1}{4} {{\rho}}^2 \underline{\kappa}\}$. The aforementioned results would then indicate for any $v$ satisfying ${Q_{N}^{\mu}}({\alpha_N^\mu}+ v)-{Q_{N}^{\mu}}({\alpha_N^\mu}) \leq \overline{\varepsilon}^2$, we have $t_{v} \leq {\rho}$, and thus 
\[
{Q_{N}^{\mu}}({\alpha_N^\mu}+ v)-{Q_{N}^{\mu}}({\alpha_N^\mu}) \geq \frac{1}{4} t_{v}^2 \underline{\kappa}, 
\] 
which then leads to the conclusion as $t_{v}^2 \geq \|\delta_{v}\|_e^2 +\|\theta_{v}\|_{W_{\mu}}^2$.
\end{proof}

\begin{lemma}[Local varaince]\label{lem:appendix local variance}
Suppose that the eigenvalues of weight matrix  $\Sigma(X)$ are uniformly upper and lower bounded by $C,1/C$ for almost all values of $X$ respectively, under Assumption \ref{assum:append a specific model specification}, we have that there exists $C_\ell>0$ and $\varepsilon_0>0$ for all sufficiently small $\varepsilon_0\geq \varepsilon>0$ such that  
\begin{align*}
\sup_{\alpha \in \mathcal{A}_N: {Q_{N}^{\mu}}({\alpha})-{Q_{N}^{\mu}}({\alpha^{\mu}_N}) \leq \varepsilon^2} \operatorname{Var}\left[ {\ell_{N}^{\mu}}(z,\alpha)-{\ell_{N}^{\mu}}(z, \alpha^{\mu}_N) \right] \leq C_\ell \varepsilon^2.
\end{align*}
\end{lemma}

\begin{proof}[Proof of Lemma \ref{lem:appendix local variance}]
Let $h= (h_{\delta}, h_{\theta}) = \gamma_{\alpha}-\gamma_{\alpha^{\mu}_N}$, $\bar{\rho} = Y-G(x,\delta_{\alpha_{\mu}^N})'\theta_{\alpha_{\mu}^N}$ and $\Delta_\alpha =G(x,\delta_{\alpha_{}})'\theta_{\alpha_{}}-G(x,\delta_{\alpha_{\mu}^N})'\theta_{\alpha_{\mu}^N}$. Note that by construction, there exists $C_1$
\begin{align*}
E \|(\Sigma(X))^{1/2}\Delta_\alpha\|_e^2 = & E\| G(x,\delta_{\alpha_{}})'\theta_{\alpha_{}}-G(x,\delta_{\alpha_{\mu}^N})'\theta_{\alpha_{}}+G(x,\delta_{\alpha_{\mu}^N})'(\theta_{\alpha_{}}-\theta_{\alpha_{\mu}^N}) \|_{\Sigma(X)}^2 \\
\leq & C_1 (\|h_{\delta}\|_e^2 + \|h_{\theta}\|_{W_{\mu}}^2 ),     
\end{align*}
where the last inequality is due to the continuity condition of $g(\cdot)$ imposed in  Assumption \ref{assum:append a specific model specification}(ii). Furthermore, 
\begin{align}
&2( {\ell_{N}^{\mu}}(z,\alpha)-{\ell_{N}^{\mu}}(z, \alpha^{\mu}_N)) = 2({\ell_{N,\gamma}^{\mu}}(z,\gamma_{\alpha})-{\ell_{N,\gamma}^{\mu}}(z, \gamma_{\alpha^{\mu}_N}))\nonumber \\
=&(\bar{\rho}-\Delta_\alpha)^{\prime} \Sigma(X)(\bar{\rho}-\Delta_\alpha)-\bar{\rho}^{\prime} \Sigma(X) \bar{\rho}+{\mu_N} (\|{\theta}_{\alpha_{\mu}^N}+h_{\theta} \|_e^2- \|{\theta}_{\alpha_{\mu}^N} \|_e^2) \nonumber \\
= & (\Delta_\alpha-2\bar{\rho})^{\prime} \Sigma(X) \Delta_\alpha  +{\mu_N} (\|{\theta}_{\alpha_{\mu}^N}+h_{\theta} \|_e^2- \|{\theta}_{\alpha_{\mu}^N} \|_e^2). \label{eq:C1}
\end{align}
which then combined with the facts that $\Sigma(X)$ has bounded eigenvalues and boundedness of ${P}_{\gamma}(X,\gamma_{\alpha}) $ and $Y$ (as $Y$ is binary) implies that there exists $C_2 >0$, 
\begin{align*}
\operatorname{Var}\left[ {\ell_{N}^{\mu}}(z,\alpha)-{\ell_{N}^{\mu}}(z, \alpha^{\mu}_N) \right] \leq  C_2 \operatorname{E}\|\Delta_\alpha\|^2_e.    
\end{align*}
Thus, $\operatorname{Var}\left[ {\ell_{N}^{\mu}}(z,\alpha)-{\ell_{N}^{\mu}}(z, \alpha^{\mu}_N) \right] \lesssim (\|h_{\delta}\|_e^2 + \|h_{\theta}\|_{W_{\mu}}^2),$ and Lemma \ref{lem:lower loss function difference} implies the final conclusion as long as ${Q_{N}^{\mu}}({\alpha})-{Q_{N}^{\mu}}({\alpha^{\mu}_N}) \leq \overline{\varepsilon}^2$ with $\overline{\varepsilon}^2$ as specified in Lemma \ref{lem:lower loss function difference}.
\end{proof}
			
\begin{assumption}\label{assum:appendix norm bound}
Using notations from Assumption \ref{assum:append a specific model specification}, let 
\[
\Lambda^\mu_N(X):=\sup _{\substack{\theta_{v}: \iota^{\prime} \theta_{v}=0, \|\theta_{v}\|_{W_{\mu}}>0,\\ 
\theta_{v, r} \geq 0 \text {~whenever~} {\theta}_{\alpha^{\mu}_N,r} =0,} }\frac{\left\|G(x, \delta_{\alpha^{\mu}_N}  )^{\top} \theta_{v}\right\|_e}{\|\theta_{v}\|_{W_{\mu}}},
\]
assume that for some $c>2$, $\sup _N E\left[\left\{1+\Lambda^\mu_N(X)\right\}^c\right]<\infty.$
\end{assumption}
Assumption \ref{assum:appendix norm bound} does not require $W_{\mu}$ to be nonsingular, and it only says that directions that are small in the $\|\cdot\|_{W_{\mu}}$-(semi)norm cannot be arbitrarily large on average.
			
\begin{lemma}[Local envolope]\label{lem:append,local bound} Under Assumptions \ref{assum:append a specific model specification}-\ref{assum:appendix norm bound},  and $\mu_N\leq C <\infty$, for any $\varepsilon>0$, there exist a measurable function $U_N(\cdot)$ such that
$$
\sup_{\alpha \in \mathcal{A}_N: {Q_{N}^{\mu}}({\alpha})-{Q_{N}^{\mu}}({\alpha^{\mu}_N}) \leq \varepsilon^2} \left|{\ell_{N}^{\mu}}(z,\alpha)-{\ell_{N}^{\mu}}(z, \alpha^{\mu}_N)\right| \leq \varepsilon  U_N(z),
$$
with $\sup_N E\left[U_N(z)\right]^c \leq C$ for $c>2$.
\end{lemma}
\begin{proof}[Proof of Lemma \ref{lem:append,local bound}]
First, the boundedness of $Y, P(\cdot) $, $\Sigma(X)$ and the simplex constraint implies that there exists $B>0$ such that,
\begin{align}
\left|{\ell_{N}^{\mu}}(z,\alpha)-{\ell_{N}^{\mu}}(z, \alpha^{\mu}_N)\right|  \leq B. \label{eq:C2}
\end{align} 
Equation (\ref{eq:C1}) implies that $\left|{\ell_{N}^{\mu}}(z,\alpha)-{\ell_{N}^{\mu}}(z, \alpha^{\mu}_N)\right| \leq C(\|\Delta_\alpha\|_e + \mu_N \|h_{\theta} \|_e)$ for some $C>0$, and $$\|\Delta_\alpha\|_e\leq  \left\{L_g+\Lambda^\mu_N(X)\right\}(\left\|h_{\delta}\right\|_e+\left\|h_{\theta}\right\|_{W_{\mu}}).$$ The proof of Lemma \ref{lem:lower loss function difference} implies that if ${Q_{N}^{\mu}}({\alpha_N^\mu}+ v)-{Q_{N}^{\mu}}({\alpha_N^\mu})\leq \varepsilon^2\leq \overline{\varepsilon}^2$, then there exists $\overline{C}>0$ 
\begin{align}
\left|{\ell_{N}^{\mu}}(z,\alpha)-{\ell_{N}^{\mu}}(z, \alpha^{\mu}_N)\right| \leq  \overline{C} \epsilon \left\{1+L_g+\Lambda^\mu_N(X)\right\}. \label{eq:C3}
\end{align}
Let $U_N(z)= \overline{C} \left\{1+L_g+\Lambda^\mu_N(X)\right\} + {B}/{\overline{\varepsilon}}$, then equations (\ref{eq:C2}), (\ref{eq:C3}) together with Assumption \ref{assum:appendix norm bound} leads to the conclusion.
\end{proof}
			
\begin{lemma}[Bracketing number]\label{lem:append,Entropy} 
Under conditions of Lemmas \ref{lem:appendix local variance}-\ref{lem:append,local bound}, let $${\mathcal{G}_{N}^\mu}(\mathrm{d})=\left\{{\ell_{N}^{\mu}}(z, \alpha)-{\ell_{N}^{\mu}}(z, \alpha^{\mu}_N): {Q_{N}^{\mu}}(\alpha, z)-{Q_{N}^{\mu}}(\alpha^{\mu}_N, z)   \leq \mathrm{d}^2, \alpha \in \mathcal{A}_N\right\}.$$ Then there exists $\mathrm{d}_N \in (0, 1)$ for some $c>0$ such that for some $a,b,C>0$
$$
\mathrm{d}_N=\inf \left\{\mathrm{d}>0: \mathrm{d}^{-2} \int_{b \mathrm{d}^2}^{a \mathrm{d}} \left[ \log \mathbb{N}_{[]}(\epsilon, {\mathcal{G}_{N}^{\mu}}  (\mathrm{d}),\|\cdot \|_{\mathrm{L}_2({P})}) \right]^{1/2} d w \leq C N^{1 / 2}\right\}.
$$  
\end{lemma}
\begin{proof}[Proof of Lemma \ref{lem:append,Entropy}] 
Without losing generality, we simply consider the additional restriction $\mathrm{d} \leq \min\{\overline{\varepsilon}/a, a/b\}$ with $\overline{\varepsilon}$ specified in Lemma \ref{lem:lower loss function difference}, and denote
\begin{align*}
& \mathcal{A}_{{\mathcal{G}_{N}^\mu}(\mathrm{d})}=\left\{\alpha- \alpha^{\mu}_N: {Q_{N}^{\mu}}(\alpha)-{Q_{N}^{\mu}}(\alpha^{\mu}_N)   \leq \mathrm{d}^2, \alpha \in \mathcal{A}_N\right\},  \\
& \mathcal{A}_{{\mathcal{G}_{N}^\mu}(\mathrm{d}), h}=\left\{h= (h_{\delta}, h_{\theta}) = \gamma_{\alpha}-\gamma_{\alpha^{\mu}_N}: {Q_{N}^{\mu}}(\alpha)-{Q_{N}^{\mu}}(\alpha^{\mu}_N)   \leq \mathrm{d}^2, \alpha \in \mathcal{A}_N\right\},  \\
& \mathcal{A}_{{\mathcal{G}_{N}^\mu}(\mathrm{d}), h, \; \varrho(\cdot)}=\left\{h= (h_{\delta}, h_{\theta}) = \gamma_{\alpha}-\gamma_{\alpha^{\mu}_N}:  \varrho(h) \leq \mathrm{d}, \alpha \in \mathcal{A}_N\right\},
\end{align*}  
where $$[\varrho(h)]^2\equiv \underline{\kappa}/4 (\|h_{\delta}\|_e^2 +\|h_{\theta}\|_{W_{\mu}}^2+\mu_N \|h_{\theta}\|_{e}^2),$$ and $\underline{\kappa}$ is as specified in Lemma \ref{lem:lower loss function difference}. Thus by the proof of Lemma \ref{lem:lower loss function difference} we know that  $ \mathcal{A}_{{\mathcal{G}_{N}^\mu}(\mathrm{d}), h}\subseteq \mathcal{A}_{{\mathcal{G}_{N}^\mu}(\mathrm{d}), h, \; \varrho(\cdot)}$ and the proof of Lemma \ref{lem:appendix local variance} shows that there exists $C_1 >0$ such that 
\[ 
\|{\ell_{N}^{\mu}}(z, \alpha)-{\ell_{N}^{\mu}}(z, \alpha^{\mu}_N) \|_{\mathrm{L}_2({P})}^2
\leq  \frac{4C_1}{\underline{\kappa}} [\varrho(\gamma_{\alpha}-\gamma_{\alpha^{\mu}_N})]^2,
\]
and thus via Theorem 2.7.11 in \cite{van1996weak} we know that 
$$\mathbb{N}_{[]}(\epsilon, {\mathcal{G}_{N}^{\mu}} (\mathrm{d}),\|\cdot \|_{\mathrm{L}_2({P})}) \leq \mathbb{N}_{[]}(\epsilon, \mathcal{A}_{{\mathcal{G}_{N}^\mu}(\mathrm{d}), h, \; \varrho(\cdot)}, 2\sqrt{\frac{C_1}{\underline{\kappa}}}\varrho(\cdot)).
$$ 

In addition, let $\mathcal{A}_{\mathrm{d}, h, \; \varrho(\cdot)}\equiv \left\{h: \; h= (h_{\delta}, h_{\theta}) \in \mathbb{R}^{d_{\delta}} \times [0,1]^{R_N}, \;  \varrho(h) \leq \mathrm{d}\right\}$, then
\begin{align}
\mathbb{N}_{[]}(\epsilon, \mathcal{A}_{{\mathcal{G}_{N}^\mu}(\mathrm{d}), h, \; \varrho(\cdot)}, 2\sqrt{\frac{C_1}{\underline{\kappa}}}\varrho(\cdot) )    \leq \mathbb{N}_{[]}(\epsilon, \mathcal{A}_{\mathrm{d}, h, \; \varrho(\cdot)}, 2\sqrt{\frac{C_1}{\underline{\kappa}}}\varrho(\cdot))  \leq (\frac{3\mathrm{d}}{2\epsilon}  \sqrt{\frac{\underline{\kappa}}{C_1}} )^{d_{\delta} +R_N}, \label{eq:proof entr}
\end{align}
where the first inequality is due to the fact that $\mathcal{A}_{{\mathcal{G}_{N}^\mu}(\mathrm{d}), h, \; \varrho(\cdot)} \subseteq \mathcal{A}_{\mathrm{d}, h, \; \varrho(\cdot)} $ and the second relates to, e.g.,  Exercise 2.1.6 in \cite{van1996weak}. From (\ref{eq:proof entr}) we see that there exists $C>0$
$$\mathrm{d}_N \leq \min\{\overline{\varepsilon}/a, a/b, C\sqrt{\frac{(d_{\delta} +R_N) \log (N)}{N}}\} .
$$
Then once note that for sufficiently small $d \in(0, \bar{d})$, the class ${\mathcal{G}_{N}^\mu}(d)$ contains two functions separated by at least $\min\{c_0,a/2\}\mathrm{d}$ in $\mathrm{L}_2(P)$, and thus $\mathrm{d}^{-2} \int_{b \mathrm{d}^2}^{a \mathrm{d}} \left[ \log \mathbb{N}_{[]}(\epsilon, {\mathcal{G}_{N}^{\mu}}(\mathrm{d}),\|\cdot \|_{\mathrm{L}_2({P})}) \right]^{1/2}d\epsilon >CN^{-1/2}$ for sufficiently small $\mathrm{d}$.
\end{proof}

\subsubsection{Intermediate results related to the bootstrap} 
			
Let $$\ell^{0,B}_N(V_n,{\alpha}) \equiv \omega_{n, N} \rho(z_n, \alpha)^{\prime}  {\Sigma}(X_n)^{-1}	\rho(z_n, \alpha)/2,$$ and define accordingly the corresponding penalized object $\ell^{\mu,B}_{N}(V_n,{\alpha}), D{\ell_{N}^{\mu,B}}$. 
          
\begin{assumption}[Bootstrap]\label{assum:CP15 boot3.ii}  
~\\
(i) For $u\in \mathcal{V}^\mu_{N},$ there exists a deterministic $\rho_N^B(u)> 0$ such that for $\varepsilon=o_p\bigl(\min\{N^{-1/2},\rho_N^B(u)\}\bigr)$ and for all $\delta>0$, $P(P^*({\widehat{\alpha}_N^{\mu,B}} \pm |\varepsilon| u\in\calN_N^\mu)\geq 1-\delta)\rightarrow 1,$ as $N$ increases. ~\\
Let 
$$\bar{E}_N^B\{g(V)\}\equiv \frac{1}{N}\sum_{i=1}^N\left[ g(V_i)-E_{P^*} [ g(V_i)]\right],
$$
and suppose that for some positive sequences $e_N=o\bigl(\min\{N^{-1/2},\rho_N^B(u_N^\tau)\}\bigr),$ the followings hold:~\\
(ii.1)
\begin{align*}
\sup_{\alpha \in \calN_N^\mu}\bar{E}_N^B\Bigl\{& {\ell_{N}^{\mu,B}}(V,\alpha \pm e_N u_N^\tau)-\ell^{\mu,B}_{N}(V,\alpha) \\
&\quad-D{\ell_{N}^{\mu,B}}(V,{\widehat{\alpha}_N^\mu})[\pm e_N u_N^\tau]\Bigr\}=O_{p^*}(e_N^2) \text{~~wpa1~~} (P). 
\end{align*} 
(ii.3)
\begin{align*}
\sup_{\alpha \in \calN_N^\mu}\Biggl|E_{P^*}\Bigl[& {\ell_{N}^{\mu,B}}(V_i,\alpha)-{\ell_{N}^{\mu,B}}(V_i,\alpha \pm e_N u_N^\tau)\Bigr] \\
&\quad -\frac{\|\alpha \pm e_N u_N^\tau-{\widehat{\alpha}_N^\mu}\|^2-\|\alpha-{\widehat{\alpha}_N^\mu}\|^2}{2}\Biggr|=O_{p^*}(e_N^2) \text{~~wpa1~~} (P). 
\end{align*} 
(iii) 
\begin{align*}
\left|\mathcal{L}^*\left\{\sqrt{N} \bar{E}_N^B\{D{\ell_{N}^{\mu,B}}(V,{\widehat{\alpha}_N^\mu})[v_{N}^\tau]\} \right\}-\mathcal{L}\left\{ \sigma_N \mathbb{Z} \} \right\}\right|=o_{P}(1),
\end{align*}
where $\mathbb{Z} \sim N(0,1)$.
\end{assumption} 
			
\begin{theorem}\label{theo:bootstrap}
Let Assumptions \ref{assum:ri}--\ref{assum:clt v2}, \ref{assum:regular functional} and \ref{assum:CP15 boot3.ii} hold, then 
$$
\begin{aligned}
\sup _{q \in \mathbb{R}} \mid & P^*( \sqrt{N}(\tau({\widehat{\alpha}_N^{\mu,B}})-\tau({\widehat{\alpha}_N^\mu})) \leq q )  \\& \qquad \qquad -P(\sqrt{N}(\tau({\widehat{\alpha}_N^\mu})-\tau(\alpha^{\mu}_N)) \leq q) \mid = o_{p^*}(1), \text {wpa1~~} (P).
\end{aligned}
$$
\end{theorem}
\begin{proof}[Proof of Theorem \ref{theo:bootstrap}]\label{proof of Theorem {theo:bootstrap}}
Without losing generality, given Assumption \ref{assum:CP15 boot3.ii}(i), we may assume that ${\widehat{\alpha}_N^{\mu,B}} \pm |\varepsilon| u\in\calN_N^\mu$ and following similar arguments for Corollary \ref{theo:3.1 v3}, we have that 
\begin{align}
&  \left|\sqrt{N}\left\langle{\widehat{\alpha}_N^{\mu,B}} -{\widehat{\alpha}_N^\mu}, u_N^\tau\right\rangle - \sqrt{N} \bar{E}_N^B\left\{D{\ell_{N}^{\mu,B}}(V, {\widehat{\alpha}_N^\mu})\left[u_N^\tau\right]\right\}\right|=o_{p^*}(1) \text{~~wpa1~~} (P), \label{eq:theo2,1} \\
& \frac{\tau(\widehat{\alpha}^{\mu,B}_{N})-\tau({\widehat{\alpha}_N^\mu})}{\|v^\tau_N \|_{s d,\mu}} - \left\langle{\widehat{\alpha}_N^{\mu,B}} -{\widehat{\alpha}_N^\mu}, u_N^\tau\right\rangle=o_p(N^{-1/2}).\label{eq:theo2,2}
\end{align}
Equations (\ref{eq:theo2,1})(\ref{eq:theo2,2}) combined with  (\ref{eq:theo1,1})(\ref{eq:theo1,2}),  and Assumptions \ref{assum:clt v2},  \ref{assum:regular functional} and \ref{assum:CP15 boot3.ii}(iii) then lead to the conclusion.
\end{proof} 
			
\end{appendices}

		\end{document}